\documentclass[prx,amsmath,amssymb,floatfix, twocolumn, superscriptaddress, notitlepage]{revtex4-2} 

\usepackage{graphicx}
\usepackage{float}
\usepackage{amssymb,amsthm,amsfonts,amstext,amsmath}
\usepackage{url}
\usepackage[colorlinks=true,citecolor=cyan,urlcolor=magenta]{hyperref}

\usepackage{xcolor}
\usepackage{tcolorbox}
\usepackage{enumitem}

\usepackage{cleveref}
\crefname{equation}{Eq.}{Eqs.}
\crefname{section}{Sec.}{Sections}
\crefname{figure}{Fig.}{Figs.}

\def\be{\begin{equation}}
\def\ee{\end{equation}}
\def\bea{\begin{eqnarray}}
\def\eea{\end{eqnarray}}
\def\bma{\begin{mathletters}}
\def\ema{\end{mathletters}}

\def\q0{\underline{0}}

\def\C{{\mathbb C}}
\def\id{{\mathbb I}}

\def\W{{\cal W}}

\def\M{{\cal M}}

\def\R{\mathbb{R}}

\def\tr{\mbox{tr}}
\def\W{{\cal W}}
\def\one{\leavevmode\hbox{\small1\normalsize\kern-.33em1}}

\def\bra#1{\langle#1|} \def\ket#1{|#1\rangle}

\def\proj#1{\ket{#1}\!\bra{#1}}

\newcommand{\new}[1]{{\color{black} #1}}

\newtheorem{theo}{Theorem}

\newtheorem{lemma}[theo]{Lemma}

\def\id{{\mathbb I}}

\begin{document}

\title{Optimization of time-ordered processes in the finite and asymptotic regime}
\author{Mirjam Weilenmann}
\affiliation{Département de Physique Appliquée, Université de Genève, Genève, Switzerland}
\affiliation{Institute for Quantum Optics and Quantum Information--IQOQI Vienna, Austrian Academy of Sciences, Boltzmanngasse 3, 1090 Vienna, Austria}
\author{Costantino Budroni}
\affiliation{Department of Physics ``E. Fermi'' University of Pisa, Largo B. Pontecorvo 3, 56127 Pisa, Italy}
\author{Miguel Navascués}
\affiliation{Institute for Quantum Optics and Quantum Information--IQOQI Vienna, Austrian Academy of Sciences, Boltzmanngasse 3, 1090 Vienna, Austria}

\begin{abstract}
Many problems in quantum information theory can be formulated as optimizations over the sequential outcomes of dynamical systems subject to unpredictable external influences. Such problems include many-body entanglement detection through adaptive measurements, computing the maximum average score of a preparation game over a continuous set of target states and limiting the behavior of a (quantum) finite-state automaton. In this work, we introduce tractable relaxations of this class of optimization problems. To illustrate their performance, we use them to: (a) compute the probability that a finite-state automaton outputs a given sequence of bits; (b) develop a new many-body entanglement detection protocol; (c) let the computer \emph{invent} an adaptive protocol for magic state detection. As we further show, the maximum score of a sequential problem in the limit of infinitely many time steps is in general incomputable. Nonetheless, we provide general heuristics to bound this quantity and show that they provide useful estimates in relevant scenarios.
 
\end{abstract}

\maketitle

\section{Introduction}

In quantum physics, we are usually concerned with manipulating quantum systems, whether we interact with them to cause them to reach (or remain in) a specific quantum state or measure them in order to observe or certify some of their properties. Finding new, better ways to perform tasks with time-ordered operations lies at the heart of current research in quantum information science. Quantum computing relies on a sequential application of gates to a quantum system and communication protocols aided by quantum systems depend on performing operations on these systems and exchanging information in a specific order. Sequences of operations are also crucial in less evident situations, e.g., specific sequences of pulses  are applied in order to read out a superconducting qubit~\cite{superconducting}. 

Traditionally, the problem of finding new ways (or sequences of operations) to complete a specific task has been dependent on the ingenuity of scientists to come up with new protocols. More recently, this type of problem is also tackled via machine-learning techniques, which require training a neural network to generate the right sequence of operations. However, a general method to optimize over sequential strategies or even a criterion to decide whether one such sequential protocol is optimal for a certain task is lacking. 

Part of the difficulty in solving these problems lies in the fact that, when operations are performed sequentially, protocols can be made adaptive, i.e.\ future operations may depend on the result of previous ones. Adaptiveness can lead to great advantages; e.g., it minimizes statistical errors in entanglement certification~\cite{prep_games} and quantum state detection~\cite{Knill}. However, it also increases the variety of protocols that must be considered for a specific task, and thus the complexity of optimizing over them. Indeed, while optimization techniques are well established in the context of single-shot or independently repeated procedures (e.g.\ in applications of nonlocality~\cite{npa,npa2} or for witnessing entanglement~\cite{Moroder}), developing them for sequential protocols remains a challenge.

In this work, we make progress on this problem. Specifically, we introduce a general technique to analyze and optimize a class of time-ordered processes that we call \emph{sequential models}. These are protocols in $N$ rounds, where in each round an interaction occurs and changes the state of a system of interest 
possibly depending on unknown or uncontrolled variables. The progress toward a certain goal is quantified by a reward that is generated at each round of the protocol and is the object of our optimization procedure.

Our first result is a technique that allows us to upper bound the maximum total reward that a sequential model can generate in $N$ rounds through a complete hierarchy of increasingly difficult convex-optimization problems. To illustrate the performance of the hierarchy, we apply it to upper bound the type-I error of an 
adaptive protocol for many-body entanglement detection~\cite{Saggio2019, prep_games}. The technique generates seemingly tight bounds for large system sizes (of order $N\approx 100$).

The same technique also allows us to deepen the study of temporal correlations generated by {\it finite-state automata} \cite{PazBook, Rabin1963}. These models appear in several problems in quantum foundations and quantum information theory such as classical simulations of quantum contextuality \cite{KleinmannNJP2011,Fagundes2017,BudroniNJP2019,Context_review}, quantum simulations of classical stochastic processes \cite{GarnerNJP2017,Elliott2018,Elliott2019}, purity certification \cite{SpeePRA2019}, dimension witnesses \cite{HoffmannNJP2018, SpeeNJP2020, SpeeNJP2020_a}, quantum advantages in the design of time-keeping devices \cite{ErkerPRX2017,Woods2022,tick_sequence_paper,Vieira_Budroni}, and classical simulation of quantum computation ~\cite{Zurel:2020PRL}; for more details, see a recent review on temporal correlations  \cite{Vitagliano2022}.
Specifically, we focus on the problem of optimal clocks, namely, how good a clock can we construct, given access to a classical automaton with $d$ internal states? Using our tools, we solve this problem for low values of $d$, by providing upper bounds on the maximum probability that the automaton outputs the `one-tick sequence' $000...01$ investigated in \cite{tick_sequence_paper,Vieira_Budroni}. These bounds match the best known lower bounds up to numerical precision.

In addition to finite sequences of operations, we are also often interested in letting a protocol or process run indefinitely, i.e., we are interested in characterizing its asymptotic behavior as $N \rightarrow \infty$. This is, e.g., important for systems that are left to evolve for a long time or in cases in which we aim to probe large systems, where this limit is a good approximation. 

To our knowledge, not much is known about time-ordered processes in the limit of infinitely many rounds (except for results on asymptotic rates in hypothesis testing and when taking specific limits~\cite{asymptotics1, asymptotics2}). In this paper we show that there is a fundamental reason for this: there are sequential models for which no algorithm can approximate their asymptotic behavior. This follows from undecidability results related to finite-state automata; specifically, a construction from Refs.~\cite{gimbert2010,undecidability}. Nevertheless, we develop a heuristic method for computing rigorous bounds on the asymptotic behavior. We further find that in the applications we consider, these bounds are close to the expected asymptotic behavior and to the lower bounds we can compute, thus substantiating the usefulness of our method. More specifically, we use the asymptotic method to bound the type-I error of the many-body entanglement-detection protocol mentioned above in the limit of infinitely many particles. The result is a bound that is, at most, at a distance of $4 \times 10^{-4}$ from the actual figure. Our asymptotic method also bounds the probability that a two-state automaton outputs the one-tick sequence in the limit of many time steps to $O(10^{-2})$.

Finally, in some circumstances it becomes necessary not to analyze but to find sequential models with a good performance. That is, given a number of parameters that we can control -- the policy -- which might affect the evolution of the system as well as the reward, we pose the problem of deciding which policy maximizes the total reward after $N$ rounds. We propose to tackle this problem via projected-gradient methods, and show how to cast the computation of the gradient of our upper bounds as a tractable convex-optimization problem. Using this approach, we let the computer discover a two-state, six-round preparation game to detect one-qubit magic states.

Our paper is structured as follows. In Sec.~\ref{sec:examples_intro}, we present three problems in quantum information theory that we will use to illustrate our methods throughout. In Sec.~\ref{sec:sequential}, we introduce the abstract notion of sequential models and show how to optimize them, before applying the optimization method in Sec.~\ref{sec:application_finite}. In Sec.~\ref{sec:asymptotics}, we prove that the asymptotic behavior of sequential models cannot be approximated in general. We also introduce a heuristic to tackle asymptotic problems, which we apply in Sec.~\ref{sec:asymptotics_examples}. In section \ref{sec:policies}, we pose the problem of policy optimization and show how to compute the gradient of the upper bounds derived in section \ref{sec:sequential}, over the parameters of the policy. We use gradient methods to solve a 
policy optimization problem in section \ref{sec:magic}. Finally, we present our conclusions in Sec.~\ref{sec:conclusion}.

\section{Problems with time-ordered operations}
\label{sec:examples_intro}

A variety of different problems involving quantum systems are of a sequential nature. In the following we introduce three settings that are made up of time-ordered operations and that all pose challenging open problems to the quantum information community.

\subsection{Temporal correlations} \label{sec:automaton_intro}

A {\it finite-state automaton} (FSA) is a mathematical structure that models an autonomous computational device with bounded memory. More formally \cite{PazBook, Rabin1963}, a $d$-state automaton is a device described, at every time, by an internal state $\sigma\in\Sigma$, with $|\Sigma|=d$. At regular intervals, rounds or time steps, the automaton updates its internal state and generates an outcome $b\in B$, depending on both its prior internal state $\sigma$ and the state $y\in Y$ of its input port. Such a double state-and-output transition is governed by the  transition matrix of the automaton $P(\sigma',b|\sigma, y)$, which indicates the probability that the automaton transitions to the internal state $\sigma'$ and outputs $b$, given that its input and prior states were, respectively, $y$ and $\sigma$.
Overall, the probability of outputs ${\bf b} = (b_1,\ldots,b_n)$ given inputs ${\bf y} = (y_1,\ldots,y_n)$ can be computed as
\begin{align}
p({\bf b}|{\bf y})= \! \! \! \! \! \! \sum_{\sigma_0,\sigma_1,\ldots,\sigma_n} \! \! \! \! \! \! &p(\sigma_0)P(\sigma_1, b_1|\sigma_0,y_1) \nonumber \\
& \cdot P(\sigma_2, b_2|\sigma_1,y_2)\ldots P(\sigma_n, b_n|\sigma_{n-1}, y_n)  \label{eq:FSA_prob}
\end{align} 
Expressions such as Eq.~\eqref{eq:FSA_prob}, i.e., generated by a classical FSA, appear in several 
problems related to classical simulations of temporal quantum correlations, such as, e.g.,  classical 
simulations of quantum contextuality \cite{KleinmannNJP2011,Fagundes2017,BudroniNJP2019}, quantum 
simulations of classical stochastic processes \cite{GarnerNJP2017,Elliott2018,Elliott2019}, classical 
simulation of quantum computation ~\cite{Zurel:2020PRL}, and quantum advantages in the design of 
time-keeping devices \cite{ErkerPRX2017,Woods2022,tick_sequence_paper,Vieira_Budroni}. 

We are interested in scenarios in which the performance of a $d$-state automaton after $n$ time steps is evaluated by yet another automaton, or, more generally, by a time-dependent process with bounded memory. We wish to bound said performance over all automata with $d$ states. 

This class of problems includes, e.g., computing the maximum probability that a $d$-state automaton generates the so-called `one-tick sequence' 
$00\ldots 01$ (the name comes from its interpretation as a clock signal ``1'' for a tick of the 
clock). This optimization problem is central in the discussion of the optimal classical and quantum models able to generate a time signal, with the consequent quantum advantages in the design of physical clocks. Unfortunately, there are no general optimization methods to upper bound the performance of a classical model and thus demonstrate a quantum advantage. The problem of optimal quantum clocks has been extensively investigated both from a theoretical (quantum information and quantum thermodynamics) perspective \cite{ErkerPRX2017,Woods2022,tick_sequence_paper,Vieira_Budroni, Woods2019, Yang2019, Yang2020, Schwarzhans2021, MeierPRL2023, Silva2023}, as well as from an experimental one \cite{ClockExp2021}.

Consider thus an FSA $A$, with states denoted by $\sigma\in\Sigma$, with $|\Sigma|=d$, inputs $y\in Y$ and outputs $b\in B$. We assume that the automaton is described by the (unknown) transition matrix $P(\sigma_{k+1},b_k|\sigma_k, y_k)$. 
Let $Q$ be a known time-dependent process, with an internal state denoted by $t$, and inputs (outputs) denoted by $z$ ($c$). The (known) transition matrix of this process at time steps $k=2,...,N$ is $Q_k(t_{k},c_k |t_{k-1}, z_{k-1})$. Each state $t_k$ of the process is associated with a reward $\alpha_k(t_k)$. At time step $1$, the state and output $t_1, c_{1}$ are generated by the distribution $Q_1(t_1, c_{1})$. 
Now, let us couple $Q$ with $A$, in the following way (see also Fig.~\ref{fig:coupling} for an illustration): after $Q$ generates $t_1, c_{1}$, we input $c_{1}$ in the automaton $A$, i.e., we let $y_1=c_{1}$. The automaton then transitions through $P$ to a state $\sigma_1$, outputting $b_1$. This output is further input into the process $Q$ (by letting $z_1=b_1$), which in turn generates $t_2,c_2$ through $Q_2$, and so on. At time step $N$ we consider the sum of all rewards,
\be
\sum_{k=1}^N\alpha_k(t_k).
\ee
\noindent Our goal is to maximize the expectation value of the total reward over all $d$-state automata $A$.
\begin{figure}
  \centering
  \includegraphics[width=9cm]{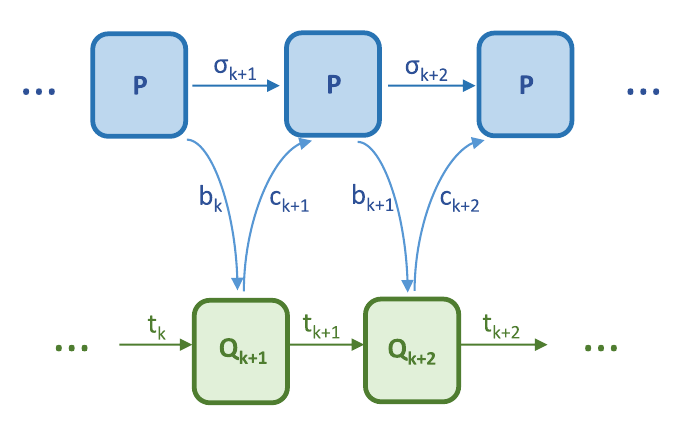}
  \caption{\textbf{An Automaton (blue) coupled to a process $Q$ (green).} The process may in general differ from round to round. }
    \label{fig:coupling}
\end{figure}
Note that by convexity, the best strategy is for the automaton to be initialized in a specific state $\sigma_0=\bar{\sigma}_0$ rather than a distribution thereof \footnote{It doesn't matter which of the initial states is chosen as these choices are equivalent up to relabeling.}. Hence, fixing the initial state, the optimal automaton is fully characterized in terms of the transition matrix $P$ that maximizes the expected reward.

This problem can be mathematically rephrased in a way that will prove useful later. Namely, the coupled systems $A$ and $Q$ can be regarded as a single dynamical system the internal state $s_k$ of which at time step $k$ corresponds to the probability distribution of the triple $(t_k, c_{k}, \sigma_k)$, i.e., $s_k:= p_k(t_k, c_{k}, \sigma_k)$. The evolution of this system from one time step to the next follows the equation of motion 
\begin{align}
p_{k+1}(t_{k+1}, c_{k+1}, \sigma_{k+1}) \! & =\! \! \! \! \! \! \! \! \! \! \sum_{c_{k}, t_k, \sigma_k,b_k} \! \! \! \! \! \!  \! \! p_k(t_k, c_{k}, \sigma_k)  P(\sigma_{k+1}, b_{k}|\sigma_k,c_{k})  \nonumber \\ & \qquad \qquad \quad \cdot Q_{k+1}(t_{k+1}, c_{k+1}|t_k,b_{k}) \nonumber \\
&=:f_k(s_k, \lambda). \label{eq:eom_automaton} 
\end{align}
\noindent Note that the evolution is given here by a polynomial $f_k$ of the state $s_k$ and an unknown evolution parameter $\lambda:=(P(\sigma',c|\sigma, z):z\in Z, c\in C,\sigma\in \Sigma)$. In this formulation of the problem, the reward $r_k$ at time step $k$ is given by
\be
r_k(s_k):=\sum_{t_k,c_{k},\sigma_k}p_k(t_k,c_{k},\sigma_k)\alpha(t_k). 
\ee
Note that the reward does not depend  on $\lambda$ here (except via the $p_k$): such a dependency could, however, be introduced by using different reward functions $r_k(s_k, \lambda)$.
In this picture, our task is to maximize the total reward over all possible evolution parameters $\lambda$.

\subsection{Entanglement certification of many-body systems}

Given an $N$-partite quantum system, we are often faced with the problem of determining whether its state is entangled. In principle, this problem can be solved by finding an appropriate entanglement witness and estimating its value on the state at hand. 
However, in addition to the problem that finding such a witness theoretically is NP-hard~\cite{Gharibian}, its estimation may involve conducting joint measurements on many subsystems of the $N$-partite state, a feat that may not be experimentally possible. Replacing such joint measurements by estimates based on local measurement statistics requires, in general, $\mbox{Exp}[N]$ state preparations, thus rendering the corresponding protocols infeasible for moderately sized $N$ (as is the case when the protocol demands full quantum state tomography). 

One way to avoid these problems, i.e.\ perform single-system measurements and reduce the number of state preparations, is to carry out an adaptive protocol~\cite{prep_games}. Such a scheme proceeds as follows. We sequentially measure the particles that constitute the $N$-particle ensemble, making future measurements depend on previous ones as well as on previous outcomes. Once we conduct the last measurement, we make a guess on whether the underlying state was entangled or not.

To choose how to measure each system and to make our final guess, we use a dynamical system $Q$. This process has internal states $t_k\in T$, and its transition matrix at time $k=1,...,N-1$ is of the form $Q_k(t_{k+1}, y_{k+1}|t_k,b_k)$, where $b_k\in B$ denotes the measurement outcome at time $k$ and $y_{k+1}$ labels the measurements to be conducted on the $(k+1)${th} particle. For each $y\in Y_k$, there exists a positive operator-valued measure (POVM) $\{M^k_{b|y}\}_b$. The final guess is generated through the function $Q_N(t_{N+1}|t_N,b_N)$, with $t_{N+1}\in\{\mbox{`entangled', `separable'}\}$. The initial state of the process and the measurement setting $y_1$ for the first particle are generated by the distribution $Q_0(t_1, y_1)$. This is illustrated in Fig.~\ref{fig:GHZsequential}.
\begin{figure}
  \centering
  \includegraphics[width=\columnwidth]{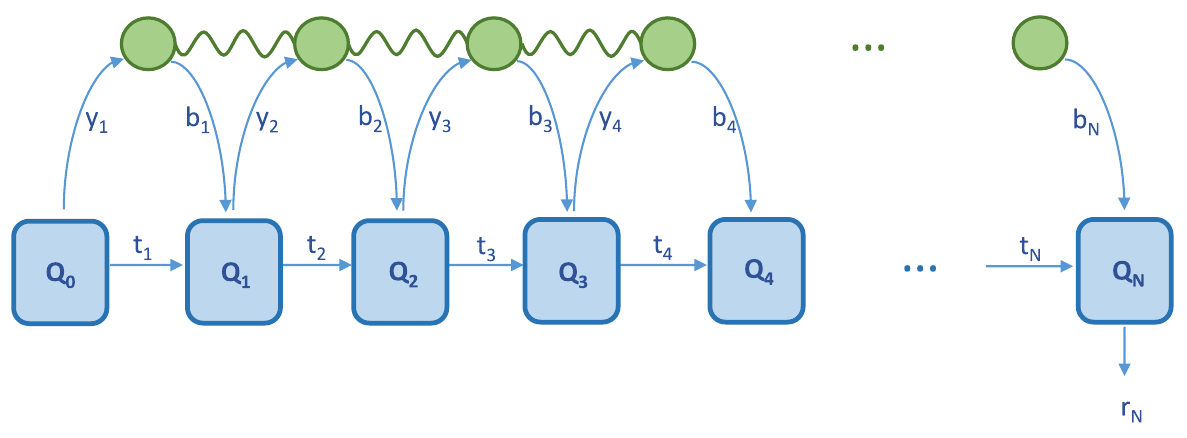}
  \caption{\textbf{The sequential model for entanglement detection.} The particles of an $N$-party state (green) are measured one by one. The $k$th measurement $y_k$ depends on the internal state of a finite-state automaton and the outcome $b_k$ affects the transition of the automaton according to $Q_{k}$. 
   }
  \label{fig:GHZsequential}
\end{figure}

If the target state $\rho$ the entanglement of which we wish to detect experimentally admits an efficient matrix product operator (MPO) decomposition \cite{MPO}, then one can efficiently compute the probability that the process outputs the result `separable', see Fig.~\ref{fig:contraction}. This probability is usually called a \emph{type-II} error, or a \emph{false negative}. 

\begin{figure}
  \centering
  \includegraphics[width=\columnwidth]{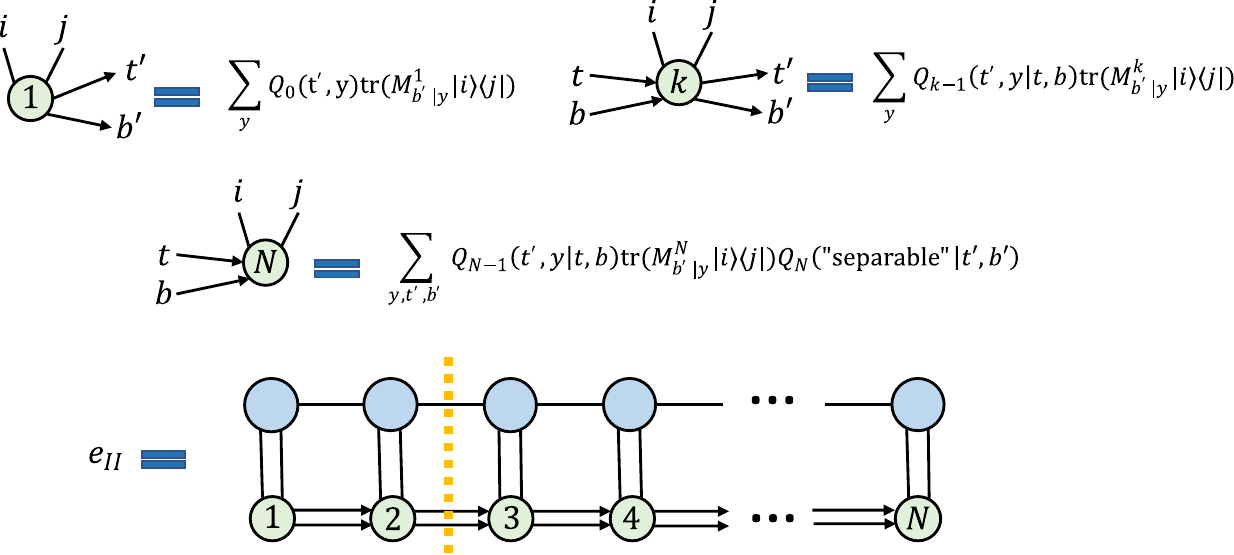}
  \caption{\textbf{Efficient computation of the type-II error.} The blue circles represent the tensors making up our MPO target state, with two vertical legs denoting the bra and ket indices of each particle and two horizontal legs of dimension $D$ --the bond dimension-- to account for the  correlations of target state. Defining the green tensors in the way indicated in the figure, their contraction with the blue tensors equals the type-II error of the corresponding entanglement-detection protocol. This contraction can be computed by multiplying matrices of dimension $|T||B|D$, as the dashed orange line indicates.
  }
  \label{fig:contraction}
\end{figure}

It remains to compute the \emph{type-I} error, or \emph{false positive} of the entanglement-detection protocol encoded in $Q$: that would correspond to the maximum probability that the process outputs `entangled' when the input is a fully separable quantum state. In computing this maximum, we can assume, by convexity, that the player has prepared a pure product quantum state. That is, at each time step $k$, we measure some pure quantum state $\rho_k$. Despite this simplification, computing the type-I error of the scheme requires maximizing a multilinear function with a size-$O(N)$ input. As $N$ grows, this becomes a more difficult problem.

Like the problem of computing the maximum performance over all $d$-state automata, maximizing the type-I error over the set of separable states can be phrased as an optimization over the sequential outputs of a deterministic dynamical system. In this case, the  internal state $s_k$ of the dynamical system at time step $k$ would correspond to the distribution $P_k$ of $t_k,y_k$. The  equation of motion of the system is
\begin{align}
P_{k+1}(t_{k+1},y_{k+1})= \! \! \! \sum_{t_k,y_k,b_k}P_k &(t_k,y_k) \tr(M^k_{b_k|y_k}\rho_k) \nonumber \\ & \cdot Q_k(t_{k+1},y_{k+1}|t_k, b_k),
\label{eq:eom_ghz_like}
\end{align}
\noindent where the quantum state $\rho_k$ can be interpreted as an uncontrolled external variable $h$ influencing the evolution of the system. At time step $N$, the system outputs the deterministic outcome
\begin{align}
    r_{N}(s_{N})=\! \! \! \! \!\sum_{t_N,y_N,b_N} \! \! \! \! \! P_N&(t_N,y_N)\tr(M^N_{b_N|y_N}\rho_N) \nonumber \\
    &\cdot Q_N(t_{N+1}|t_N, b_N)\delta_{t_{N+1}, \mbox{`entangled'}}.
\end{align}

\subsection{Independent identically distributed (IID) preparation strategies in  quantum preparation games}
\label{sec:iid_strategies}
The notion of preparation games aims to capture the structure of general (possibly adaptive) protocols, where quantum systems are probed sequentially \cite{prep_games}. More precisely, a \emph{quantum preparation game} is an $N$-round task involving a player and a referee. In each round, the player prepares a quantum state, which is then probed by the referee. In round $k$, before carrying out his measurement, the referee's current knowledge of the source used by the player is encoded in a variable $t_k\in T_k$, with $|T_k|<\infty$, called the \emph{game configuration}. This variable depends 
nontrivially on the past history of measurements and measurement results: in each round, it guides the referee in deciding which measurement to perform next and changes depending on its outcome. This double role of the game configuration can be encoded in the POVM used by the referee for that round. That is, assuming that the game configuration before the measurement is $t\in T_k$, the probability that the new game configuration is $t'\in T_{k+1}$ is given by $\tr(\rho_kM^k_{t'|t})$, where $\rho_k$ is the player's $k${th} state preparation; and $\{M^k_{t'|t}:t'\in T_{k+1}\}$, the POVM implemented by the referee when the game configuration is $t$. At the end of the preparation game, a score $g(t_{N+1})$, where $t_{N+1}$ denotes the final game configuration, is generated by the referee according to some scoring system $g(t)$ for $t\in T_{N+1}$.

In some circumstances -- e.g., in an entanglement detection protocol where the player tries to (honestly) prepare a specific entangled state -- it makes sense to consider a player that follows an independent identically distributed (IID) strategy, meaning that the player produces the same state $\rho$ in each round. 

Consider thus the problem of computing the maximum game score, for a fixed scoring system, achievable with preparation strategies of the form $\rho^{\otimes N}$, with $\rho\in C$, where $C$ is some set of states.
In Ref.~\cite{prep_games}, it has been shown how to compute the maximum score of a preparation game under IID strategies when the set $C$ of feasible preparations has finite cardinality. In this work, we consider the case in which $C$ is continuous.

This problem can be modeled through a dynamical system internal state $s_k$ of which at time $k$ corresponds to $P_k(t_k)$, the distribution of game configurations at the beginning of round $k$. The equation of motion of the system is
\begin{equation}
P_{k+1}(t')=\sum_{t\in T_k}P_k(t)\tr(\rho M^k_{t'|t}),
\end{equation}
\noindent where the unknown evolution parameters $\lambda:=\rho\in C$ correspond to the player's preparation. The dynamical system outputs a reward at time $N$, namely
\begin{equation}
r_{N}(s_{N},\lambda)=\sum_{t'\in T_{N+1}}g(t')\sum_{t\in T_k}P_N(t)\tr(\rho M^N_{t'|t}).
\end{equation}
\noindent Our goal is thus to maximize the  reward of the system over all possible values of the evolution parameter $\lambda$, i.e., over all $\rho\in C$. 

As an aside we note that, in some circumstances, it becomes necessary to compute the maximum average game score achievable by an \emph{adversarial player}, who has access to the current game configuration and can thus adjust their preparation in each round accordingly. Curiously enough, this problem is easier than optimizing over IID strategies, and, in fact, a general solution is presented in Ref.~\cite{prep_games}. Notwithstanding, we show below that the computation of the maximum game score of a preparation game under adaptive strategies can also be regarded as a particular class of the sequential problems considered in this paper.

\section{Sequential models and their optimization}
\label{sec:sequential}
The three problems described above are -- even though from a physical point of view rather different  -- structurally very similar. Indeed, they are all examples of the type of problem that we introduce more abstractly in the following. 
Consider a scenario in which the state of a dynamical system is fully described at time step $k$ by a variable $s_k\in S_k$. Between time steps $k$ and $k+1$, the system transitions from $s_k$ to another state $s_{k+1}$ depending on $s_k$, a number of uncontrolled variables $h_k\in H_k$ and a number of unknown evolution parameters $\lambda\in \Lambda$. The equation of motion that describes this transition is
\begin{equation}
    s_{k+1}=f_k(s_k, h_k, \lambda),
    \label{motion_eq}
\end{equation}
\noindent where $S_k$, $H_k \ \forall k $ and $\Lambda$ are sets of parameter values.  

Our goal is to optimize over this type of system, according to some figure of merit that is given by the problem at hand. We model this by assuming that, in each time step, the system emits a reward (or penalty) $r_k(s_k, h_k, \lambda)$. Note that this reward may be the zero function for certain $k$, so that this includes problems where only the final outcome matters as a special case. We shall refer to such a system from now on as a \emph{sequential model} (see 
Fig.~\ref{fig:model}). 
\begin{figure}[t]
  \centering
  \includegraphics[width=9cm]{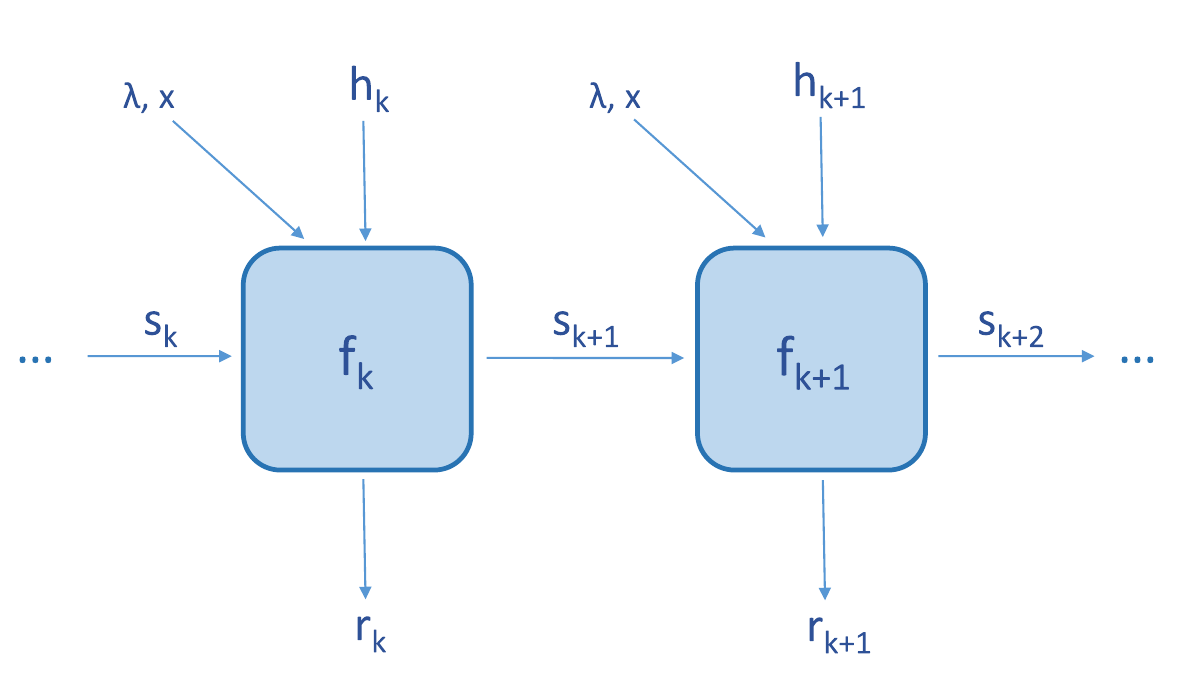}
  \caption{\textbf{A pictorial representation of a sequential model.} Note that the equation of motion itself may capture several interactions or be composed of several steps (e.g.\ the measurement of a quantum system followed by the update of a finite-state automaton). }
  \label{fig:model}
\end{figure}

 Our goal is to estimate the maximum total reward $\nu^\star$ over all possible values of $\lambda, h_1,...,h_N$. That is, we wish to solve the problem
\begin{align}
    \nu^\star \ := \ \max_{\lambda, h} \quad &\sum_{k=1}^N r_k(s_k,h_k,\lambda)\nonumber\\
    \mbox{such that } \quad &\lambda\in \Lambda, \ h \in \times_{k=1}^N H_k \label{main_problem}\\
    &s_{k+1}=f_k(s_k, h_k, \lambda) \ \ \forall  k. \nonumber
\end{align}
\noindent We generally assume that the initial state $s_1$ is known; otherwise, we can absorb it into the definition of $\lambda$ in the sense that we add one time step that generates the initial state. 
In the absence of $\lambda$, the above is mathematically equivalent to a deterministic Markov decision problem, if we regard $h$ as an action \cite{markov_book}.

\subsection{Reformulation of the optimization over sequential models}

To compute the optimal value $\nu^\star$, we consider an approach based on dynamical programming methods \cite{markov_book}. It consists in defining \emph{value functions} $V_k(s_k,\lambda)$, which represent the maximum reward achievable between time steps $k$ and $N$ over all possible values of $h_k,...,h_N$, starting from the state $s_k$, and under the assumption that the evolution parameters take the value $\lambda$. Approximating these functions with polynomials and invoking known characterizations of positive polynomials on a semialgebraic set, we manage to derive monotone, converging sequences of tractable upper bounds on $\nu^\star$. \new{We remark that similar ideas have been explored in the context of continuous-time optimal-control problems, which can be formulated as an infinite-dimensional linear program, for which semidefinite-program (SDP) relaxations based on polynomial optimization techniques exist; see Ref.~\cite{LasserreOCP2008} and references therein. 
}

In terms of the reward functions, value functions are given by
\begin{align}
    V_{N}(s_N,\lambda)&=\max_{h_N\in H_N}r_N(s_N,h_N,\lambda),\nonumber\\
    V_k(s_k,\lambda) &= \max_{h_k\in H_k}r_k(s_k,h_k,\lambda)+ V_{k+1}(f_k(s_k, h, \lambda),\lambda)
    \label{bellman}
\end{align}
\noindent and where we recover $\nu^\star$ in the final step as
\begin{equation}
    \nu^\star=\max_{\lambda\in\Lambda} V_1(s_1,\lambda).
    \label{nu_star}
\end{equation}

\begin{widetext}
\noindent The problem of finding $\nu^\star$ can thus be reduced to 
\begin{align}
    \min_{V_1,...,V_N,\nu} \ \ &\nu \nonumber\\
    \mbox{such that } \quad   &V_{N}(s_N,\lambda) \geq r_N(s_N,h,\lambda) \quad \forall h\in H_N, \ s_N\in S_N, \ \lambda\in\Lambda \nonumber\\
     &V_k(s_k,\lambda) \geq r_k(s_k,h,\lambda)+ V_{k+1}(f_k(s_k, h, \lambda),\lambda) \quad \forall h\in H_k, \ s_k\in S_k, \ \lambda\in\Lambda \nonumber\\
    &\nu \geq V_1(s_1,\lambda) \quad \forall \lambda\in\Lambda. \label{dual} 
\end{align}
\end{widetext}

\noindent Indeed, on one hand, any feasible point of the problem given in Eq.~\eqref{dual} provides an upper bound on $\nu^\star$. On the other hand, the $V_1,...,V_N,\nu^\star$, as defined by Eqs.\ \eqref{bellman}, \eqref{nu_star}, are also a feasible point of \eqref{dual}. Hence the solution of the problem given in Eq.~\eqref{dual} is $\nu^\star$.

Unfortunately, unless their domain is finite, there is no general method for optimization over arbitrary functions. What is feasible is to solve problem given in Eq.~\eqref{dual} under the assumption that $\{V_k(s_k,\lambda)\}_k$ belong to a class of functions ${\cal F}$ described by finitely many parameters. Since, in general, the optimal functions $V_k(s_k,\lambda)$ might not belong to ${\cal F}$, the result $\bar{\nu}$ of such an optimization is an upper bound on the actual solution $\nu^\star$ of the problem. The class ${\cal F}$ of polynomial functions is very handy, as it is closed under composition and dense in the set of continuous functions with respect to the uniform norm. This is the class we are working on in this paper. 

For the rest of the paper, we thus  assume that, for all $k \in \{1,2, \ldots, N \}$, $f_k(s_k,h_k,\lambda), r_k(s_k,h_k,\lambda)$ are polynomials on $s_k,h_k,\lambda$ and that the sets $S, H_k$ and $\Lambda$ are bounded sets in some metric space that can be described by a finite number of polynomial inequalities. Note that, under these assumptions, the value functions are continuous on $s_k$ and $\lambda$. 

Call, then, $\nu^n\geq \nu^\star$ the result of Eq.~(\ref{dual}) under the constraint that $\{V_k(s_k,\lambda)\}_k$ are polynomials of a given degree $n$. Since any continuous function defined on a bounded set can be arbitrarily well approximated by polynomials, it follows that $\lim_{n\to\infty} {\nu}^n=\nu^\star$.

There is the added difficulty that solving (\ref{dual}) entails enforcing positivity constraints of the form 
\begin{equation}
p(x)\geq 0,\forall x\in X,
\end{equation}
where $p(x)$ is a polynomial on the vector variable $x\in \R^m$ and $X$ is a bounded region of $\R^m$ defined by a finite number of polynomial inequalities (i.e., a basic semialgebraic set). Fortunately, there exist several complete (infinite) hierarchies of sufficient criteria to prove the positivity of a polynomial on a semialgebraic set \cite{krivine, stengle, schmuedgen, putinar}. In algebraic geometry, any such hierarchy of criteria is called a \emph{positivstellensatz}. 

Two prominent positivstellens\"atze are Schm\"udgen's \cite{schmuedgen} and Putinar's \cite{putinar}. Both hierarchies have the peculiarity that, for a fixed index $k$, the set of all polynomials satisfying the $k${th} criterion is representable through a  SDP. The application of Schm\"udgen's and Putinar's positivstellens\"atze for polynomial optimization thus leads to two complete hierarchies of SDP relaxations. The SDP relaxation based on Putinar's positivstellensatz is known as the Lasserre-Parrilo hierarchy \cite{lasserre, parrilo}. The reader can find a description of both hierarchies in App.~\ref{app:polynomials}. While the hierarchy based on Schm\"udgen's positivstellensatz is more computationally demanding than the Lasserre-Parrilo hierarchy, it converges faster.  In our numerical calculations below, we use a hybrid of the two. 

The use of the $k^{th}$ (sufficient) positivity test of a given positivstellensatz while optimizing over polynomial value functions of degree $n$ results in an upper bound $\nu^{n}_k$ on $\nu^n$, with $\lim_{k\to\infty}\nu^n_k=\nu^n$. For our purposes, the take-home message is that $\nu^n_k\geq \nu^\star$, for all $k,n$ and $\lim_{n\to\infty}\lim_{k\to\infty}\nu^n_k=\nu^\star$.

We finish this section by noting that there is nothing fundamental about the use of polynomials to approximate value functions. In fact, there exist other complete families of functions that one might utilize to solve the problem given in Eq.~\eqref{dual}. A promising such class is the set of \emph{signomials}, or linear combinations of exponentials, for which there also exists a number of positivsellens\"atze in the literature \cite{signom_pos1,signom_pos2}. In this case, the corresponding sets of positive polynomials are not SDP representable  but are nonetheless tractable convex sets. As such, they are suitable for optimization.

\subsection{Quantum preparation games as sequential models}
\label{sec:prep_games}

Sequential models capture the structure of various types of time-ordered processes, including the three introduced in Sec.~\ref{sec:examples_intro}. Here, we argue that sequential models are actually rather general: optimizations over adversarial strategies in quantum preparation games can be modeled as sequential problems, too. In fact, their resolution by means of value functions allows one to rederive previous results on the topic~\cite{prep_games}.

Let ${\cal S}$ be a set of quantum states and consider a preparation game where the player is allowed to prepare, at each round $k$, any state $\rho(t_k)\in {\cal S}$ depending on the current game configuration $t_k$. This preparation game can be interpreted as sequential model with $\Lambda=\emptyset$, $s_k=p_k(t_k)$, $h_k=(\rho^k(t):t\in T_k)$ and equation of motion
\begin{equation}
p_{k+1}(t')=\sum_{t\in T_k}p_k(t)\tr(\rho^k(t)M^k_{t'|t}).
\end{equation}

We consider the problem of maximizing the average game score (or reward)
\begin{equation}
r_{N+1}(p_{N+1})=\sum_{t\in T_{N+1}}p_{N+1}(t)g(t),
\end{equation}
\noindent where $g$ is the score function of the game  \cite{prep_games}.

In principle, we could formulate this problem as in Eq.~(\ref{dual}). Note that $V_{N+1}=r_{N+1}$ is linear in $p_{N+1}$. This motivates us to consider an ansatz of linear value functions $V_k(p_{k})$ for this problem. 
That is, we aim to solve the problem given in Eq.~(\ref{dual}) under the assumption that
\begin{equation}
V_k(p_{k})=\sum_{t\in T_k}\mu_k(t)p_k(t),
\end{equation}
\noindent for some vector $(\mu_k(t):t\in T_k)$. The condition ${V_{k}(p_k)\geq V_{k+1}(p_{k+1})}$ then translates to
\begin{align}
\sum_{t\in T_k} \! \! \mu_k(t)p_k(t) \! &\geq  \! \max_{ \{\rho^k(t) \}_t} \! \sum_{t'\in T_{k+1}} \! \! \! \! \mu_{k+1}(t')  \! \! \sum_{t\in T_k}  \! \! p_k(t)\tr(\rho^k(t)M^k_{t'|t})\nonumber\\
&=\sum_{t\in T_k}p_k(t)\max_{\rho^k(t)}\tr \! \! \left(\! \!\rho^k(t)\! \!\sum_{t'\in T_{k+1}}\! \!\! \!V_{k+1}(t')M^k_{t'|t}\! \! \right)\! .
\end{align}
Setting $p_k(t)=\delta_{t,j}$, for $j\in T_k$, we arrive at the recursion relations:
\begin{align}
&\mu_{N+1}(t)=g(t),\nonumber\\
&\mu_k(t)=\max_{ \rho^k(t)}\tr\left(\rho^k(t)\sum_{t'\in T_{k+1}}\mu_{k+1}(t')M^k_{t'|t}\right),
\label{recur_prep}
\end{align}
which allow us to compute $V_1(t_0)$. This quantity is, in principle, an upper bound on the actual solution of the problem in Eq.~(\ref{dual}), because we have enforced the extra constraint that the value functions are linear. The upper bound is, however, tight, as can be seen by computing the average score of the preparation game  under the preparation strategy given by the maximizers $\bar{\rho}^k(t)$ of the problem in Eq.~(\ref{recur_prep}). In fact, all the above is but a complicated way to arrive at the recursion relations provided in Ref.~\cite{prep_games} to compute the maximum score of a preparation game.

\section{Application to optimizing time-ordered processes}
\label{sec:application_finite}

All the sequential problems considered in this work could, in principle, be solved (or at least approximated) through the straightforward application of the Lasserre-Parrilo hierarchy~\cite{lasserre, parrilo}. However, the degree of the polynomials in such problems grows linearly with the number of steps $N$, thus making the corresponding SDP intractable for more than a few iterations.  Our reformulation of the sequential optimization problem as in (\ref{dual}), on the other hand, allows us to circumvent this issue and optimize sequential models for much higher values of $N$.

\subsection{Optimizations over finite-state automata: Probability bounds} \label{sec:automaton}

In this section, we investigate the maximum probability with which a classical FSA can generate a given output sequence. This problem has been extensively investigated for input-output sequences  \cite{BudroniNJP2019}, as well as sequences with only inputs \cite{tick_sequence_paper,Vieira_Budroni}. In particular, we consider the so-called one-tick sequence, $00\ldots01$, introduced in Ref.~\cite{tick_sequence_paper} in connection with the problem of optimal clocks. For a sequence of total length $L$, we use the compact notation $0^{L-1}1$. Here, we use the general techniques developed in 
Sec.~\ref{sec:sequential} to tackle this problem. Note that this is an instance of the type of problem introduced in Sec.~\ref{sec:automaton_intro}. 

Consider a $d$-state automaton without inputs and with binary outputs, i.e., $Y=\emptyset$, $B=\{0,1\}$. We wish to compute the maximum probability that the automaton outputs the one-tick sequence $0^{L-1}1$, denoted $P_{\rm max}^{\rm d}(L)$. For this purpose, let us construct a simple process $Q$ that, coupled to the automaton, produces an expected reward that equals $P(0^{L-1}1)$. $Q$ has two internal states, i.e., $t_k\in\{0,1\}$ and no outputs ($C=\emptyset$). Its initial state is $t_1=1$, and its transition matrices are as follows: 
\begin{align}   Q_k(t_{k+1}|t_k,b_k)&=\delta_{t_{k+1},t_k(1- b_k)},\mbox{ for } k=1,...,L-1,\nonumber\\
Q_{L}(t_{L+1}|t_{L},b_{L})&=\delta_{t_{L+1},t_Lb_L}.
\end{align}
\noindent For time steps $k=1,...,L-1$, the process $Q$ therefore keeps being in state $1$ as long as the automaton outputs $0$'s; if the automaton outputs any $1$, the state of the process changes to $0$ and stays there until the end of the game. At time step $L$, process $Q$ receives the output $b_L$. If $b_L=0$ or $t_{L}=0$, then $t_{L+1}=0$; otherwise, $t_{L+1}=1$. The probability that the automaton produces the sequence $0^{L-1}1$ thus corresponds to the probability that $t_{L+1}=1$. This corresponds to a reward function 
\begin{align}
    \alpha_k(t_k)&=0, \mbox{ for }k=1,...,L-1, \nonumber\\
    \alpha_{L+1}(t_{L+1})&=\delta_{t_{L+1},1}.
\end{align}

Using the techniques developed in Sec.~\ref{sec:sequential}, we map this problem to a polynomial optimization problem. Due to the choice of reward (together with the process $Q$), we see that the only terms from the equation of motion given in Eq.~\eqref{eq:eom_automaton} that contribute to the final reward are the ones with the correct $b_k$. This allows a further simplification of the model. 

First, let us redefine the internal states of the sequential model to be the vector
\begin{equation}
s_k=(p_k(t_k=1,\sigma_{k+1}))_{\sigma},
\end{equation}
\noindent described by $d$ parameters $s_k^\sigma$, with constraints $s_k^\sigma \geq 0$ for $\sigma=0,\ldots, d-1$, and $1-\sum_\sigma s_k^\sigma\geq 0$. 
Next, we redefine the unknown evolution parameter $\lambda$ to have dimension $d^2$, through: $\lambda:=(P(\sigma',b=0|\sigma):\sigma,\sigma'=0,1)$, with constraints
\begin{align}
P(\sigma',b=0|\sigma) &\geq 0, \mbox{ for all } \sigma,\sigma'=0,\ldots,d-1\nonumber\\
1-\sum_{\sigma'}P(\sigma',b=0|\sigma) &\geq 0, \mbox{ for }\sigma=0,\ldots, d-1.
\end{align}

For $k=1,...,L-1$, the equation of motion of this sequential model is

\begin{equation}
p_{k+1}(t_k=1,\sigma_{k+1})=\sum_{\sigma_k}p_{k}(t_{k-1}=1,\sigma_{k})P(\sigma_{k+1},0|\sigma_k),
\end{equation}
while the reward function reads
\begin{align}
    r_k &=0, \mbox{ for }k=1,...,L-1;\nonumber\\
    r_{L}(s_{L},\lambda) &=\sum_{\sigma}p_L(t_L=1,\sigma)\left(1-\sum_{\sigma'}P(\sigma',b=1|\sigma)\right).
\end{align}

Now, we apply an SDP relaxation \`a la Lasserre \cite{lasserre} to obtain an upper bound on the solution of the sequential problem for the cases of $d=2$ and $d=3$. In the latter case, no nontrivial analytical upper bound has previously been known \cite{tick_sequence_paper, Vieira_Budroni}. For details regarding the exact implementation we refer to App.~\ref{app:implementation1}.
Figures \ref{fig:d2} and \ref{fig:d3} display the upper bounds computed with this method. We find that, in every case, the obtained upper bound matches the output probability of the automaton proposed in Refs.~\cite{tick_sequence_paper, Vieira_Budroni}, up to numerical precision. This implies that the upper bounds computed are all tight. {The optimal models for these
sequences, except in the case $L = d+1$, correspond to the so-called cyclic model~\cite{tick_sequence_paper}. The model consists of deterministic transitions from one state to the other, and the
probability of emitting 1 is nonzero only in the last state. As a consequence, the bound depends only on
$\lfloor L/d\rfloor$, corresponding to the $2$- and $3$-periodic plateaus observed in Figs.~\ref{fig:d2} and \ref{fig:d3}.}

In order to compare the performance of our algorithm with those of the previously known method, namely, the Lasserre-Parrilo hierarchy ~\cite{lasserre, parrilo}, we also perform a numerical optimization of our problem with the latter method. While our method has been able to reach $L=50$ for $d=2$ and $L=10$ for $d=3$, the Lasserre-Parrilo hierarchy could not go beyond the case $L=7$ for $d=2$. This boost in performance was to be expected as  the degree of the polynomials in the Lasserre-Parrilo hierarchy grows linearly with the number of steps $L$, whereas in our method it remains constant. For more details, see Apps.~\ref{app:polynomials} and \ref{app:implementation0}.

\begin{figure}
  \centering
  \includegraphics[width=\columnwidth]{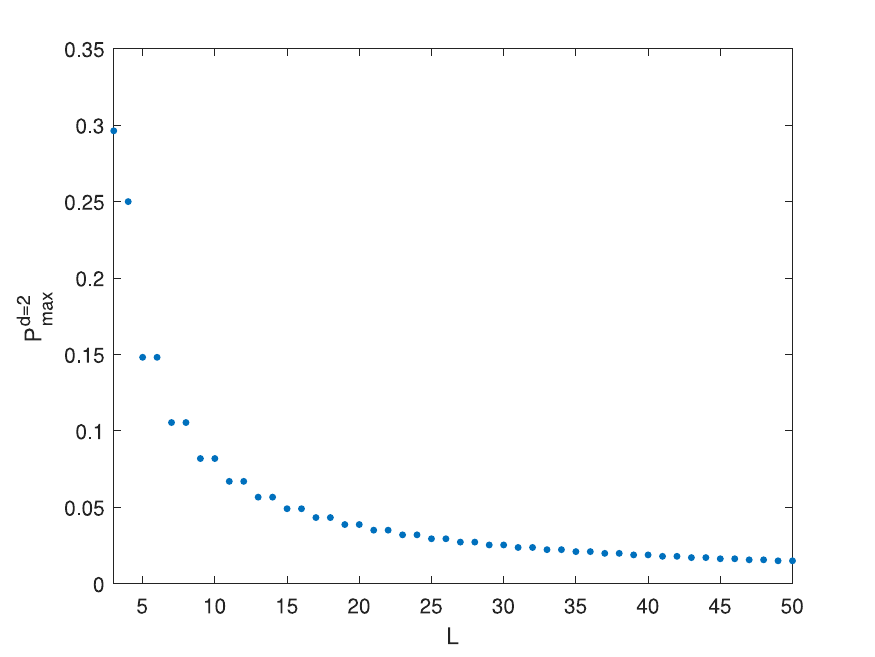}
  \caption{\textbf{Optimization over two-state automata.} The upper bounds (red) follow a pattern, where from $L=5$ on pairs of consecutive values coincide. Our bounds furthermore match the explicit automata found in Ref.~\cite{tick_sequence_paper} up to numerical precision, thus providing a tight bound and proving the optimality of the models found in Ref.~\cite{tick_sequence_paper, Vieira_Budroni}. The optimization has been performed with MOSEK \cite{mosek} and CVXPY \cite{cvxpy}. Each solution provided by the solver has been tested to verify that linear and positivity constraints are satisfied up to numerical precision.}
  \label{fig:d2}
\end{figure}

\begin{figure}
  \centering
  \includegraphics[width=\columnwidth]{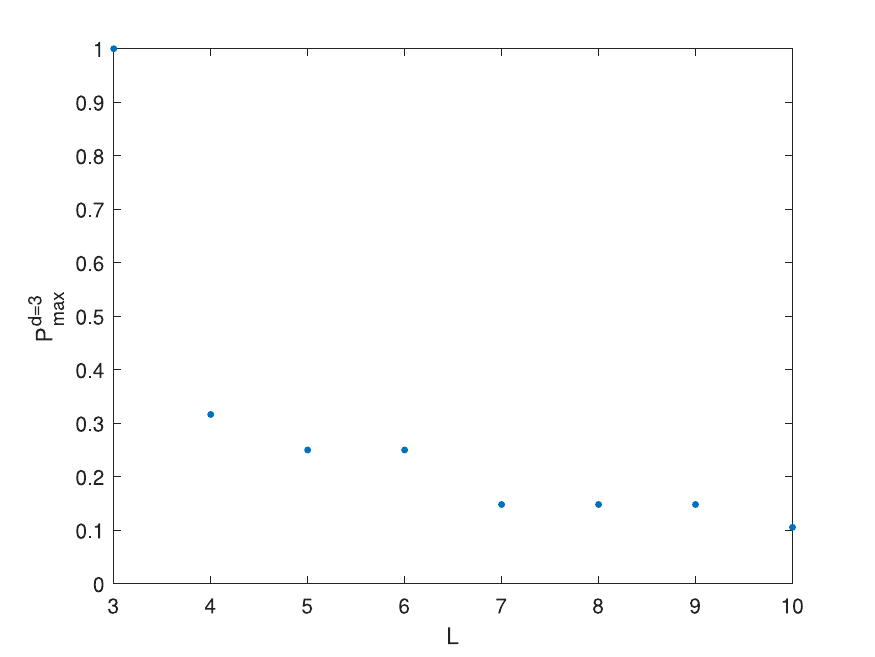}
  \caption{\textbf{Optimization over three-state automata.} Our method allows us also to optimize over three-state automata. Again, we recover up to numerical precision the bound computed with the explicit model found in Refs.~\cite{tick_sequence_paper, Vieira_Budroni}, proving that our bound is tight and that the model previously found was optimal. The optimization has been performed with SCS~\cite{scs} and CVXPY \cite{cvxpy}. Each solution provided by the solver has been tested to verify that linear and positivity constraints are satisfied up to numerical precision.}
  \label{fig:d3}
\end{figure}

In addition to the numerical values for $P^{\rm d}_{\rm max}$, we can also extract the value functions from our optimizations and use those to prove our bounds analytically. After deriving those value functions, we can further use them to find optimal automata for the problem at hand. 

For extracting value functions, let us note first that from the numerical solutions the coefficients of the polynomials $V_k$ can be directly extracted, as these are  optimization variables of the problem. However, the value functions may not be unique and extracting polynomials with suitable (rational or integer) coefficients from the numerics is more challenging. Ways to further simplify these value functions are elaborated in App.~\ref{app:value_fct}

The method for optimizing sequential models employed here allows us to derive upper bounds in terms of value functions, without treating the actual variables of the model -- in this case the $\{p_k\}_k$ and $P$ -- explicitly \footnote{Recall that the variables of the polynomial optimization are the coefficients of the $V_k$ rather than the variables of the sequential model $s_k$.}. \emph{A priori}, the solution hence does not tell us what the values of the optimal dynamical systems achieving the bounds may be, even if the bounds are tight. However, this can be remedied by first extracting the value functions from the optimization and then running a second optimization, this time fixing the value functions and optimizing for the variables of the sequential model, here $\{p_k\}_k$ and $P$. More precisely, if we found an optimal upper bound $\nu^\star$, then all value function inequalities have to be tight and there is an optimal dynamical system that achieves the bound. Now if we can find a system that satisfies all inequalities with equality, then this will automatically be optimal. 
To find such a model, we can thus consider the optimization problem 
\begin{align}
    \max_{s_k, \lambda, h_k} &V_1(s_1,\lambda) \nonumber\\
    \mbox{such that }    &V_{N}(s_N,\lambda) = r_N(s_N,h_N,\lambda) \nonumber\\
     &V_k(s_k,\lambda) = r_k(s_k,h_k,\lambda) \! + \! V_{k+1}(f_k(s_k, h_k, \lambda),\lambda)   \nonumber \\
     &h_k\in H_k \ \forall k, \ s_k\in S_k \ \forall k, \  \lambda\in\Lambda 
\end{align}
This is again a polynomial optimization problem that can be solved with an SDP relaxation. In the following examples we use the SOLVEMOMENT functionality of YALMIP~\cite{yalmip} at level $4$ and the solver MOSEK~\cite{mosek} to solve this.\footnote{Note that SOLVEMOMENT sometimes allows optimal solutions to the problem to be directly extracted, as is the case in the examples below. This is not surprising, as in cases where the Lasserre hierarchy has converged at a given finite level, the optimal values of the polynomial variables can be found in the moment matrix.} 
We find that in these cases, the problem is feasible and we extract optimal automata.

\begin{itemize}[leftmargin=0.25cm]
\item In the case $d=2$ and $L=3$ above, we obtain 
\begin{align}
V_1(s_1, \lambda) =(P(0,0|0) q_1(0) + P(0,0|1) q_1(1)) P(b=1|0) \nonumber\\
+(P(1,0|0)q_1(0) + P(1,0|1) q_1(1)) P(b=1|1),
\end{align}
where $q_k(\sigma):=p_k(1,\sigma)$. Note that in this case this is the natural optimization problem that one would also formulate directly to solve the problem at hand, however, for higher $L$ the decomposition into value functions differs. 
Now, using that the final reward is 
\begin{equation}
r_2=P(b=1|0)q_2(0)+ P(b=1|1)q_2(1),
\end{equation}
and the state-transition rules that translate to 
\begin{align}
    q_2(0)&= P(0,0|0) q_1(0)+P(0,0|1) q_1(1), \\
    q_2(1)&= P(1,0|0) q_1(0) + P(1,0|1) q_1(1),
\end{align}
we immediately see that 
$V_1(s_1,\lambda)-r_2(s_2(s_1),\lambda) = 0$. Furthermore, we can derive upper bounds on 
$V_1(s_1(s_0),\lambda)$, where we have from the state-transition rules $q_1(0)=P(0,0|0)$ and $q_1(1)=P(1,0|0)$ (as the starting state is  $q_0(1)=1$, $q_0(2)=0$) that
\begin{align}
    V_1(s_1(s_0),\lambda)\! = \! (P(1,0|1)^2 \! \!+ \! \!P(2,0|1) P(1,0|2)) P(b=1|1) \nonumber \\ +(P(1,0|1) P(2,0|1) + P(2,0|1) P(2,0|2)) P(b=1|2).
\end{align}
This gives a polynomial upper bound on the solution that indeed is upper bounded by $0.29630$, as a numerical maximization of the value function over the $P$, $q_k$ confirms. The optimization also shows that this maximum is achieved by an automaton with $P(0,0|0)=P(1,0|1)=\frac{1}{3}$, $P(1,0|0)=\frac{2}{3}$, $P(0,0|1)=0$, thus confirming this value as the optimal tight upper bound.\footnote{Deriving this upper bound analytically, is straightforward if we require in addition that $P(b=1|0)=0$, which we can impose without changing our upper bounds. In this case the polynomials simplifies in a way that inspecting the gradient immediately leads to these bounds.} This also coincides with the optimal automaton proposed in~\cite{tick_sequence_paper}.

\item For $d=2$ and $L=4$ we obtain the two value functions\footnote{Notice that we have changed the notation for value functions from $\{V_k\}_k$ to $\{W_k\}_k$ here, to make the distinction from the $L=3$ case clear.}
\begin{align}
    W_2(s_2,\lambda) &= V_1(s_2,\lambda) \\
    W_1(s_1,\lambda) &= 0.25 (1+q_1(0)(q_1(0)-P(0,0|0)) \nonumber \\& \qquad \qquad +q_1(1)(q_1(1)-P(1,0|0)))
\end{align}
Now, we can see that 0.25 is an upper bound for $W_1(s_1,\lambda)$, as 
the state-update rules and $q_0(0)=1$ and $q_0(1)=0$ in this case imply that $q_1(0)=P(0,0|0)$ and $q_1(1)=P(1,0|0)$.
As a last relation~\footnote{That $W_2(s_2, \lambda)$ upper bounds the reward follows like in the $L=3$ case above.} we have to check that 
\begin{equation}
    W_1(s_1,\lambda)-W_2(s_2(s_1),\lambda) \geq 0.
\end{equation}
A numerical optimization confirms this polynomial inequality. Requiring that $W_1(s_1,\lambda)=W_2(s_2(s_1),\lambda)$ while maximizing $W_1(s_1,\lambda)$ under the constraints on $s_1$ and $s_2$ gives an optimal solution for the automaton of  $P(1,0|0)=1$, $P(0,0|0)=0$ and $P(0,0|1)=\frac{1}{2}$, $P(1,0|1)=0$.\footnote{Analytically, checking this is again more straightforward when assuming additionally that $P(b=1|0)=0$.} This again coincides with the automaton from~\cite{tick_sequence_paper} for $L=4$.
\end{itemize}

Note that there are usually many  different value functions that could solve the problem. Extracting value functions with few terms and integer or fractional coefficients from the optimization problem, however, requires some tuning (see App.~\ref{app:value_fct}). The process becomes more cumbersome as we increase the number of variables. We thus do not provide value functions for larger $L$ here. We remark that extracting valid value functions from the optimization is, nevertheless, always possible.

\subsection{Entanglement detection in many-body systems: the automaton-guided GHZ game} \label{sec:entanglement}

To illustrate the performance of our SDP relaxations to solve entanglement certification problems in many-body systems, {we focus our attention on the $N$-qubit Greenberger-Horne-Zeilinger (GHZ) state
\begin{equation}
   \ket{GHZ} = \frac{1}{\sqrt{2}}\left(\ket{0}^{\otimes N}+\ket{1}^{\otimes N}\right)
\end{equation}
of size up to $N=35$. This is a timely application, as GHZ states of sizes up to $20$ qubits can already be prepared~\cite{GHZ18, GHZ20cat, GHZ20catRydberg}. 
The main goal in devising such protocols is to obtain a procedure that requires as few state preparations as possible, for reliably certifying entanglement. In this section we propose such a one-shot protocol and demonstrate that following this procedure indeed needs fewer repetitions $n$ than a family of protocols based on shadow tomography~\cite{kueng_shadow}. }

The main feature of our method is that we work in a hypothesis-testing setting. This means that we are also able to prove bounds on the performance of the protocol for the worst-case separable state (corresponding to the type-I error), which is difficult to compute and generally not provided when estimating fidelities. Our method has the additional advantage of relying only on single-qubit quantum operations (as compared e.g.\ to the 18-qubit example from Ref.~\cite{GHZ18} that uses two-qubit gates not just in the preparation but also in the disentangling phase) and  it can also deal with protocols that are adaptive (see Sec.~\ref{sec:entanglement2}), (as compared to determining the expectation of an operator, e.g.\ the parity~\cite{GHZ20cat, GHZ20catRydberg}, by fixed measurements, or compared to evaluating Bell inequalities~\cite{review_detection}). This latter aspect is what generally allows us to reduce the number of repetitions of our protocols.

It is worth remarking that, in \emph{any} one-shot entanglement-detection protocol that aims to detect the pure GHZ state, the type-I and type-II errors $e_I, e_{II}$ will satisfy the relation
\begin{equation}
e_I+e_{II}\geq\frac{1}{2}.
\label{limit_GHZ}
\end{equation}
To see why, call $M$ the $N$-qubit effective POVM element (depending on the protocol, the corresponding POVM will be global, LOCC or one-way LOCC) associated with declaring the system `entangled' and let $\sigma$ be the separable state
\begin{equation}\label{state}
\sigma=\frac{1}{2}(\proj{0}^{\otimes N}+\proj{1}^{\otimes N}).
\end{equation}
Then we have that 
\begin{align} 
1-e_{II}-e_I(\sigma) &=\tr\{M(\proj{GHZ}-\sigma)\} \nonumber\\
&\leq\frac{1}{2}\left\|\proj{GHZ}-\sigma\right\|_1 \nonumber\\
&=\frac{1}{2}, \label{bound}
\end{align}
where $e_{I}(\sigma)$ is the probability of declaring $\sigma$ ``entangled''. Since $e_I\geq e_I(\sigma)$, we arrive at Eq.~\eqref{limit_GHZ}. 

The protocol that we propose in the following saturates the relation in Eq.~\eqref{limit_GHZ} in the limit $N\to\infty$. In this sense, despite solely relying on one-qubit measurements, our protocol is asymptotically optimal.

Let us, then, describe our entanglement-detection protocol. The protocol is inspired by the GHZ game, previously considered in Refs.~\cite{Mermin1990, Boyer2004}, where here the role of the memory is taken by a classical FSA. We thus call this protocol the \emph{automaton-guided GHZ game}. 
Specifically, picture a protocol for many-body entanglement verification where, for each of the first $N-1$ particles, one of the following two measurements is performed with probability $\frac{1}{2}$:
\begin{align*}
 M_1 &=\{ M_{1|1}=\ket{+}\bra{+}, \ M_{2|1}=\ket{-}\bra{-} \} \ , \\M_2 &=\{ M_{1|2}=\ket{+i}\bra{+i}, \ M_{2|2}=\ket{-i}\bra{-i} \}, 
\end{align*}
where $\ket{\pm}=\frac{1}{\sqrt{2}}(\ket{0}\pm \ket{1})$ and $\ket{\pm i}=\frac{1}{\sqrt{2}}(\ket{0}\pm i \ket{1})$. This process is used to update the state of a four-state automaton that is initialized in state $t_1=1$ as displayed in Fig.~\ref{fig:GHZautomaton}.
\begin{figure}[h]
  \centering  \includegraphics[width=\columnwidth]{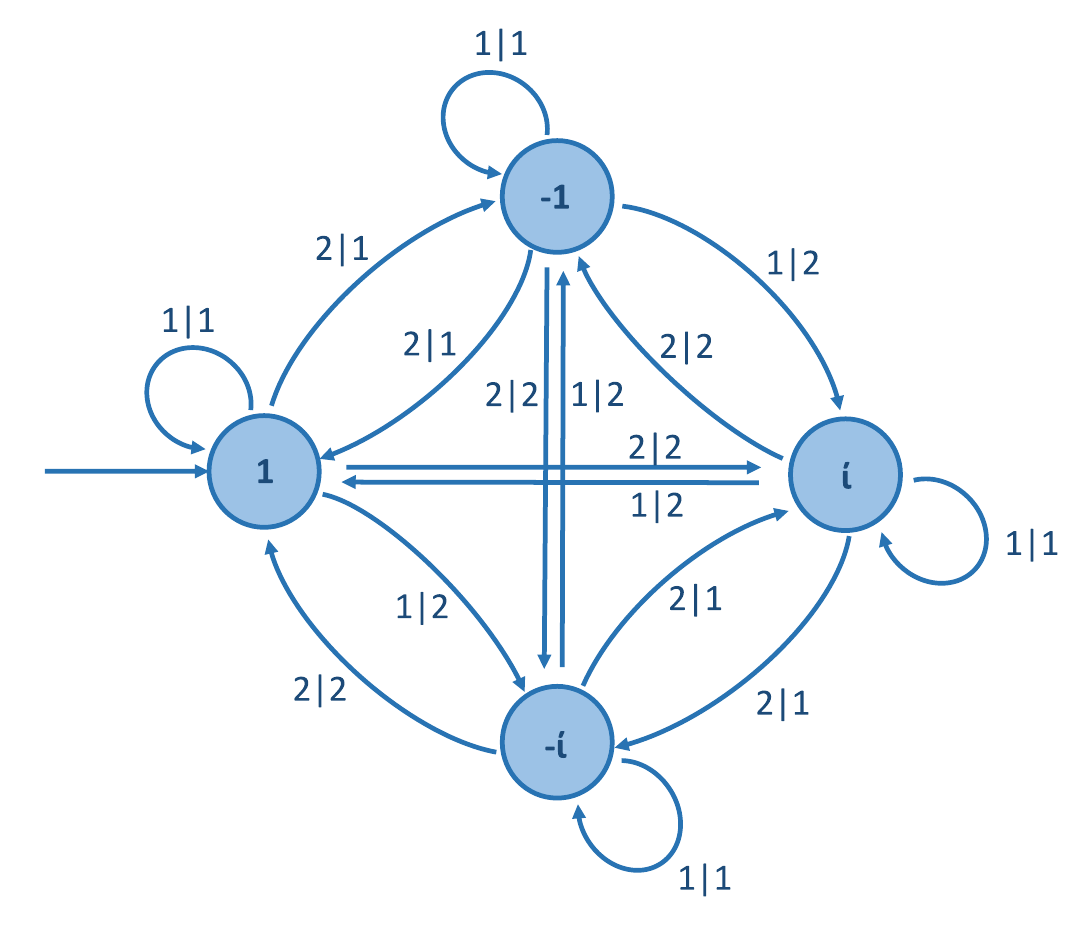}
  \caption{\textbf{The memory update of the four-state automaton in the GHZ game.} The transitions between states depend deterministically on $b$ and $y$, which are used to label the edges as $b|y$. In each round of the associated sequential model, this updates the automaton state $t \in \{1,-1,i,-i\}$. }
  \label{fig:GHZautomaton}
\end{figure}
For the final particle, we choose measurement $1$ for states $t_{N}=\pm 1$ of the automaton and measurement $2$ for $t_{N}=\pm i$. The verification of an entangled state is considered successful if the outcome, $b_{N}$, of this additional measurement is $b_{N}=1$ for $t_{N} \in \{1,i\}$ or $b_{N}=2$ for $t_{N} \in \{-1,-i\}$, 
in which case we update $t_{N+1}$ to ``entangled''; otherwise ``separable''.

This protocol succeeds with probability 1 for a GHZ state of any dimension. Indeed, think of an $k$-partite GHZ-like state of the form
\begin{equation}
\frac{1}{\sqrt{2}}\left(\ket{0}^{\otimes k}+\alpha\ket{1}^{\otimes k}\right).
\label{GHZ-like}
\end{equation}
\noindent where $\alpha\in\{1,-1,i,-i\}$. Suppose that we measure one of the qubit states in the basis $\{\ket{+},\ket{-}\}$ ($\{\ket{i},\ket{-i}\}$), and obtain the result $\ket{\pm}$ ($\ket{\pm i}$). Then, the final postselected state is an $(k-1)$-partite state of the form given in Eq.~\eqref{GHZ-like}, but this time with phase $\pm\alpha$ ($\mp i\alpha$).

Due to its transition matrix, the internal state of the automaton in Fig.~\ref{fig:GHZautomaton} keeps track of the phase of the GHZ-like state, as its qubits get sequentially measured, i.e., $t=\alpha$. 
The reward function simply verifies that the state of the last qubit  corresponds to $\frac{1}{\sqrt{2}}(\ket{0}+\alpha\ket{1})$.

The procedure introduced in Sec.~\ref{sec:sequential} allows us to derive upper bounds on the  maximum probability of success that is achievable with separable states, i.e.\ the worst-case type-I error $e_I$. 
We display our bounds for different system sizes $n$ in Fig.~\ref{fig:GHZbounds}.
\begin{figure}
  \centering \includegraphics[trim=3.6cm 9.6cm 3.6cm 9.6cm, width=\columnwidth]{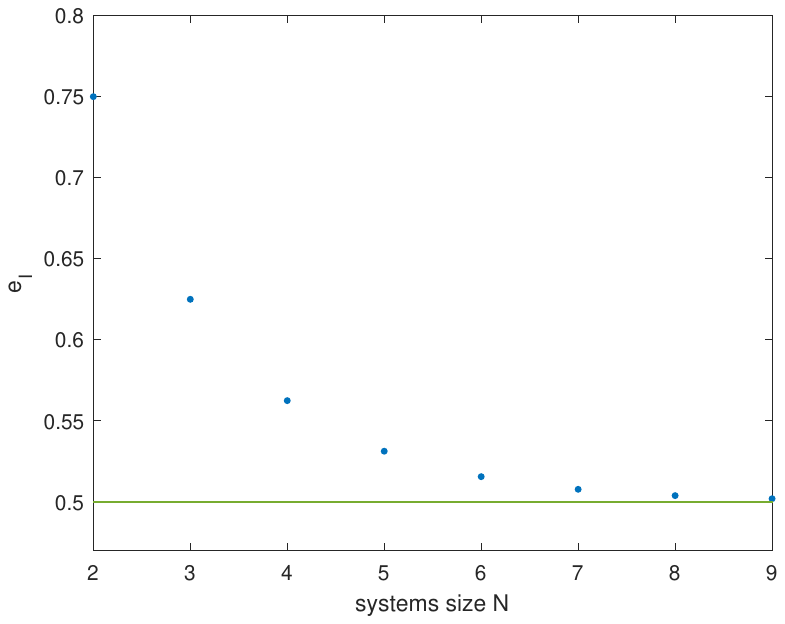}
  \caption{\textbf{The upper bounds for the winning probability of the best separable states in the GHZ game.} We observe that the winning probability (blue) drops quickly toward 0.5 (green line). The bounds are displayed up to $N=9$ for aesthetic reasons. In fact, we have checked that we can derive bounds up to $N=35$ with this method. With our asymptotic hierarchy (see Sec.~\ref{sec:asymptotics_ghz}), we further confirm an upper bound on the asymptotic value of $0.5004$ as $N \rightarrow \infty$.}
  \label{fig:GHZbounds}
\end{figure}
Note that the size of the optimization problem allows us to reach round numbers as high as $N=35$. {As the reader can appreciate, $e_I$ seems to converge to $\frac{1}{2}$ in the limit $N\to\infty$, thus saturating the inequality in Eq.~\eqref{limit_GHZ}.} Details on the implementation (including precision issues and how to resolve them in this case) are presented in App.~\ref{app:implementation2}.

We further compare the detection performance of the automaton-guided GHZ game with a family of entanglement-detection protocols based on shadow tomography~\cite{kueng_shadow}, in the following called \emph{shadow protocols}. 

There exist different protocols for shadow tomography, depending on the type of measurements conducted on the qubits. Since for the automaton-guided GHZ game we only allow the measurement of each qubit individually, we opt for shadow protocols involving random measurements in the three Pauli bases. 

Our family of shadow protocols for entanglement detection thus works as follows: 
\begin{enumerate}
\item A random measurement in one of the Pauli bases is carried out on each qubit and both the result and the basis are recorded. This procedure is repeated for $n$ experimental rounds (or state preparations). 
\item The resulting classical shadow is then used to produce an estimator $f$ for the underlying  fidelity of the quantum state with respect to the GHZ state.
\item If $f$ is greater than a threshold $0\leq \theta\leq 1$, then the protocol outputs the result `entangled'; otherwise, it outputs ``separable''.
\end{enumerate}
There is thus one shadow protocol for each value of $\theta$. 

The computation of the type-II error $e_{II}(\theta)$ for this family of protocols is straightforward-- not so the computation of $e_I(\theta)$, due to the need to optimize the probability that the protocol outputs ``entangled'' over the set of separable input states. Because of this, we replace the maximization over all separable states by a maximization over the three separable states $\ket{0}^{\otimes N}, \ket{+}^{\otimes N}$ and $\ket{+i}^{\otimes N}$. The resulting quantity $\tilde{e}_I(\theta)$ is therefore a lower bound on the actual value of $e_I(\theta)$, i.e., $e_I(\theta)\geq \tilde{e}_I(\theta)$. Since our aim is to show that the automaton-guided GHZ game outperforms the shadow protocols, underestimating the error in the latter is sufficient, as it provides us with a lower bound on the advantages that we observe for the automaton-guided GHZ game.

For different values of $\theta$, we have estimated the type-I and type-II errors $\tilde{e}_I(\theta),e_{II}(\theta)$ of the shadow protocols for the $5$-partite GHZ state in $n=50$ experimental rounds \footnote{Given the small numbers of experimental rounds, we used the mean instead of the `median of means' recommended in \cite{kueng_shadow} to estimate the GHZ fidelity in the shadow protocols.}. To compare these numbers against the performance of the one-round automaton-guided GHZ game, we extended the latter to $n$ experimental repetitions in the following way: 
\begin{enumerate}
\item Run the automaton-guided GHZ game for $n$ experimental rounds, noting down the number of times $n_e$ that the protocol outputs ``entangled''.
\item If $\frac{n_e}{n}$ exceeds the threshold $0\leq\phi\leq 1$, declare the state entangled; otherwise, separable.
\end{enumerate}

Pairs of values $(e_I, e_{II})$ for both families of protocols are plotted in Fig.~\ref{fig:shadow_tomography}. Note that the curve for the automaton-guided GHZ games lies completely below that of the shadow protocols. In fact, for some values of the threshold $\phi$ the automaton-guided GHZ game exhibits values with $e_I\approx e_{II}=0$. This was to be expected. Indeed, let $\bar{e}_{I}\approx 0.53$ be the type-I error of the one-shot GHZ game. Then, for $\phi=1$, the entanglement of the GHZ state is still detected with probability $1$ (and so $e_{II}=0$). Furthermore, a separable state can only be mislabeled as entangled with probability upper bounded by $e_I=\bar{e}_{I}^{50}\approx 0$. 

One might have expected that shadow-based protocols have the upper hand for large values of the system size $N$. Nonetheless, from Fig.~\ref{fig:GHZbounds}, it is clear that the $(e_I(\phi), e_{II}(\phi))$ curve of the automaton-guided protocol will improve with increasing $N$. This is not the case for shadow-based GHZ fidelity estimation, (see Ref.~\cite{kueng_shadow}). Thus, for larger $N$, the advantage of our automaton-guided GHZ game will be even larger. \new{Finally, we remark that we have compared our automaton-guided GHZ game with the shadow-based GHZ fidelity estimation of \cite{kueng_shadow}. However, a better estimate on the GHZ fidelity can be obtained by selecting a different sampling rule, e.g., Paulis from the stabilizer group, following the idea of Ref.~\cite{FlammiaPRL2011}.}
\begin{figure}
  \centering \includegraphics[width=\columnwidth]{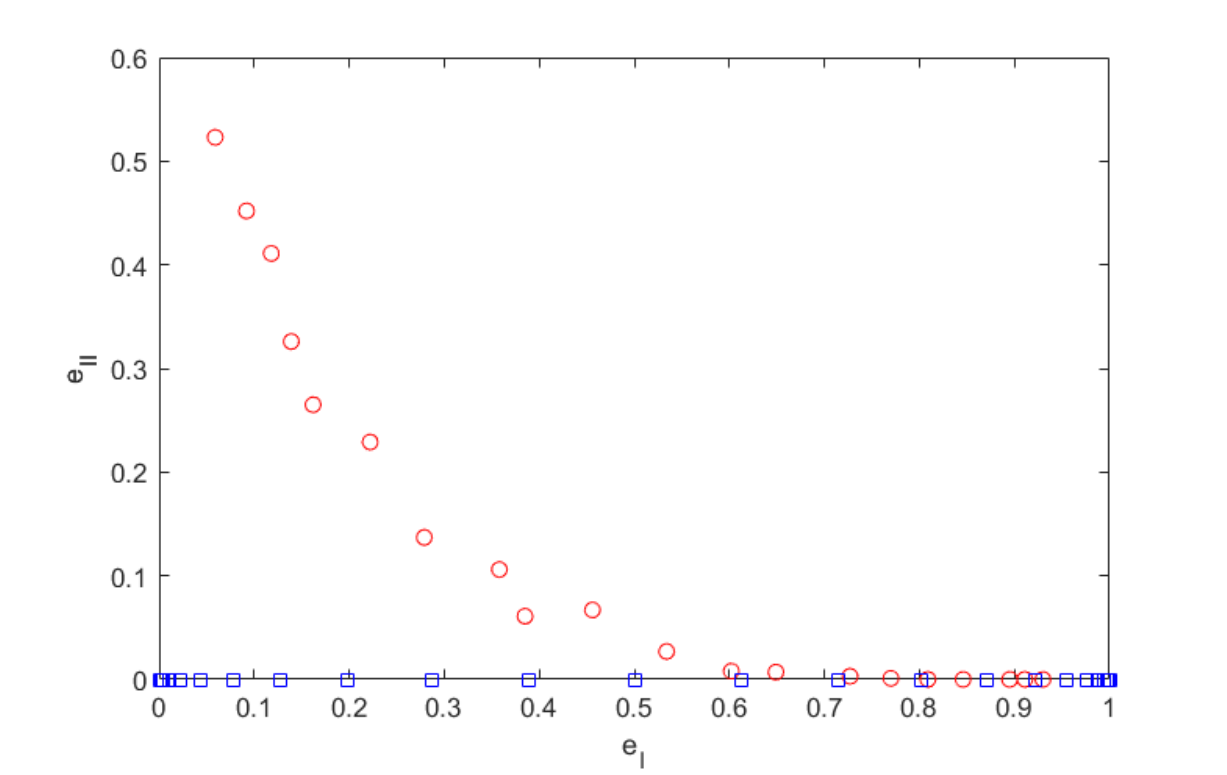}
  \caption{\textbf{Type-I errors versus type-II errors.} The blue squares denote the automaton method; the red circles, the shadow protocols. For some values of the threshold $\phi$, the automaton method is essentially perfect.}
  \label{fig:shadow_tomography}
\end{figure}

\subsection{Entanglement detection in many-body systems: The bit-sequence protocol} \label{sec:entanglement2}

To illustrate how our method allows us to optimize protocols that are genuinely adaptive, we consider here the following family of $N$-round games that we call \emph{bit-sequence protocols}.

In the first round, we perform, with equal probability, each of the following measurements and record the outcome: 
\begin{align*}
 M_1 &=\{ M_{1|1}=\ket{0}\bra{0}, \ M_{2|1}=\ket{1}\bra{1} \} \ , \\M_2 &=\{ M_{1|2}=\ket{+}\bra{+}, \ M_{2|2}=\ket{-}\bra{-} \}. 
\end{align*}
In each of the $N-1$ subsequent rounds, we then use the following procedure. In round $k$, we take the measurement outcome ($1$ or $2$) of round $k-1$ as the new measurement choice. We perform the respective measurement and record the new outcome. After the last measurement in round $N$, we check whether the final outcome is $2$, in which case we assign a reward of $1$; otherwise the reward is $0$. 

For each $N$, the bit-sequence protocol defines an $N$-qubit operator that is applied to the $N$-qubit state of which we aim to certify its entanglement. In the following we will consider an $N$-qubit state, $\rho_N$, that is an eigenstate to the maximal eigenvalue of this operator defined by the $N$-round protocol.

In Fig.~\ref{fig:bit-sequence}, we compare the score, denoted $S_{\rm max}$, that we obtain with $\rho_N$ with the maximum score that can be obtained with a separable state. The score for separable states is computed using the procedure from Sec.~\ref{sec:sequential}, see App~\ref{app:implementation3}.

In this protocol, the maximal score obtained with separable states corresponds -- as in the protocol of Sec.~\ref{sec:entanglement} -- to the maximal type-I error $e_I$, while the score for an entangled state $\rho_N$ corresponds to $1 - e_{II}(\rho_N)$.
\begin{figure}
  \centering \includegraphics[trim=3.6cm 8.7cm 2.9cm 9.6cm, width=\columnwidth]{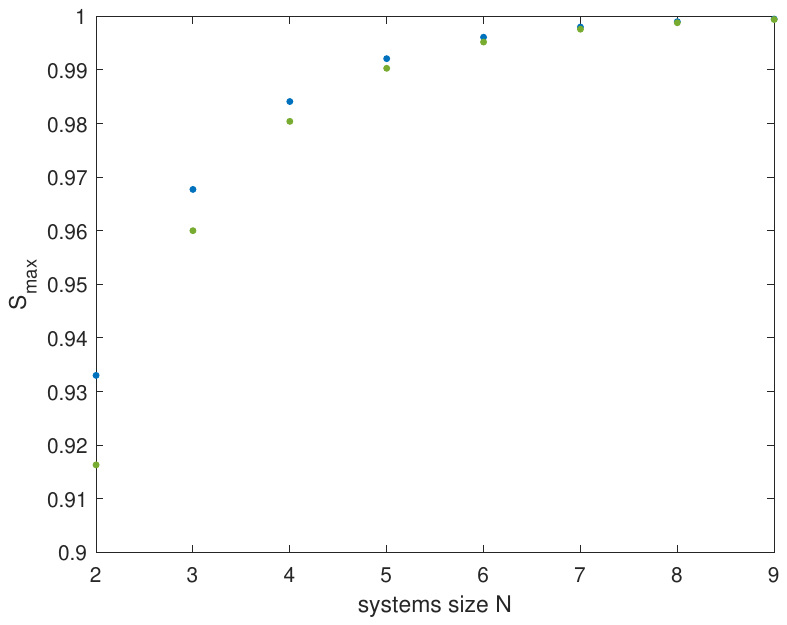}
  \caption{\textbf{Upper bounds for the winning probability of the best separable states compared to $\rho_N$ in a single shot.} The green dots give an upper bound on the score that can be obtained with separable states in the bit-sequence game. We observe that this value lies strictly below the value that we obtain with the states $\rho_N$ (blue).}
  \label{fig:bit-sequence}
\end{figure}

We observe a strict separation between the scores for $\rho_N$ and the upper bounds we obtain for separable states. This demonstrates that the bit-sequence protocols can certify the entanglement of this family of states. Our single-shot protocols can furthermore be turned into $n$-round protocols using, e.g., the ideas of meta-games from Ref.~\cite{prep_games}, in order to amplify this separation.

\section{Asymptotic behavior of sequential models} \label{sec:asymptotics}

In some scenarios, the considered sequential problem admits a natural generalization to arbitrarily many rounds or time steps $N$. In that case, we may also be interested in the asymptotic behavior of the total reward as $N \rightarrow \infty$. However, as we shall see in the following (Sec.~\ref{sec:uncomp}), there are sequential models for which approximately computing the asymptotic maximum reward (or penalty) is an undecidable problem. Thus the best we can hope for are heuristics that work well in practice. In Sec.~\ref{sec:heuristics} we introduce such heuristic methods and we demonstrate in Sec.~\ref{sec:asymptotics_examples} that these methods lead to useful bounds for the problems that we have considered in the finite regime in Secs.~\ref{sec:automaton} and \ref{sec:entanglement}.

\subsection{Sequential models with uncomputable asymptotics}
\label{sec:uncomp}

To see that there are sequential models for which the asymptotic maximum reward cannot be well approximated, we rely on an intermediate result presented in~\cite{undecidability}, where specific probabilistic finite-state automata for which their asymptotic behavior cannot be well approximated were constructed.

Consider a finite-state automaton ${\cal A}$ with a (finite) state and input sets $\Sigma$ and $X$, initial state $\sigma_0\in\Sigma$ and no outputs. The automaton is said to be \emph{freezable} if there exists an input $x_0\in X$ that keeps the automaton in the same internal state. 
Call some subset $A\subset \Sigma$ the set of \emph{accepted states}. Then we can consider how likely it is to bring the state of the automaton from $\sigma_0$ to some accepted state in $n$ time steps by suitably choosing the input word $(x_1,...,x_n)\in X^{\times n}$. Define

\begin{equation}
{\rm val}_n({\cal A})=\max_{\bar{x}\in X^{\times n}}\mbox{Prob}\left(\sigma_n\in A\right).
\end{equation}
The \emph{value of the automaton} ${\rm val}({\cal A})$ is defined as $\sup_n {\rm val}_n({\cal A})$. Note that, if ${\cal A}$ is freezable, then the function ${\rm val}_n({\cal A})$ is nondecreasing in $n$. In that case, 

\begin{equation}
{\rm val}({\cal A})=\lim_{n\to\infty}{\rm val}_n({\cal A}).
\end{equation}

Given a freezable outputless automaton ${\cal A}$, consider the sequential model with $S_k$ being the set of probability distributions over $\Sigma$, ${H}_k=X$, $\Lambda=\emptyset$ and $s_1=\sigma_0$. The equation of motion is

\begin{equation}
P_{k+1}(\sigma')=\sum_\sigma P_{k}(\sigma)P_{{\cal A}}(\sigma'|\sigma, h),
\end{equation}
\noindent where $P_{k+1}\in S_{k+1}, P_{k}\in S_{k}$, and $P_{{\cal A}}$ denotes the transition matrix of the automaton.

In a sequential problem with $n$ rounds, we define the reward of the system as $r_k=0 \ \forall \ k<n$, $r_n(s_n)=\sum_{\sigma \in A} P_n(\sigma)$. From all the above, it is clear that the maximum total reward of the system corresponds to $\nu_n^\star:={\rm val}_n({\cal A})$. Hence, $\lim_{n\to\infty}\nu_n^\star={\rm val}({\cal A})$.

How difficult is to compute the asymptotic total reward for this kind of sequential models? The next result provides a very pessimistic answer.\footnote{The statement in~\cite{undecidability} is more general in the sense that the lemma is proven when additionally requiring the automata to be resettable. This property is, however, not required for our application and we thus omit it from the statement below.}

\begin{lemma}[Lemma~1 from \cite{undecidability}]
For any rational $\lambda \in (0,1]$, there exists a family ${\cal T}_\lambda$ of freezable probabilistic finite-state automata, with alphabet size $|X|=5$ and with $|\Sigma|=62$ states, such that, for all ${\cal A}\in T_\lambda$,
\begin{enumerate}[label=(\alph*)]
    \item ${\rm val}({\cal A}) \geq \lambda$ or ${\rm val}({\cal A}) \leq \frac{\lambda}{2}$.
    \item it is undecidable which is the case.   
\end{enumerate}  
\end{lemma}

Now, consider the class ${\cal C}$ of sequential problems for which the asymptotic value exists and is contained in $[0,1]$. This class contains the sequential problem of computing the value of any automaton in ${\cal T}_1$. It follows that there is no general algorithm to compute the asymptotic value of any sequential problem in ${\cal C}$ up to an error strictly smaller than $1/4$. Otherwise, one could use it to discriminate between the automata in ${\cal T}_1$ with value $1/2$ and those with value $1$, in contradiction with the above lemma.

\subsection{Heuristic methods for computing upper bounds on the asymptotics of sequential models} \label{sec:heuristics}

Upper bounds on the asymptotic value of a sequential model as $N \rightarrow \infty$ can be achieved by deriving a function $\nu(N)$ that upper bounds the solution of the corresponding sequential problem of size $N$ for any $N$ and taking the limit. 
A way to achieve this is to encode, via functional constraints, recursion relations that define value functions for arbitrarily high $N$. For simplicity of presentation, let us consider a sequential model where $r_k=0$ for $k\not=N$ and $f_k$, $S_k$, and $H_k$ are independent of $k$.\footnote{Notice that this restriction is not necessary. In particular, analogous methods can be used to deal with scenarios where the $f_k, r_k$ are $k$-dependent, as long as we can formulate analogous constraints in terms of polynomials of the corresponding bounded variable $\alpha$ below.}

In the following, we show how to derive asymptotic bounds in case of convergence that is polynomial in $\frac{1}{N}$ as well as for exponential decay with $N$. Similar methods can also be used to bound the solution of the sequential problem even when the latter diverges, e.g., they can be used to prove that $\nu(N)\leq 4N^3$ (although, in this case, reducing the corresponding constraints to optimizations over functions of bounded variables might require some regularization trick (e.g.: $V_k\to V_k/k^3$). 

\subsubsection{Sequential models with polynomial speed of convergence in $\frac{1}{N}$}

Let us assume that, for high $N$, $\nu(N)\approx \bar{\nu}+\frac{A}{N}$, and we wish to compute \new{$\bar{\nu}$}. This problem might arise, e.g., when we wish to study the asymptotic type-I error of a quantum preparation game~\cite{prep_games}. 
To simplify the notation, we use negative indices $-j$ to denote $N-j$. Following the formulation in Eq.~\eqref{dual}, in order to obtain an upper bound on the solution of the sequential problem for size $N$, we need to find $\nu$, $\{V_{-j}(s,\lambda)\}_{j=0}^{N-1}$ such that
\begin{align}
V_0(s,\lambda)&\geq r(s,h,\lambda),\label{initial}\\
V_{-(j+1)}(s,\lambda)&\geq V_{-j}(f(s,h,\lambda),\lambda), \forall  j \in \{0,\ldots,N-1\} \label{recursion}\\
\nu&\geq V_{-(N-1)}(s_1,\lambda).
\label{UB_inverted}
\end{align}
The value functions $V_{-j}$ obtained according to Eq.~\eqref{dual} are generally different for each problem size $N$, so they do not help us solve the asymptotic problem.

Introducing an extra parameter into the value functions, in order to account for the increase in $j$, however, allows us to use the same value function from some $k$ on. To see this, suppose that we manage to show that there exist functions $\{W_{-j}(s,\lambda)\}_{j=0}^k$, $W(\alpha,s,\lambda)$ such that
\begin{align}
W_{0}(s,\lambda) &\geq r(s,\lambda,h),\label{initial_asym}\\
W_{\!-j}(s,\lambda) &\geq W_{\!-(j-1)}(f(s,h, \lambda),\lambda), \forall j \! \in \! \{1,\ldots,k \}, \label{recursion_init}\\
W \! \!\left(\!\frac{1}{k},s,\lambda\!\right) \! &\geq W_{-k}(s,\lambda),\label{attachment}\\
W \! \! \left( \!\frac{\alpha}{\alpha+1},s,\lambda \!\right) \!   &\geq W(\alpha,f(s,h,\lambda),\lambda), \ \forall \alpha\in\left[0,\frac{1}{k}\right]
\label{recursion_inf}
\end{align}
\noindent Then, the functions
\begin{equation}
V_{-j}(s,\lambda):=\begin{cases}W_{-j}(s,\lambda) \quad &j=0,\ldots ,k,\\W\left(\frac{1}{j},s,\lambda\right) \quad &j=k+1,\ldots,N-1\end{cases}
\end{equation}
\noindent satisfy the constraints in Eqs.~(\ref{initial}) and (\ref{recursion}). Intuitively, in Eqs.~(\ref{attachment}) and  (\ref{recursion_inf}), the variable $\alpha$ must be understood to represent $\frac{1}{j}$, in which case $\frac{\alpha}{\alpha+1}$ corresponds to $\frac{1}{j+1}$. The constraint in Eq.~(\ref{recursion_inf}) hence implies the recursion relation in Eq.~(\ref{recursion}) for all $j\geq k$, while the case $j< k$ is  taken care of by Eq.~(\ref{recursion_init}). Note that the value of $\alpha$ is bounded; this is important if we wish to use polynomial optimization methods to enforce the inequality constraints in Eq.~(\ref{recursion_inf}) (see Sec.~\ref{sec:asymptotics_examples}).

Furthermore, if, for some real numbers $B,C,...$, the extra condition
\begin{equation}
\bar{\nu}+A\alpha+B\alpha^2+C\alpha^3+...\geq W\left(\alpha,s_1,\lambda\right), \ \forall\alpha\in \left[0,\frac{1}{k}\right],
\label{asymptotic_UB}
\end{equation}
\noindent is satisfied, then, by Eq. (\ref{UB_inverted}), the function $\bar{\nu}+A\alpha+B\alpha^2+C\alpha^3+...$ upper bounds the solution of the sequential problem. A heuristic to upper bound $\bar{\nu}^\star$ on the limit $\lim_{N\to\infty}\nu(N)$ thus consists in minimizing $\bar{\nu}$ under constraints~\footnote{It is advisable, on the grounds of numerical stability, to also restrict the ranges of possible values for $A, B,C,...$.} 
Eqs.~(\ref{initial_asym})-(\ref{recursion_inf}), and (\ref{asymptotic_UB}).

The previous scheme is expected to lead to good upper bounds in situations in which, for high $N$, the optimal value functions $V_{0},...,V_{-(N-1)}$ of the problem given in Eq.~(\ref{dual}) satisfy 
$V_{-(N-1)}\approx \bar{V}+\frac{\tilde{V}}{N}$,
 for some functions $\bar{V},\tilde{V}$. That is, they work well as long as $\lim_{N\to\infty}V_{-(N-1)}$ exists. Experience with dynamical systems suggests, though, that even if the solution of the sequential problem tends to a given value $\bar{\nu}$, the sequence $(V_{-(N-1)})_N$ might asymptotically approach a \emph{limit cycle}. This means that for many problems only every $\mathcal{C}$-th value function follows the required pattern: in such situations,  there  exist a natural number $\mathcal{C}$ (the period) and functions ${V}^0,...,{V}^{\mathcal{C}-1}$ such that, for high enough $N$, $V_{-N}\approx {V}^{R}$, with $R=N\mbox{ (mod }\mathcal{C})$. Thus, instead of a single recursion relation in Eq.~\eqref{recursion_inf}, we need to construct a family of such relations.

\begin{widetext}
This situation can thus be accounted for by considering functions $\{W_{-j}(s,\lambda)\}_{j=0}^k$, $\{W^j(\alpha,s,\lambda)\}_{j=0}^{\mathcal{C}-1}$ such that
\begin{align}
W_{0}(s,\lambda) &\geq r(s,\lambda,h),\nonumber\\
W_{-j}(s,\lambda) &\geq W_{-(j-1)}(f(s,h, \lambda),\lambda), \ j=1,\ldots,k,\nonumber\\
W^0\left(\frac{1}{k},s,\lambda\right) &\geq W_{-k}(s,\lambda),\nonumber\\
W^1\left(\frac{\alpha}{\alpha+1},s,\lambda\right) &\geq W^0(\alpha,f(s,h,\lambda),\lambda), \ \forall \alpha\in\left[0,\frac{1}{k}\right],\nonumber\\
&\vdots\nonumber\\
W^{\mathcal{C}-1}\left(\frac{\alpha}{\alpha+1},s,\lambda\right) &\geq W^{\mathcal{C}-2}(\alpha,f(s,h,\lambda),\lambda), \ \forall \alpha\in\left[0,\frac{1}{k+\mathcal{C}-1}\right],\nonumber\\
W^0\left(\frac{\alpha}{\alpha+1},s,\lambda\right) &\geq W^{\mathcal{C}-1}(\alpha,f(s,h,\lambda),\lambda), \ \forall \alpha\in\left[0,\frac{1}{k+\mathcal{C}}\right].
\label{asymptotic_cons_cycle}
\end{align}
\end{widetext}

\subsubsection{Sequential models with exponential speed of convergence in $N$}
\label{sec:asymptotics_exp}
Suppose that the speed of convergence of the sequential model is exponential, i.e., $\nu(N)\approx \bar{\nu}+A\gamma^N$, for $\gamma<1$. Assuming, for simplicity, that $\lim_{N\to\infty}V_{-(N-1)}$ exists, how could we estimate both $\bar{\nu}$ and $\gamma$? 

One way to check that the postulated convergence rate holds is to verify the existence of functions $\{W_{-j}(s,\lambda)\}_{j=0}^k$, $W(\alpha,s,\lambda)$ such that
\begin{align}
W_{0}(s,\lambda) &\geq r(s,\lambda,h),\nonumber\\
W_{-j}(s,\lambda) &\geq W_{-(j-1)}(f(s,h, \lambda),\lambda), \ j=1,...,k,\nonumber\\
W\left(\gamma^k,s,\lambda\right) &\geq W_{-k}(s,\lambda),\nonumber\\
W(\gamma\alpha,s,\lambda) &\geq W(\alpha,f(s,\lambda,h),\lambda), \ \forall \alpha\in\left[0,\gamma^k\right]
\label{asymptotic_cons_exp}
\end{align}
\noindent and Eq.~(\ref{asymptotic_UB}) hold. The last line of the above equation is the induction constraint, wherein this time $\alpha$ is to be interpreted as $\gamma^{j}$. As before, we have defined $\alpha$ in terms of $j$ such that the former quantity is bounded. To treat this problem within the formalism of convex optimization, one would try minimizing $\bar{\nu}$ for different values of $\gamma$. Each such minimization, if feasible, would result in an upper bound $\nu(N)$ on the solution of the sequential problem with $\nu(N)\approx \bar{\nu}+O(\gamma^N)$.

\subsection{Adaptation to finite numerical precision} \label{errors}
In standard numerical optimizations, the computer outputs value functions that only satisfy the positivity constraints in Eqs.~(\ref{initial}-\ref{asymptotic_UB}) approximately. More generally, conditions of the form $\bullet\geq 0$ are weakened to $\delta+\bullet\geq 0$, for some $1\gg\delta\geq 0$. 

For small values of $\delta$, the conditions in Eqs.~(\ref{initial}-\ref{attachment}) and (\ref{asymptotic_UB}) are not problematic, in the sense that a small infeasibility of any of these inequalities leads to a constant increase in the value function by that amount. This is why in the resolution of Eq.~(\ref{dual}), we have so far not considered numerical precision issues explicitly.  Indeed, let $\{\delta_j\}_{j=k+1}$ be such that
\begin{align}
W_{0}(s,\lambda)+\delta_0 &\geq r(s,\lambda,h),\nonumber\\
W_{-j}(s,\lambda)+\delta_{j} &\geq W_{-(j-1)}(f(s,h, \lambda),\lambda), \forall j \in \{1,\ldots,k\} \nonumber\\
W_{-k}(s,\lambda)  &\leq W\left(\frac{1}{k},s,\lambda\right)+\delta_{k+1},\nonumber\\
  W\left(\alpha,s_1,\lambda\right) &\leq \bar{\nu}+A\alpha+B\alpha^2+C\alpha^2+...+\delta_{\infty}. \label{eq:delta_fin}
\end{align}
Then, the functions $\tilde{W}_{-j}:=W_{-j}+\sum_{i=0}^j\delta_i$ and $\tilde{W}:=W+\sum_{i=0}^{k+1}\delta_i$ satisfy conditions in Eq.~(\ref{initial}-\ref{attachment}). If, in addition, condition in Eq.~(\ref{recursion_inf}) is satisfied by the original functions, then 
\begin{equation}
\bar{\nu}+\sum_{j=0}^{k+1}\delta_k+\delta_\infty 
\end{equation}
is a sound upper bound of the asymptotic value of the game. 

However, precision issues become more problematic if the relation in Eq.~(\ref{recursion_inf}) is only satisfied approximately, i.e., if 
\begin{equation}
W\left(\frac{\alpha}{\alpha+1},s,\lambda\right)+\hat{\delta}\geq W(\alpha,f(s,h,\lambda),\lambda).
\end{equation}
In that case, for $N>k$, the optimal reward of the $N${th} game round is upper bounded by
\begin{equation}
\bar{\nu}+\sum_{j=0}^{k+1}\delta_k+\delta_\infty+\hat{\delta}(N-k).
\end{equation}
That is, the bound grows linearly with $N$. In that case, the asymptotic upper bounds $\bar{\nu}$ obtained numerically can only be considered valid for $N\ll O\left(\frac{1}{\hat{\delta}}\right)$.

\section{Applications of asymptotic optimization of time-ordered processes} \label{sec:asymptotics_examples}

In this section, we bound the asymptotic behavior of the time-ordered processes that have been treated at a finite regime in Secs.~\ref{sec:automaton} and \ref{sec:entanglement}, respectively, relying on the techniques from Sec.~\ref{sec:asymptotics}. While the former is an example of polynomial speed of convergence in $\frac{1}{N}$, the latter illustrates the technique in the exponential case.

\subsection{Asymptotic behavior of finite-state automata}

The probability that a two-state automaton generates a one-tick sequence of length $L$ seems to converge to $0$ as $O\left(1/L\right)$. It can also be verified that this is the convergence rate of the optimal model found in Ref.~\cite{tick_sequence_paper}. In fact, this model gives the probability of the one-tick sequence 
\begin{equation}
p(2k)=\frac{1}{k}\left(1-\frac{1}{k}\right)^{k-1} \approx \frac{1}{\mathbf{e} k}, \text{ for } k \rightarrow \infty,
\end{equation} 
where $\mathbf{e}$ is the basis of natural logarithms, and $p(2k-1)=p(2k)$. This explicit example provides an upper bound on the converge rate for $\nu(L)$, i.e., $\nu(L)$ cannot converge to zero faster than $2/\mathbf{e}L$. Similarly to this model,  the value $\nu(L)$ of the sequential game computed numerically satisfies $\nu(2k)=\nu(2k-1)$ from $k=2$ onward. 
 This suggests that we should apply the relaxation in Eq.~(\ref{asymptotic_cons_cycle}), with $\mathcal{C}=2$. 

Remember that $s_k:=(p_k(0),p_k(1))$, where $p_k(t)$ denotes the probability that the automaton has so far produced  the outcome $0$ and is in state $t$, and $\lambda:=(P(\sigma',0|\sigma):\sigma,\sigma')$, the  transition matrix of the automaton. We allow the discrete-value functions $W_{-k}(s,\lambda)$ to be linear combinations of all monomials of the form $s^a\lambda^b$, with $a,b=0,1,2$. The continuous-value functions $W^1, W^2$ are linear combinations of the monomials
$s^a\lambda^b\alpha^c$, with $a,b,c=0,1,2$. Positivity constraints such as the last line of
Eq.~(\ref{asymptotic_cons_cycle}) are mapped to a polynomial problem by multiplying the whole expression by $(1+\alpha)^2$. 

To solve this problem, we have used the SDP solver MOSEK \cite{mosek}, in combination with YALMIP \cite{yalmip}. After $28$ iterations, MOSEK returned the upper bound $0.0265$, together with the UNKNOWN message. The solution found by MOSEK was, however, almost feasible, with parameters {PFEAS}$=4\times 10^{-7}$;  DFEAS$=8.6\times 10^{-7}$ \texttt{GFEAS}$=7.5\times 10^{-11}$ and {MU}$=1.3\times 10^{-7}$. It is therefore reasonable to assume\footnote{To further support the conjecture that the value reported by MOSEK is close to the actual solution of the SDP, we considered the value functions $\{W_{-j}\}_{j=0}^3, W^1, W^2$ output by the solver and then used the Lasserre hierarchy to verify that relations (\ref{asymptotic_cons_cycle}) held. We only found tiny violations of positivity (of order $-10^{-5}$) for the last two recursion relations.} that the \emph{actual} solution of the SDP, which for some reason MOSEK could not find, is indeed close to $2\times 10^{-2}$. This figure is, indeed, greater than the conjectured null asymptotic value, but nonetheless a nontrivial upper bound thereof.

\subsection{Asymptotics of entanglement-detection protocols} \label{sec:asymptotics_ghz}

We observe that for small $N$, $p_{\rm max}^{\rm sep}$ behaves as 
\begin{equation}
p_{\rm max}^{\rm sep} \approx \frac{1}{2} + \frac{1}{2^N}. \label{eq:asymptoticGHZ}   
\end{equation}
This value of $p^{\rm sep}$ can also be achieved, by preparing the separable state $\ket{+}^{\otimes N}$. For completeness, we illustrate this in App.~\ref{app:lower_bound_GHZ}. This shows that our bounds from Fig.~\ref{fig:GHZbounds} for the finite $N$ regime are tight (up to numerical precision).

The behavior in Eq.~\eqref{eq:asymptoticGHZ} also implies that we expect an asymptotic value of $\frac{1}{2}$. Implementing the procedure to bound the asymptotic behavior from Sec.~\ref{sec:asymptotics_exp}, we have been able to confirm an upper bound of $0.5004$, before running into precision problems using YALMIP \cite{yalmip} and the solver MOSEK \cite{mosek}. For the implementation, we have extended the one described in App.~\ref{app:lower_bound_GHZ}. with the additional relations from Sec.~\ref{sec:asymptotics_exp}. Note that a value of $\frac{1}{2}$ corresponds to randomly guessing whether the state is entangled {and is achievable for any $N$ with the state in Eq.~\eqref{state}. It is thus a lower bound on $e_I$.} Note that gap between the success probability of $1$ for the GHZ state and approximately $\frac{1}{2}$ for the worst-case separable state, approximately reaches the maximal possible value of such a gap of $\frac{1}{2}$ {(as explained in Eq.~\eqref{bound})}. This illustrates that the multiround GHZ game performs essentially optimally for the task of certifying the entanglement of a GHZ state.

\section{Policies with guaranteed minimum reward}
\label{sec:policies}
Let us complicate our sequential model a little bit, by adding a \emph{policy}. The idea is that the parameters $x$ describing our policy also affects the  evolution of the system, i.e.,
\begin{equation}
    s_{k+1}=f_k(s_k, h_k, \lambda, x)
    \label{motion_eq_policy}
\end{equation}
\noindent as well as the reward $r_k(s_k, h_k, \lambda, x)$ obtained in each time step.

Ideally, we would like to find the policy $x$ that minimizes the maximum penalty. That is, we wish to solve the problem
\be
\min_{x}\max_{h,\lambda} \ \sum_{k=1}^N r_k(s_k, h_k, \lambda, x),
\ee
\noindent under the assumption that Eq.~(\ref{motion_eq_policy}) holds and that we know the initial state $s_1$.

For fixed $x$, though, it might be intractable to compute the minimum expected reward. A more realistic and practical goal is thus to find a policy $x$ and a sufficiently small upper bound $\nu$ on its maximum reward (or penalty) $\nu^\star(x)$. This is the subject of the next section.

\subsection{Optimization over policies}
For any $x$, let $\nu(x)$ denote an upper bound on $\nu^\star(x)$, obtained through some relaxation of problem given in Eq.~(\ref{dual}). We propose to use standard gradient-descent methods or some variant thereof, such as Adam \cite{adam}, to identify a policy $x^\star$ with an acceptable guaranteed minimum reward $\nu(x^\star)$.

More concretely, call $X$ the set of admissible policies, and assume that $X$ is convex. Starting from some guess $x^{(0)}$ on the optimal policy, and given some decreasing sequence of non-negative values $(\epsilon_k)_k$, \emph{projected gradient descent} works by applying the iterative equation
\begin{equation}
    x^{(k+1)}=\pi_X\left(
    {{x^{(k)}}}-\epsilon_k\nabla_x\nu(x^{(k)})\right),
  \label{gradient_descent}
\end{equation}
\noindent where $\pi_Z(x)$ denotes the projection of vector $x$ in set $Z$, i.e., the vector $z\in Z$ closest to $x$ in Euclidean norm. For $k\gg 1$, we would expect policy $x^{(k)}$ to be close to minimal in terms of the guaranteed reward $\nu(x^{(k)})$.

The use of gradient methods requires estimating $\nabla_x\nu(\bar{x})$ for arbitrary policies $\bar{x}$. Considering that we have so far obtained bounds $\nu(x)$ (for fixed $x$) numerically from relaxations of Eq.~\eqref{dual}, this task requires additional theoretical insight. The rest of this section is thus concerned with formulating an optimization problem to estimate the gradient of $\nu(x)$.

Suppose, for simplicity, that $x$ consists of just one real parameter, i.e., we wish to compute $\frac{d}{dx}\nu(\bar{x})$. Fix the policy to be $\bar{x}$, let $\delta x>0$, and let us consider problem given in Eq.~(\ref{dual}). To arrive at a tractable problem, one normally considers a simplification of the following type:
\begin{align}
    \nu(\bar{x})&:=\min_{V_1,...,V_N,\nu} \nu\nonumber\\
    \mbox{s.t. }&V_{N}(s_N,\lambda)\!  \geq_Q \!  r_N(s_N,h,\lambda, \bar{x}),\nonumber\\
    &V_k(s_k,\lambda) \! \geq_Q \! r_k(s_k,h,\lambda,\bar{x}) \! + \! V_{k+1}(f_{k+1}(s_k,h_k,\lambda),\lambda),\nonumber\\
    &\nu \! \geq_Q \!  V_1(s_1,\lambda).
    \label{dual_ansatz}
\end{align}
\noindent Here the relation $a\geq_Q b$ signifies that the expression $a-b$, evaluated on some convex superset $Q$ of the set of probability distributions on the variables $\{s_k,h_k\}_k,\lambda$, is non-negative. In the polynomial problems considered in this paper, this superset $Q$ is, in fact, dual to the positivstellensatz used to enforce the positivity constraints. In the case of Putinar's positivstellensatz, $Q$ would correspond to the set of monomial averages $\{\langle z_{i_1}...z_{i_N}\rangle\}$ that, arranged in certain ways, define positive semidefinite matrices (for details, see \cite{lasserre}). 

Since $Q$ contains all probability distributions, $a\geq_Q b$ implies that $a\geq b$, and so the solution of Eq.~(\ref{dual_ansatz}) is a feasible point of Eq.~(\ref{dual}) and hence an upper bound $\nu(\bar{x})$ on $\nu^\star(\bar{x})$. For the time being, we  assume that the minimizer $V_1,...,V_N$ of problem given in Eq.~(\ref{dual_ansatz}) is unique for policy $\bar{x}$. We will drop this assumption later.

Now, let us consider how problem given in Eq.~(\ref{dual_ansatz}) changes when we replace $\bar{x}$ by $\bar{x}+\delta x$. We assume not only that the solution satisfies $\nu(\bar{x}+\delta x)\approx \nu(\bar{x})+\mu\delta x$, but also that the optimal slack variables $V_1,...,V_N$ also experience an infinitesimal change, i.e., 
\begin{equation}
V_k(s_k,\lambda)\to V_{k}(s_k,\lambda) + \hat{V}_k(s_k,\lambda)\delta x.
\label{diff_cond}
\end{equation}
\noindent That is, not only is the solution differentiable, but also the minimizer of the optimization problem.

\begin{widetext}
Ignoring all terms of the form $o(\delta x)$ in Eq.~\eqref{diff_cond}, it is clear that $\nu+\mu\delta x$ is an upper bound on $\nu^\star(\bar{x}+\delta x)$ if and only if
\begin{align}
    V_N(s_N,\lambda)+\hat{V}_N(s_N,\lambda)\delta x \geq_Q & \ r_N(s_N,h,\lambda,\bar{x})+\frac{\partial}{\partial x} r_N(s_N,h,\lambda,\bar{x})\delta x,\label{diff_N}\\
    V_k(s_k,\lambda)+\hat{V}_k(s_k,\lambda)\delta x \geq_Q & \ r_k(s_k,h,\lambda,\bar{x})+\frac{\partial}{\partial x} r_k(s_k,h,\lambda,\bar{x})\delta x+V_{k+1}(s_{k+1},\lambda)+\hat{V}_{k+1}(s_{k+1},\lambda)\delta x \\
    &\ +\delta x\sum_{\sigma}\frac{\partial V_{k+1}}{\partial s_{k+1}^\sigma}\frac{\partial f^\sigma_k(s_k,h_k,\lambda,x)}{\partial x}\nonumber\\
    \nu+\mu\delta x \geq_Q&  \ V_1(s_1,\lambda)+\hat{V}_1(s_1,\lambda)\delta x\label{diff_nu}.
\end{align}
Here, the symbol $\sigma$ is used to denote the coordinates of $s_{k+1}$.
\end{widetext}

Our goal is now to turn this into an optimization problem that allows us to find $\mu$. For this purpose it is useful to invoke the fact that the slack variables $\{V_k\}_k$ are optimal, i.e., they not only satisfy the conditions in the problem given in Eq.~\eqref{dual_ansatz}, but all inequalities must be saturated. Namely, for each $k$, there exists $q\in Q$ such that 
\begin{equation}
    \left\langle V_k(s_k,\lambda)-r_k(s_k,h_k,\lambda)- V_{k+1}(s_{k+1},\lambda)\right\rangle_q=0.
    \label{saturate}
\end{equation}
To proceed, we need the following lemma, which allows us to phrase the constraints in Eqs.~\eqref{diff_N}-\eqref{diff_nu} in terms of the quantities $\hat{V}_k$. 

\begin{lemma}
Let $a\geq_Q 0$, and let the set $Q':=\{q: q\in Q,\langle a\rangle_q=0\}$ be nonempty. Then, $a+\hat{a}\delta x\geq_Q 0$ for $\delta x\to 0^+$ if and only if $\hat{a}\geq_{Q'} 0$. 
\end{lemma}
\begin{proof}
$Q'\subset Q$, so $a+\hat{a}\delta x\geq_Q 0$ implies that $\hat{a}\geq_{Q'} 0$. To see the converse implication, note that, for $q\in Q\setminus Q'$, $\langle a\rangle_q>0$ and therefore $\lim_{\delta x\to 0^+}\langle a+\hat{a}\delta x\rangle_q> 0$. To determine if $\langle a+\hat{a}\delta x\rangle_q\geq 0$ for all $q\in Q$ in the limit $\delta x\to 0^+$, it is thus sufficient to check that the property holds for $q\in Q'$, i.e., that $a\geq_{Q'} 0$. 
\end{proof}

By duality, condition $\hat{a}\geq_{Q'} 0$ is equivalent to

\begin{equation}
\xi a+\hat{a}\geq_Q 0,\mbox{ for some } \xi\in\mathbb{R}.
\end{equation}

\noindent From the above equation and relations in Eq.~(\ref{saturate}), it follows that (under the hypothesis that the minimizer $V$ is differentiable and unique at $\bar{x}$) the limit $D^{+}\nu(\bar{x})\equiv \lim_{\delta x\to 0^+}\frac{\nu(\bar{x}+\delta x)-\nu(\bar{x})}{\delta x}$ equals

\begin{widetext}
\begin{align}
    &\min_{\mu,\hat{V},\xi} \mu\nonumber\\
    \mbox{such that } \ &\hat{V}_N(s_k,\lambda)\geq_Q \frac{\partial}{\partial x} r_N(s_N,h_N,\lambda,\bar{x})+\xi_N\left(V_N(s_k,\lambda)-r_N(s_N,h_N,\lambda,\bar{x})\right),\nonumber\\
    &\hat{V}_k(s_k,\lambda)\geq_Q \frac{\partial}{\partial x} r_k(s_k,h,\lambda,\bar{x})+\hat{V}_{k+1}(s_{k+1},\lambda)+\sum_{\sigma}\frac{\partial V_{k+1}(s_{k+1},\lambda)}{\partial s_{k+1}^\sigma}\frac{\partial f^\sigma_k(s_k,h_k,\lambda,x)}{\partial x}+\nonumber\\
    & \qquad \qquad \quad +\xi_k\left(V_k(s_k,\lambda)-r_k(s_k,h_k,\lambda,\bar{x})-V_{k+1}(s_{k+1},\lambda)\right),\nonumber\\
    &\mu\geq_Q \hat{V}_1(s_1,\lambda)+\xi_1\left(\nu-V_1(s_1,\lambda)\right).
    \label{perturbed}
\end{align}
If, in addition, the sets
\begin{align}
&Q'_k:=\{q\in Q:q\left(V_k(s_k,\lambda)-r_k(s_k,h_k,\lambda,\bar{x})-V_{k+1}(s_{k+1},\lambda)\right))=0\}, k=1,...,N,\nonumber\\
&Q'_0:=\{q\in Q:q\left(\nu-V_1(s_1,\lambda)\right))=0\}
\end{align}
have each cardinality $1$, i.e., $Q_k'=\{q_k\}$, the problem to solve becomes even simpler, namely:
\begin{align}
    &\min_{\mu,\hat{V},\xi} \mu\nonumber\\
    \mbox{such that } \ &q_N\left(\hat{V}_N(s_k,\lambda)\right)= q_N\left(\frac{\partial}{\partial x} r_N(s_N,h_N,\lambda,\bar{x})\right),\nonumber\\
    &q_k\left(\hat{V}_k(s_k,\lambda)\right)= q_k\left(\frac{\partial}{\partial x} r_k(s_k,h,\lambda,\bar{x})+\hat{V}_{k+1}(s_{k+1},\lambda)+\sum_{\sigma}\frac{\partial V_{k+1}(s_{k+1},\lambda)}{\partial s_{k+1}^\sigma}\frac{\partial f^\sigma_k(s_k,h_k,\lambda,x)}{\partial x}\right),\nonumber\\
    &\mu= q_0\left(\hat{V}_1(s_1,\lambda)\right).
    \label{perturbed_single_q}
\end{align}
\end{widetext}

The derivation of Eq.~(\ref{perturbed}) relies on the hypothesis that the optimal slack variables $V:=\{V_k\}_k$ for policy $x=\bar{x}$ are unique. Should  this not be true, the procedure to compute $D^+\nu(\bar{x})$ requires solving a nonconvex-optimization problem. The reader can find it in App.~\ref{app:gradient}, together with a heuristic to tackle it.

Now, let us return to the case of a multivariate policy $x$. If the gradient $\nabla_x \nu(\bar{x})$ exists, then its 
$i${th} entry corresponds to the limit
\begin{equation}
\lim_{\delta x\to 0^+}\frac{\nu(\bar{x}+\delta x\ket{i})-\nu(\bar{x})}{\delta x}.
\end{equation}
\noindent Each of these entries can be computed, in turn, via the procedure sketched above. This requires solving an optimization problem of complexity comparable to that of computing $\nu(\bar{x})$. Moreover, each such optimization can be performed separately for each coordinate of $x$. Namely, the process of computing the gradient of $\nu(\bar{x})$ can be \emph{parallelized}. This allows us, through Eq.~(\ref{gradient_descent}), to optimize over policies consisting of arbitrarily many parameters, as long as we have enough resources to compute $\nu(\bar{x})$.

\section{Application: optimization of adaptive protocols for magic state detection}
\label{sec:magic}
Magic states, i.e., states that cannot be expressed as convex combinations of stabilizer states, are a known resource for quantum computation: together with Clifford gates, they allow us to conduct universal quantum operations efficiently \cite{magic1}. This raises the problem of certifying whether a given source of states is able to produce them \cite{magic2}. 

In the qubit case, magic states are those that cannot be expressed as a convex combination of the eigenvectors of the three Pauli matrices. 
Equivalently, magic states are those the Bloch vector $\vec{n}$ of which violates at least one of the inequalities:
\begin{equation}
\left\{\sum_{j=1}^3 n_ja_j\leq 1: a_1,a_2,a_3\in\{-1,1\}\right\},
\end{equation}
(for an illustration, see Fig.~\ref{fig:magic}).
\begin{figure}[h]
  \centering \includegraphics[width=7cm]{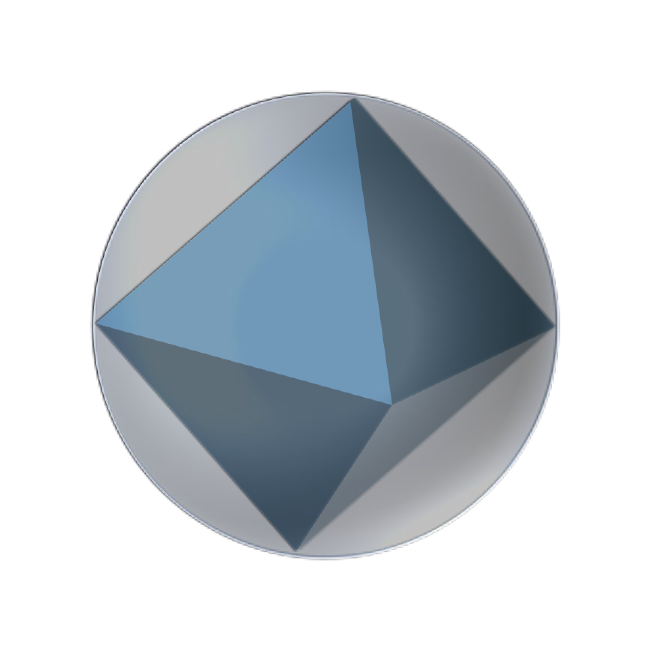}
  \caption{\textbf{Qubit magic states.} Represented in the Bloch sphere, magic states of a qubit are all those that lie outside the blue octahedron. In this section the goal is to certify whether a general state is magic or not.}
  \label{fig:magic}
\end{figure}

We wish to devise an $N$-round adaptive measurement protocol that tells us whether a given source is preparing magic states. Namely, if the source can just produce nonmagic states, we expect the protocol to declare the source ``magic'' with low probability $e_{I}$ (the protocol's type-I error). If, however, the source is actually preparing independent copies of any quantum state violating one of these inequalities by an amount greater than or equal to $\delta$, we wish the protocol to declare the source `non-magic' with low probability $e_{II}$ (the type-II error of the protocol), see Fig.~\ref{fig:magic2}.
\begin{figure}[h]
  \centering \includegraphics[width=7cm]{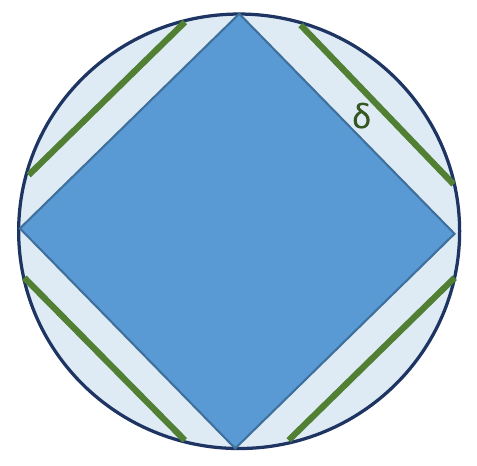}
  \caption{\textbf{An illustration of magic state detection.} We display a cut through the Bloch sphere showing nonmagic (blue square) and magic states (region outside the square). The goal is to devise an $N$-round protocol that is able to reliably classify states into magic and nonmagic ones, as long as the considered magic states are at least $\delta$ away from the nonmagic ones. The protocol derived below, based on a two-state automaton, is valid in one of the regions outside the square of nonmagic states.}
  \label{fig:magic2}
\end{figure}

To do so, we formulate the protocol as a time-independent preparation game \cite{prep_games}. That is, in every measurement round $k$, our knowledge of the  magic of the state is encoded in the configuration $t\in T$ of the game, with $|T|<\infty$. The state prepared by the source is then measured by means of the POVM with elements $(M_{t'|t}: t'\in T)\subset B(\C^2)^{\times |T|}$. The result $t'$ of this measurement is the  new configuration of the game. When we exceed the total number of measurement rounds $n$, we must guess, based on the current game configuration, whether the state source is magic or not. For simplicity, we declare the source to be magic if $t\in A\subset T$, with $|A|=\left\lfloor\frac{|T|}{2}\right\rfloor$. The initial state of the game, $t_1$, is chosen such that $t_1\not\in A$.

The game is thus defined by the POVMs $\{ M_{t'|t}: t' \in T \}_t$. For computing $e_{I}$, we assume that the `dishonest' player, who claims to produce magic states but does not, knows at every round the current game configuration. They can thus adapt their state preparation correspondingly to increase the type-I error. As explained in 
Sec.~\ref{sec:prep_games}, the computation of the maximum score of a preparation game under adaptive strategies can be cast as a sequential problem with $h_k:=(\rho_k(t):t\in T_k)$. In this case, the set of feasible states, $C$, is generated by convex combinations of $\{\psi_i\}_{i=1}^6$, the eigenvectors of the three Pauli matrices: linear maximizations over $C$ therefore correspond to maximizations over these six ``vertices''. Putting it all together, we find that computing the  type-I error of the game amounts to applying the recursion relation:
\begin{align}
\mu_N(t) &=\max_i \sum_{t'\in A}\tr(\proj{\psi_i}M_{t'|t}),\nonumber\\
\mu_k(t) &=\max_i\sum_{t'}\tr(\proj{\psi_i}M_{t'|t})\mu_{k+1}(t'),\nonumber\\
e_I &=\mu_1(t_1).
\end{align}

Let us assume that all the maximizations above have a unique maximizer, and call $i(k,t)$ the maximizer corresponding to the $k$th round and game configuration $t$. Then, the gradient of $e_I$ with respect to the POVM element $M_{\tau'|\tau}$ satisfies the recursion relations:
\begin{align}
\vec{\nabla}\mu_N(t) &=\delta_{\tau,t}\proj{\psi_{i(N,\tau)}},\nonumber\\
\vec{\nabla}\mu_k(t) & =\sum_{t'}\tr(\proj{\psi_{i(k,t)}}M_{t'|t})\vec{\nabla}\mu_{k+1}(t') \nonumber\\ &\qquad \qquad \ +\delta_{\tau,t}\mu_{k+1}(\tau')\proj{\psi_{i(k,\tau)}},\nonumber\\
\vec{\nabla}e_I &=\vec{\nabla}\mu_1(t_1).
\end{align}
It can hence be computed efficiently in the number of measurement rounds.

Now, consider the computation of the type-II error under IID strategies, the formulation of which as a sequential problem can be found in Sec.~\ref{sec:iid_strategies}. We first focus on the set ${\cal S}_a$ of qubit states with Bloch vector $\vec{n}$ satisfying $\sum_ja_jn_j\geq 1+\delta$, for some $a_1,a_2,a_3\in\{-1,1\}$. Applied to a source that always produces the same state $\rho\in {\cal S}$, the magic detection protocol can be modeled through a sequential model, with internal state $s_k$ at time $k$ given by the probability distribution over the game configurations $P_k(t)$ in round $k$, just before the measurement. The state $\rho\in{\cal S}$ prepared by the source is to be identified with the evolution parameters $\lambda$, as the equation of motion of the model is:

\begin{equation}
P_{k+1}(t')=\sum_{t}\tr(\rho M_{t'|t})P_k(t).
\end{equation}
To calculate the type-II error, we assign the model the rewards $r_k=0$, for $k=1,...,N-1$ and 

\begin{equation}
r_N(P_N,\rho)=\sum_tP_N(t)\sum_{t'\in A}\tr(\rho M_{t'|t}).
\end{equation}

We can thus compute an upper bound on $e_{II}({\cal S}_a)$, the maximum type-II error achieved with states in class ${\cal S}_a$, via the SDP relaxation of Eq.~(\ref{dual}); and its gradient with respect to the POVM element $M_{t'|t}$, via Eq. (\ref{perturbed}) \footnote{The dual problem of (\ref{dual}) does not seem to have a unique solution, because we find that the use of Eq. (\ref{perturbed_single_q}) results in infeasible semidefinite programs.} We assume that the solution of the problem given in Eq.~(\ref{dual}) is unique, so the techniques of App.~\ref{app:gradient} are not necessary. 

Our original problem, though, is to compute the type-II error of all states violating at least one of the facets of the magic polytope by an amount $\geq \delta$.  Thus we have that $e_{II}=\max_ae_{II}({\cal S}_a)$. The gradient (or, more properly said, subgradient) of this function is $\vec{\nabla}e_{II}({\cal S}_{a^\star})$, where $a^\star$ is the argument of $\max_ae_{II}({\cal S}_a)$.

We are now ready to apply gradient descent to minimize the combined error $e_I+e_{II}$. We have used the gradient method Adam \cite{adam}, followed at each iteration by a projection onto the (convex) set of protocols $\{M_{t'|t}:t'\in T\}_{t \in T}$, with $M_{t'|t}\geq 0,\forall t,t'\in T, \sum_{t'}M_{t'|t}=\id,\forall t\in T$. For $|T|=2$, $n=6$, gradient descent converges to a value $e_I+e_{II}\approx 1$, probably an indication that one cannot detect all magic states with a two-state automaton. If, however, we restrict the problem to that of detecting those states that satisfy the inequality $\sum_i n_i\geq 1+\delta$, we arrive at the plot shown in Fig.~\ref{fig:gradient_descent}.

\begin{figure}
  \centering \includegraphics[width=9cm]{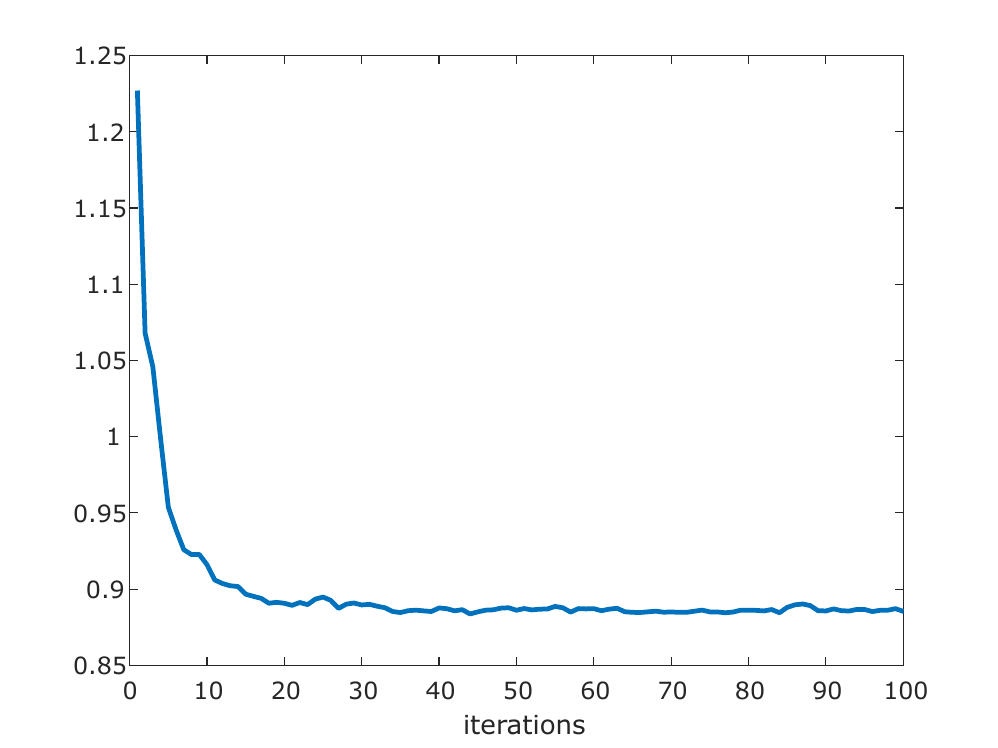}
  \caption{\textbf{The upper bounds for $e_I+e_{II}$ as a function of the number of gradient iterations.} Starting from a random preparation game, we have used the Adam algorithm \cite{adam} to decrease the sum of the type-I and type-II errors of a magic state detection protocol. For this computation, we set the  learning rate of the algorithm to $0.01$. }
  \label{fig:gradient_descent}
\end{figure}

As the reader can appreciate, after a few iterations, the algorithm successfully identifies a strategy with $e_I+e_{II}<1$. Beyond $50$ iterations, the objective value does not seem to decrease much further. This curve has to be understood as a proof of principle that gradient methods are useful to devise preparation games. Note that, for a general preparation game, the number of components of the gradient is $4|T|^2$. Optimizations over preparation games with high $|T|$ would thus benefit from the use of parallel processing to compute the whole gradient vector.

\section{Conclusions} \label{sec:conclusion}
We have proposed a method to relax sequential problems involving a finite number of rounds. This method allows us to solve optimization problems of sizes that have not been tractable with previous convex optimization techniques. Moreover, a variant of the method generates upper bounds on the solution of the problem in the asymptotic limit of infinitely many rounds. 

We have demonstrated the practical use of our methods by solving an open problem: namely, we have computed the maximum probability for any finite-state automaton of a certain fixed dimension to produce certain sequences, thus proving that the values conjectured in Refs.~\cite{tick_sequence_paper,Vieira_Budroni} are optimal (up to numerical precision). We further illustrated how to turn such bounds into analytical results. 

As we have shown, the methods are also relevant in the context of certifying properties of quantum systems, especially in entanglement detection. Specifically, our method enables us to bound the misclassification error for separable states in multiround schemes as well as in the asymptotic limit, which we illustrate by means of a GHZ-game-inspired protocol, for which we also show that our bounds are tight. 

Finally, we have explained how to combine these methods with gradient-descent techniques to optimize over sequential models with a certified maximum reward. This has allowed the computer to generate an adaptive protocol for magic-state detection.

In some of our optimization problems, we have found that otherwise successful SDP solvers, such as MOSEK \cite{mosek}, have failed in some cases to output a reliable result. The cause of this atypical behavior is unclear, and we lack a general solution for this issue. That said, the results from the numerical optimization have still allowed us to extract and confirm reliable bounds from the optimizations in most cases.

In two of the considered applications we found tight upper bounds on the quantities of interest. This may indicate that at least in reasonably simple cases our hierarchical method converges at relatively low levels of complexity. Finding conditions for convergence at such low levels is an interesting question that we leave open.

Quantum technologies are currently at the verge of building computing devices that can be used beyond purely academic purposes.
These involve more than a few particles and thus the problem of certifying properties of systems of intermediate sizes is currently crucial. 
This is relevant not only from the perspective of a company building such devices but also from the perspective of a user who may want to test them 
independently. {This progress is somewhat in tension with the state of research in certification, e.g. in entanglement theory, where methods are only abundant in the few-particle regime. Our GHZ game illustrates that automaton-guided protocols may bridge this gap: modeling one-way LOCC protocols through sequential processes might allow the certification of large classes of many-body systems, thus offering an alternative to protocols based on shadow tomography~\cite{aaronson_shadow, kueng_shadow}.}

Our methods are furthermore able to certify systems beyond the linear regime that is employed when considering the usual types of witnesses for certifying
systems, which from a mathematical viewpoint essentially distinguish convex sets by finding a separating hyperplane. As suggested in the example of magic states, our techniques beyond this regime also mean that we can analyse more general classes of systems simultaneously. Indeed, using a finite-state automaton of large enough dimension, we expect to be able to detect any states that have at least a certain distance from the set of nonmagic states. As optimization methods improve, we expect such \emph{universal} protocols to become more and more common as they are also applicable  with minimal knowledge about the state of the system at hand.

Finally, many different problems encountered in quantum information theory have a sequential structure similar to the one analyzed in this paper. Thus we expect that our methods will soon find applications beyond those studied in the present work. In this regard, it would be interesting to adapt our techniques to analyze quantum communication protocols. Another natural research line would be to extend our methods to model interactions with an evolving quantum system, rather than a (classical) finite state automaton.

\acknowledgements
\noindent
This research was funded in whole, or in part, by the Austrian Science Fund (FWF) [Projects 10.55776/M3109 (Lise-Meitner), 10.55776/ZK3 (Zukunftskolleg), 10.55776/F71 (BeyondC), and 10.55776/P35509] and the Swiss National Science Foundation (Ambizione PZ00P2\_208779).

\appendix

\section{The Lasserre-Parrilo hierarchy and variants}\label{app:polynomials} \label{sec:polynomials}

The Lasserre-Parrilo method provides a complete hierarchy of sufficient conditions to certify that a polynomial is non-negative, if evaluated in a compact region defined by a finite number of polynomial inequalities \cite{lasserre, parrilo}. Given a number of polynomials $g_1(z),...,g_u(z), \tilde{g}_1(z),...,\tilde{g}_v(z)$, we define the region $Z\subset \R^n$ as
\begin{equation}\label{eq:def_Z_set}
Z=\{z\in\R^n:g_1(z),...,g_u(z)\geq 0, \tilde{g}_1(z)=...=\tilde{g}_v(z)=0\}
\end{equation}
\noindent Let $p(z)$ be a polynomial, and consider the problem of certifying that, for all $z\in Z$, $p(z)\geq 0$. A sufficient condition is that $p(z)$ can be expressed as

\be
p(z)=\sum_i\tilde{g}_i(z)\tilde{f}_i(z)+\sum_i f^i(z)^2+\sum_{i,j} f^i_j(z)^2g_j(z),
\label{SOS}
\ee
\noindent where $\tilde{f}_i(z), f_i(z), f_{ij}(z)$ are polynomials on the vector variable $z$. The above is called a \emph{Sum of Squares (SOS) decomposition} for polynomial $p(z)$. Deciding whether $p(z)$ admits a SOS with polynomials $\tilde{f}_i(z),f_i(z),f_{ij}(z)$ of bounded degree can be cast as a semidefinite program (SDP) \cite{lasserre}. Moreover, as shown in \cite{lasserre}, as long as $Z$ is a bounded set and $p(z)$ is strictly positive in $Z$, such a decomposition always exists \footnote{More precisely, a sufficient condition for the existence of an SOS decomposition for a positive polynomial $p(z)$ is that, for some $K>0$, the polynomial $K-\sum_iz_i^2$ admits a SOS decomposition. The latter, in turn, implies that $Z$ is bounded.}.
\begin{widetext}
A weaker criterion for positivity consists in demanding the existence of a decomposition\footnote{This adaptation of the more common Lasserre-Parrilo hierarchy relies on the Schmuedgen Positivstellensatz \cite{schmuedgen} instead and is employed here to obtain a better numerical performance in the problems considered later.} 

\be
\begin{split}
p(z)=\sum_i\tilde{g}_i(z)\tilde{f}_i(z)+\sum_i f^i(z)^2 +\sum_{i,j} f^i_{j}(z)^2g_j(z) +\sum_{i}\sum_{j>k} f^i_{jk}(z)^2g_j(z)g_k(z)+\sum_{i}\sum_{j>k>l} f^i_{jk}(z)^2g_j(z)g_k(z)g_l(z)+...
\label{SOS_boost}
\end{split}
\ee
\noindent If all terms in the decomposition have degree $n$ or less on its variables, we say that $p(z)\in SOS^n(g,\tilde{g})$.

We next sketch how to apply either hierarchy to tackle problem (13) in the main text. Let $(g^k,\tilde{g}^k)$ be the polynomials defining the sets $S_k, S_{k+1}, H_k, \Lambda$, as well as the condition $s_{k+1}=f_k(s_k,h_k,\lambda)$. Take a sufficiently high natural number $n$ and consider the following SDP:
\begin{align}
    &\nu^n:=\min_{V_1,...,V_N,\nu} \nu\nonumber\\
    \mbox{subject to }&V_1,...,V_N,\mbox{polynomials of degree $n$},\nonumber\\
    &V_{N}(s_N,\lambda)- r_N(s_N,h,\lambda)\in SOS^n(g^N,\tilde{g}^N)\nonumber\\
    &V_k(s_k,\lambda)- r_k(s_k,h,\lambda)- V_{k+1}(f_k(s_k, h, \lambda),\lambda)\in SOS^n(g^k,\tilde{g}^k),\nonumber\\
    &\nu- V_1(s_1,\lambda)\in SOS^n(g^1,\tilde{g}^1).
    \label{dual_poly}
\end{align}
From the discussion above, it follows that $\nu^n\geq \nu^\star$. Moreover, if $S_k,H_k,\Lambda$ are bounded, $\lim_{n\to\infty}\nu^n=\nu^\star$.

\end{widetext}

The same idea can be used to model constraints (47) and (49) of the main text, by demanding the functions $\{W_{-j}\}_{j=0}^k$, $W$ to be polynomials of $\alpha,s,\lambda$. Because of the way we defined $\alpha$, this variable is bounded, and hence the Lasserre-Parrilo hierarchy is guaranteed to converge. Enforcing constraints of the form (45) of the main text, where one of the polynomials is evaluated with $\frac{\alpha}{\alpha+1}$, can be dealt with by introducing a new variable $\beta$ representing $\frac{\alpha}{\alpha+1}$, together with the polynomial constraint $\beta(\alpha+1)=\alpha$. The latter, in turn, can be modeled though the conditions $\beta(\alpha+1)-\alpha\geq 0$, $\alpha-\beta(\alpha+1)\geq 0$.

In order to represent the conditions in Eqs.~\eqref{SOS} and \eqref{SOS_boost} in a SDP, we need to represent polynomials in a vector form. This is done, for instance, by fixing a basis of all monomials $\{ m_i(z) \}_{i=1}^{N}$ up to the degree that is necessary to describe the polynomials in our problem. Let us denote this basis as ${\rm MB}_{\rm in}$. Each polynomial, then, is a linear combination of these monomials. This allows us to write the SOS condition in a SDP form, such as, e.g.,
\begin{equation}
\begin{split}
     \sum_{i} f^i(z)^2g(z) = \sum_{k,l} Z_{kl} m_k(z) m_l(z) g(z)
\end{split}
\end{equation}
for $Z$ an appropriately chosen positive-semidefinite matrix. Similarly, we have
\begin{equation}
    \tilde{g}(z)\tilde{f}(z) = \sum_i \alpha_i g(z) m_i(z),
\end{equation}
for some vector $\mathbf{\alpha}$. The results of these operations must, then, be written again in a vector form in terms of a new basis of monomial, which we denote as ${\rm MB}_{\rm out}$. The two bases, ${\rm MB}_{\rm in}$ and ${\rm MB}_{\rm out}$ allow us to represent all constraints as linear and positive-semidefinite constraints. See \cite{lasserre} for more details and the next sections for concrete implementations of this optimization problem.

\section{Implementation of optimization of sequential models}\label{app:implementation0}

\subsection{Probability bounds for the one-tick sequence}\label{app:implementation1}
In order to implement the optimization problem in Eq.~\eqref{dual_poly}, we need first to fix the relevant variables for describing our problem. To keep the notation lighter, we discuss the case $d=2$, the case $d=3$ can be easily understood as its generalization. To describe the states of the automaton we use the variables
\begin{equation}
s_k=(p_k(t_k=1,\sigma_{k+1}))_{\sigma},
\end{equation}
with constraints $s_k^0,s_k^1\geq 0$, $1-s_k^0-s_k^1\geq 0$, which gives  two variable for each $k$. Notice that there is no need to describe the case $t_k=0$ as that would not contribute in the calculation of the probability of successfully outputing the one-tick sequence.

Similarly, for the transition parameters $\lambda$ one has just four variables: $\lambda:=(P(\sigma',b=0|\sigma):\sigma,\sigma'=0,1)$, 
since for the case $b=1$ we do not need to keep track of the internal state and the probability can be recovered by the normalization condition. They satisfy the constraints
\begin{align}
&P(\sigma',b=0|\sigma)\geq 0, \mbox{ for all } \sigma,\sigma'=0,1,\nonumber\\
&1-\sum_{\sigma'}P(\sigma',b=0|\sigma)\geq 0, \mbox{ for }\sigma=0,1.
\end{align}

Finally, there are the equality constraints that come from the the equation of motion 
\begin{equation}
p_{k+1}(t_k=1,\sigma_{k+1})=\sum_{\sigma_k}p_{k}(t_{k-1}=1,\sigma_{k})P(\sigma_{k+1},0|\sigma_k),
\end{equation}
for $k=1,...,L-1$. Finally, the initial state is fixed to be $s_1=(1,0)$.

This gives at each time step (except the initial and final one) four variables for the state $s_{k}^{0}, s_k^{1}, s_{k+1}^0, s_{k+1}^1$ and four variables for the transition parameters $\{\lambda_{\sigma,\sigma'}\}_{\sigma, \sigma' =0,1}$. These are the variables on which the polynomials appearing in Eq.~\eqref{eq:def_Z_set} are built. 

To impose the problem's constraints, we need the two bases of monomials ${\rm MB}_{\rm in}$ and ${\rm MB}_{\rm out}$ discussed in \ref{sec:polynomials}. These allow us to write the polynomials $g_i,\tilde g_i$ appearing in the constraints in \cref{eq:def_Z_set} as well as the polynomials in \cref{SOS} and (possibly) \cref{SOS_boost}, as vectors. For the specific calculations of the probability for the one-tick sequences, Figs.~5 and 6 from the main text, it was sufficient to use Putinar's positivstellensatz, i.e., \cref{SOS}. The input and output bases of monomials are chosen as follows.
\begin{itemize}
\item Construct three sets: monomials of degree two in the variables $(s_k^0,s_k^1)$, in the variables $(s_{k+1}^0,s_{k+1}^1)$ and in the variables $\{\lambda_{\sigma \sigma'}\}_{\sigma \sigma'}$.
\item Combine these three sets together and add monomials of the form $s_k^i \lambda_{\sigma \sigma'}$, for $\sigma,\sigma', i=0,1$ and of the form $s_{k+1}^i \lambda_{\sigma \sigma'}$ for $\sigma,\sigma', i=0,1$. This operation concludes the construction of the input basis ${\rm MB}_{\rm in}$. 
\item Construct in a similar way two other bases, ${\rm MB}_k$ and ${\rm MB}_{k+1}$, defined as follows. ${\rm MB}_k$ is the union of monomials of degree two in the variables, respectively, $(s_k^0,s_k^1)$ and $\{\lambda_{\sigma \sigma'}\}_{\sigma \sigma'}$, together with monomials of the form  $s_k^i \lambda_{\sigma \sigma'}$, for $\sigma,\sigma', i=0,1$. ${\rm MB}_{k+1}$ has the same construction with $s_k^i$ substituted by $s_{k+1}^i$. In other words, these bases are defined in the same way as ${\rm MB}_{\rm in}$, but this time excluding all monomials containing $s_{k+1}$ or containing $s_k$. 
\item The output basis is given by the tensor product of three copies of the input basis. In other word, each monomial of this list is a product of three monomials appearing in the input list. This concludes the construction of the output basis ${\rm MB}_{\rm out}$.
\end{itemize}

Once we have these bases, we want to implement the constraints
\begin{equation}
V_k(s_k,\lambda) - V_{k+1}(s_{k+1},\lambda) \in SOS^\ell(g_i, \tilde{g}_i),
\end{equation}
where $SOS^\ell(g_i, \tilde{g}_i)$ denotes the fact that the SOS constraints are implemented up to a level $\ell$ which depends on the choice of the input (and possibly output) basis above. 
The function  $V_k(s_k,\lambda)$ is a polynomial given by a linear combination of monomials in the list ${\rm MB}_k$. Similarly, $V_{k+1}(s_{k+1},\lambda)$ is a linear combination of elements of ${\rm MB}_{k+1}$.
The condition of belonging to $SOS^\ell(g_i, \tilde{g}_i)$ means that $V_k(s_k,\lambda) - V_{k+1}(s_{k+1},\lambda)$ can be written in the form of \cref{SOS} for properly chosen polynomials $\tilde{f}_i, f_i, f_{ij}$. In our implementation, $\tilde{f}_i, f_i, f_{ij}$ are convex combinations of monomials in the ${\rm MB}_{\rm in}$ basis. For the way in which ${\rm MB}_{\rm out}$ is chosen, all possible monomials arising from the product of
 $\tilde{f}_i, f_i, f_{ij}$ with $\tilde{g}_i, g_i$ appear in ${\rm MB}_{\rm out}$. 
 
For the case $d=2$, the SDP is run with CVXPY \cite{cvxpy} and MOSEK \cite{mosek} up to $L=50$. The values coincide with the explicit model found in \cite{tick_sequence_paper} up to the fifth decimal digit.  For the case $d=3$, the SDP is run with  CVXPY \cite{cvxpy} and  SCS \cite{scs} up to $L=10$. The numerical solutions coincide with the optimal model found in \cite{tick_sequence_paper} and the values explicitly reported in \cite{Vieira_Budroni}, up to the fourth decimal digit. For each explicit solutions provided by the solvers we verified that indeed the linear and positivity constraints are satisfied up to numerical precision.

\subsection{Extracting value functions for the one-tick-sequences and finding optimal automata} \label{app:value_fct}

In order to extract value functions that are suitable for further treatment, simplifications of the output of the polynomial optimization (from App.~\ref{app:implementation1}) is desirable. Here it is useful to impose any symmetries and additional knowledge regarding the structure of the problem at hand. This may allow us, for instance, to impose that certain monomials share the same coefficient and thus to reduce the number of independent variables.   

In cases where there are redundant variables involved, we can further realize this by inspecting the value functions. If we had for instance chosen an implementation of the problem that additionally involves the variables $P(0,1|0)$ and $P(0,1|1)$, we could find value functions that don't involve and of the monomials containing them. The other way around, we can check that we indeed need all variables used in the implementation of App.~\ref{app:implementation1}, more precisely, while we could remove e.g.\ $P(0,0|1)$ or $P(1,0|1)$ in the $d=2$, $L=3$ case, these are needed for higher $L$.

As the numerical optimization tends to assign non-zero values to all optimization variables, it is further useful to simplify the value functions by eliminating irrelevant monomials, meaning setting their coefficients to zero. It turns out that this works well for the problem at hand when following a heuristic procedure by first computing the optimum $P^{\rm d *}_{\rm max}$ and then iterating the following two steps:
\begin{itemize}
\item Fix a threshold $\nu$ and set all coefficients of the $V_k$ smaller than $\nu$ to zero for the next optimization. If there are no such coefficients, increase $\nu$.
\item Solve the problem of optimizing $P^{\rm d}_{\rm max}$ and check that the optimal value is still $P^{\rm d *}_{\rm max}$. If not, put coefficients back until recovering $P^{\rm d *}_{\rm max}$.
\end{itemize}

After this procedure the polynomials can then further be adjusted manually, e.g.\ by requiring that certain coefficients are equal.

\subsection{Entanglement detection in the automaton-guided GHZ game}\label{app:implementation2}

\subsubsection{Implementation of a polynomial optimization problem to bound the performance of separable states in the automaton-guided GHZ game}

The automaton-guided GHZ state is conducted using a 4-state automaton that is updated by measuring a qubit state in each round. 
Due to normalization, this means that we need polynomials in $3$ variables for the 4-state automaton, $\{p_k(t)\}_{t\in \{1,-1,i\}}$ , and $2$ variables for the Bloch vector of the state in each round, $n_1,n_2$, where $n=(n_1,n_2,n_3)$ denotes the Bloch vector of the state in round $k$, (except in the very last round, where we generate the reward as a function of $s_{N+1}$ in a way that updates this through multiplication with the success probability of the last measurement in each case. Due to the loss of normalization the last step thus needs one additional polynomial variable).

The value functions are constructed following an adaptation of Schm\"udgen's positivstellensatz in Eq.~\eqref{SOS_boost}. The polynomials $V_k$ are taken to be of mixed degree $2$, i.e., terms of order up to $p_k(t)p_k(t') n_s n_{s'}$, $t,t'\in  \{1,-1,i\}$, $s,s'\in \{1,2 \}$. This means that for an inequality involving $V_k$ and $V_{k+1}$ we consider polynomials of degree up to 2 in each of $\{p_k(t)\}_{t\in \{1,-1,i\}}$, $\{p_{k+1}(t)\}_{t\in \{1,-1,i\}}$, $n_1,\ n_2$, $n'_1,  \ n'_2$ (where $n=(n_1,n_2,n_3)$ denotes the Bloch vector of the state in round $k$ and  $n'=(n'_1,n'_2,n'_3$) is the Bloch vector in round $k+1$). This means that we consider monomials of the form $p_k(t)p_k(t') p_{k+1}(t'')p_{k+1}(t''') n_s n_{s'}n'_{s''} n'_{s'''}$. 
The polynomials $g_i$ are here given by positivity of the  probabilities of the automaton ($p_k(t=1) \geq 0, p_k(t=-1) \geq 0, p_k(t=i) \geq 0, 1-p_k(t=1)-p_k(t=-1)- p_k(t=i) \geq 0$, and the same for $p_{k+1}$) as well as normalization of the Bloch vectors of ($1-n_1^2-n_2^2 \geq 0$ and  $1-n_1'^2-n_2'^2 \geq 0$). The equation of motion (see (5) from the main text) for this example defines the $\tilde{g}_i$.

The equation of motion is in this case of degree $1$ in $p_k$, $p_{k+1}$, $\rho_k$, so that to exhaust the degree we have at our disposal we can multiply it with polynomials of degree $1$ in $p_k$, $p_{k+1}$, $\rho_k$ and degree $2$ in $\rho_{k+1}$. The squared polynomials multiplying the products $g=g_jg_kg_l...$ are chosen to be linear combination of at most degree $1$ in the variables that are not part of $g$. 

\subsubsection{Resolution of numerical problems in the optimization of the automaton-guided GHZ game}

We implemented the above problem using YALMIP~\cite{yalmip} and the solver MOSEK~\cite{mosek} for solving the semidefinite programs. 
To confirm that our results remain trustworthy for large $N$, we further extracted all value functions from the solution and then checked each inequality $V_N \geq r_N$, $V_k-V_{k+1} \ \forall k$, $\nu \geq V_1$ separately. For this purpose, we set up a constrained polynomial optimization problem (constrained by positivity constraints and equation of motion) to maximize the violation of each inequality separately. As opposed to the case of the original optimization we did not set up the semidefinite formulation manually but used the SOLVEMOMENT functionality of YALMIP (setting level=$3$ except for k=30 where we used level=$2$) and MOSEK for this step. The sum of the separate violations led to the error estimate for our results (see also Sec.~VC from the main text). 

For $N=35$ we obtained a value of $0.5025$, where the upper bound was computed following the procedure~(50) applied to the finite $N$ case with ${\nu}=0.5002$. 

{
\subsection{Entanglement detection by means of bit-sequence protocols}\label{app:implementation3}

The bit-sequence protocol is performed using a 2-state automaton that is updated by measuring a qubit state in each round. This automaton saves the outcome of the last measurement, based on which the next one is decided. 
Due to normalization, this means that we need polynomials in $1$ variable for the 2-state automaton, $\{p_k(t)\}_{t\in \{1\}}$, and $2$ variables for the Bloch vector of the state in each round, $n_1,n_3$, where $n=(n_1,n_2,n_3)$ denotes the Bloch vector of the state in round $k$~\footnote{For convenience, we choose the first and third entry here since we perform Pauli measurements and thus this way the equation of motion can be formulated more directly.}. 

The value functions are constructed following an adaptation of Schm\"udgen's positivstellensatz in Eq.~\eqref{SOS_boost}. The polynomials $V_k$ are taken to be of mixed degree $2$, i.e., terms of order up to $p_k(t)p_k(t') n_s n_{s'}$, $t,t' = 1$, $s,s'\in \{1,2 \}$. This means that for an inequality involving $V_k$ and $V_{k+1}$ we consider polynomials of degree up to 2 in each of $\{p_k(t)\}_{1}$, $\{p_{k+1}(t)\}_{1}$, $n_1,\ n_2$, $n'_1,  \ n'_2$ (where $n=(n_1,n_2,n_3)$ denotes the Bloch vector of the state in round $k$ and  $n'=(n'_1,n'_2,n'_3$) is the Bloch vector in round $k+1$) This means that we consider monomials of the form $p_k(t)p_k(t') p_{k+1}(t'')p_{k+1}(t''') n_s n_{s'}n'_{s''} n'_{s'''}$. 
The polynomials $g_i$ are here given by positivity of the  probabilities of the automaton ($p_k(t=1) \geq 0,  1-p_k(t=1) \geq 0$, and the same for $p_{k+1}$) as well as normalization of the Bloch vectors of ($1-n_1^2-n_2^2 \geq 0$ and  $1-n_1'^2-n_2'^2 \geq 0$). The equation of motion (see (5) from the main text) for this example defines the $\tilde{g}_i$.
}

\section{Lower bound on the maximum winning probability in the automaton-guided GHZ game with separable states}\label{app:lower_bound_GHZ}

The winning probability in the $N$-round GHZ game can be lower bounded by computing the score of any $N$-partite separable state. In the following we show that the optimal separable strategy consists in preparing the state $\ket{+}^{\otimes n}$, which indeed recovers our bounds (up to numerical precision). 

For this state, the measurement outcome for measurement $1$ is $1$ in each round, while for measurement $2$, both outcomes occur with probability $\frac{1}{2}$. This means that the states of the automaton in round $k$ have probabilities ${\bf p_k}=(p_k(1), p_k(-1), p_k(i), p_k(-i))=(1/4+1/2^n,1/4-1/2^n, 1/4,1/4)$, which is computed as ${\bf p_k}=M^{k-1} {\bf p_1}$, where ${\bf p_1}=(1,0,0,0)$ is the initial state of the sequential model and 
$$M=\left(\begin{array}{rrrr} 
\frac{1}{2} & 0 & \frac{1}{4} & \frac{1}{4} \\ 
0 & \frac{1}{2} & \frac{1}{4} & \frac{1}{4} \\
\frac{1}{4} & \frac{1}{4} & \frac{1}{2} & 0 \\
\frac{1}{4} & \frac{1}{4} & 0 & \frac{1}{2} \\
\end{array} \right).$$
The final winning probability is $p^{\rm sep}(\ket{+})=p_N(1)+ \frac{1}{2}(p_N(i)+p_N(-i))$.

\section{Computing the gradient under non-uniqueness of the optimal value functions}
\label{app:gradient}
The derivation of Eq.~(66) from the main text relies on the assumption that the optimal slack variables $V:=\{V_k\}_k$ for policy $x=\bar{x}$ are unique. In that case, it seemed natural to postulate that they are differentiable at $x=\bar{x}$, i.e., that the optimal slack variables of the perturbed problem, with policy $x=\bar{x}+\delta x$, should be of the form of (60) from the main text.

If there exists more than one minimizer for problem (59) from the main text, though, then it is still natural to postulate that there exists at least \emph{one} differentiable minimizer. However, there is no reason why the value of this minimizer should coincide with the value returned by the numerical solver. In other words, given a minimizer $V$ of (59), the optimization problem (66) in general returns an upper bound on $D^{+}\nu(\bar{x})$, rather than its exact value.
\begin{widetext}
To tackle the possible non-uniqueness of the optimal slack variables $V$, consider, for fixed $\xi$, the following variant of problem (66):
\begin{align}
    &\min_{\mu,V,\hat{V}} \mu\nonumber\\
    \mbox{such that }&\hat{V}_N(s_k,\lambda)\geq_Q \frac{\partial}{\partial x} r_N(s_N,h_N,\lambda,\bar{x})+\xi_N\left(V_N(s_k,\lambda)-r_N(s_N,h_N,\lambda,\bar{x})\right),\nonumber\\
    &\hat{V}_k(s_k,\lambda)\geq_Q \frac{\partial}{\partial x} r_k(s_k,h,\lambda,\bar{x})+\hat{V}_{k+1}(s_{k+1},\lambda)+\sum_{\sigma}\frac{\partial V_{k+1}(s_{k+1},\lambda)}{\partial s_{k+1}^\sigma}\frac{\partial f^\sigma_k(s_k,h_k,\lambda,x)}{\partial x}+\nonumber\\
    &+\xi_k\left(V_k(s_k,\lambda)-r_k(s_k,h_k,\lambda,\bar{x})-V_{k+1}(s_{k+1},\lambda)\right),\nonumber\\
    &\mu\geq_Q \hat{V}_1(s_1,\lambda)+\xi_1\left(\nu-V_1(s_1,\lambda)\right),\nonumber\\
    &V_{N}(s_N,\lambda)\geq_Q r_N(s_N,h,\lambda, \bar{x}),\nonumber\\
    &V_k(s_k,\lambda)\geq_Q r_k(s_k,h,\lambda,\bar{x})+ V_{k+1}(f_{k}(s_k,h_k,\lambda),\lambda),\nonumber\\
    &\nu(\bar{x})\geq_Q V_1(s_1,\lambda).
    \label{perturbed2}
\end{align}
\noindent This problem optimizes, not just over $\hat{V}$, but also over the minimizers $V$ of problem (59).

\end{widetext}
A coordinate descent-based heuristic to compute the limit $D^{+}\nu(\bar{x})$ is thus the following: given a minimizer $V^{(j)}$ of the unperturbed problem, we solve problem (66), hence obtaining the slack variables $\xi^{(j)}$. Next, we solve problem in Eq.~(\ref{perturbed2}) for $\xi=\xi^{(j)}$, obtaining the new unperturbed solution $V^{(j+1)}$. Starting from an initial solution of the unperturbed problem $V^{(0)}$, we hence generate a sequence of minimizers $V^{(1)}, V^{(2)},...$ for problem (59). The corresponding sequence of solutions of problem (66), with $V=V^{(1)}, V^{(2)},...$, is obviously decreasing, and, if the initial seed $V^{(0)}$ is close enough to the optimal solution, one would expect it to converge to $D^{+}\nu(\bar{x})$. 

Since the estimation of $D^{+}\nu(\bar{x})$ requires solving a non-convex optimization problem, it is possible that the solution reached through coordinate descent is not optimal. A simple consistency check to verify that coordinate descent did not get stuck in a local minimum consists in computing the converse limit $D^{-}\nu(\bar{x}):= \lim_{\delta_{x}\to 0^{+}}\frac{\nu(\bar{x})-\nu(\bar{x}-\delta \ket{i})}{\delta x}$ and see if both quantities coincide (they should, if $\nu(x)$ is differentiable in $x=\bar{x}$). It is not difficult to adapt the derivation of Eqs. (66), (\ref{perturbed2}) to estimate $D^{-}\nu(\bar{x})$ instead of $D^{+}\nu(\bar{x})$.

\bibliography{seq_references}

\begin{thebibliography}{81}%
\makeatletter
\providecommand \@ifxundefined [1]{%
 \@ifx{#1\undefined}
}%
\providecommand \@ifnum [1]{%
 \ifnum #1\expandafter \@firstoftwo
 \else \expandafter \@secondoftwo
 \fi
}%
\providecommand \@ifx [1]{%
 \ifx #1\expandafter \@firstoftwo
 \else \expandafter \@secondoftwo
 \fi
}%
\providecommand \natexlab [1]{#1}%
\providecommand \enquote  [1]{``#1''}%
\providecommand \bibnamefont  [1]{#1}%
\providecommand \bibfnamefont [1]{#1}%
\providecommand \citenamefont [1]{#1}%
\providecommand \href@noop [0]{\@secondoftwo}%
\providecommand \href [0]{\begingroup \@sanitize@url \@href}%
\providecommand \@href[1]{\@@startlink{#1}\@@href}%
\providecommand \@@href[1]{\endgroup#1\@@endlink}%
\providecommand \@sanitize@url [0]{\catcode `\\12\catcode `\$12\catcode
  `\&12\catcode `\#12\catcode `\^12\catcode `\_12\catcode `\%12\relax}%
\providecommand \@@startlink[1]{}%
\providecommand \@@endlink[0]{}%
\providecommand \url  [0]{\begingroup\@sanitize@url \@url }%
\providecommand \@url [1]{\endgroup\@href {#1}{\urlprefix }}%
\providecommand \urlprefix  [0]{URL }%
\providecommand \Eprint [0]{\href }%
\providecommand \doibase [0]{https://doi.org/}%
\providecommand \selectlanguage [0]{\@gobble}%
\providecommand \bibinfo  [0]{\@secondoftwo}%
\providecommand \bibfield  [0]{\@secondoftwo}%
\providecommand \translation [1]{[#1]}%
\providecommand \BibitemOpen [0]{}%
\providecommand \bibitemStop [0]{}%
\providecommand \bibitemNoStop [0]{.\EOS\space}%
\providecommand \EOS [0]{\spacefactor3000\relax}%
\providecommand \BibitemShut  [1]{\csname bibitem#1\endcsname}%
\let\auto@bib@innerbib\@empty
\bibitem [{\citenamefont {Wallraff}\ \emph {et~al.}(2005)\citenamefont
  {Wallraff}, \citenamefont {Schuster}, \citenamefont {Blais}, \citenamefont
  {Frunzio}, \citenamefont {Majer}, \citenamefont {Devoret}, \citenamefont
  {Girvin},\ and\ \citenamefont {Schoelkopf}}]{superconducting}%
  \BibitemOpen
  \bibfield  {author} {\bibinfo {author} {\bibfnamefont {A.}~\bibnamefont
  {Wallraff}}, \bibinfo {author} {\bibfnamefont {D.~I.}\ \bibnamefont
  {Schuster}}, \bibinfo {author} {\bibfnamefont {A.}~\bibnamefont {Blais}},
  \bibinfo {author} {\bibfnamefont {L.}~\bibnamefont {Frunzio}}, \bibinfo
  {author} {\bibfnamefont {J.}~\bibnamefont {Majer}}, \bibinfo {author}
  {\bibfnamefont {M.~H.}\ \bibnamefont {Devoret}}, \bibinfo {author}
  {\bibfnamefont {S.~M.}\ \bibnamefont {Girvin}},\ and\ \bibinfo {author}
  {\bibfnamefont {R.~J.}\ \bibnamefont {Schoelkopf}},\ }\bibfield  {title}
  {\bibinfo {title} {Approaching unit visibility for control of a
  superconducting qubit with dispersive readout},\ }\href
  {https://doi.org/10.1103/PhysRevLett.95.060501} {\bibfield  {journal}
  {\bibinfo  {journal} {Phys. Rev. Lett.}\ }\textbf {\bibinfo {volume} {95}},\
  \bibinfo {pages} {060501} (\bibinfo {year} {2005})}\BibitemShut {NoStop}%
\bibitem [{\citenamefont {Weilenmann}\ \emph {et~al.}(2021)\citenamefont
  {Weilenmann}, \citenamefont {Aguilar},\ and\ \citenamefont
  {Navascués}}]{prep_games}%
  \BibitemOpen
  \bibfield  {author} {\bibinfo {author} {\bibfnamefont {M.}~\bibnamefont
  {Weilenmann}}, \bibinfo {author} {\bibfnamefont {E.~A.}\ \bibnamefont
  {Aguilar}},\ and\ \bibinfo {author} {\bibfnamefont {M.}~\bibnamefont
  {Navascués}},\ }\bibfield  {title} {\bibinfo {title} {Analysis and
  optimization of quantum adaptive measurement protocols with the framework of
  preparation games},\ }\bibfield  {journal} {\bibinfo  {journal} {Nature
  Communications}\ }\textbf {\bibinfo {volume} {12}},\ \href
  {https://doi.org/10.1038/s41467-021-24658-9} {10.1038/s41467-021-24658-9}
  (\bibinfo {year} {2021})\BibitemShut {NoStop}%
\bibitem [{\citenamefont {Geller}\ \emph {et~al.}(2022)\citenamefont {Geller},
  \citenamefont {Cole}, \citenamefont {Glancy},\ and\ \citenamefont
  {Knill}}]{Knill}%
  \BibitemOpen
  \bibfield  {author} {\bibinfo {author} {\bibfnamefont {S.}~\bibnamefont
  {Geller}}, \bibinfo {author} {\bibfnamefont {D.~C.}\ \bibnamefont {Cole}},
  \bibinfo {author} {\bibfnamefont {S.}~\bibnamefont {Glancy}},\ and\ \bibinfo
  {author} {\bibfnamefont {E.}~\bibnamefont {Knill}},\ }\href
  {https://doi.org/10.48550/ARXIV.2204.00710} {\bibinfo {title} {Improving
  quantum state detection with adaptive sequential observations}} (\bibinfo
  {year} {2022})\BibitemShut {NoStop}%
\bibitem [{\citenamefont {Navascu\'es}\ \emph {et~al.}(2007)\citenamefont
  {Navascu\'es}, \citenamefont {Pironio},\ and\ \citenamefont
  {Ac\'{i}n}}]{npa}%
  \BibitemOpen
  \bibfield  {author} {\bibinfo {author} {\bibfnamefont {M.}~\bibnamefont
  {Navascu\'es}}, \bibinfo {author} {\bibfnamefont {S.}~\bibnamefont
  {Pironio}},\ and\ \bibinfo {author} {\bibfnamefont {A.}~\bibnamefont
  {Ac\'{i}n}},\ }\bibfield  {title} {\bibinfo {title} {Bounding the set of
  quantum correlations},\ }\href@noop {} {\bibfield  {journal} {\bibinfo
  {journal} {Phys. Rev. Lett.}\ }\textbf {\bibinfo {volume} {98}},\ \bibinfo
  {pages} {010401} (\bibinfo {year} {2007})}\BibitemShut {NoStop}%
\bibitem [{\citenamefont {Navascu\'es}\ \emph {et~al.}(2008)\citenamefont
  {Navascu\'es}, \citenamefont {Pironio},\ and\ \citenamefont
  {Ac\'{i}n}}]{npa2}%
  \BibitemOpen
  \bibfield  {author} {\bibinfo {author} {\bibfnamefont {M.}~\bibnamefont
  {Navascu\'es}}, \bibinfo {author} {\bibfnamefont {S.}~\bibnamefont
  {Pironio}},\ and\ \bibinfo {author} {\bibfnamefont {A.}~\bibnamefont
  {Ac\'{i}n}},\ }\bibfield  {title} {\bibinfo {title} {A convergent hierarchy
  of semidefinite programs characterizing the set of quantum correlations},\
  }\href@noop {} {\bibfield  {journal} {\bibinfo  {journal} {New J. Phys.}\
  }\textbf {\bibinfo {volume} {10}},\ \bibinfo {pages} {073013} (\bibinfo
  {year} {2008})}\BibitemShut {NoStop}%
\bibitem [{\citenamefont {Moroder}\ \emph {et~al.}(2013)\citenamefont
  {Moroder}, \citenamefont {Bancal}, \citenamefont {Liang}, \citenamefont
  {Hofmann},\ and\ \citenamefont {G\"uhne}}]{Moroder}%
  \BibitemOpen
  \bibfield  {author} {\bibinfo {author} {\bibfnamefont {T.}~\bibnamefont
  {Moroder}}, \bibinfo {author} {\bibfnamefont {J.-D.}\ \bibnamefont {Bancal}},
  \bibinfo {author} {\bibfnamefont {Y.-C.}\ \bibnamefont {Liang}}, \bibinfo
  {author} {\bibfnamefont {M.}~\bibnamefont {Hofmann}},\ and\ \bibinfo {author}
  {\bibfnamefont {O.}~\bibnamefont {G\"uhne}},\ }\bibfield  {title} {\bibinfo
  {title} {Device-independent entanglement quantification and related
  applications},\ }\href {https://doi.org/10.1103/PhysRevLett.111.030501}
  {\bibfield  {journal} {\bibinfo  {journal} {Phys. Rev. Lett.}\ }\textbf
  {\bibinfo {volume} {111}},\ \bibinfo {pages} {030501} (\bibinfo {year}
  {2013})}\BibitemShut {NoStop}%
\bibitem [{\citenamefont {Saggio}\ \emph {et~al.}(2019)\citenamefont {Saggio},
  \citenamefont {Dimi{\'c}}, \citenamefont {Greganti}, \citenamefont {Rozema},
  \citenamefont {Walther},\ and\ \citenamefont {Daki{\'c}}}]{Saggio2019}%
  \BibitemOpen
  \bibfield  {author} {\bibinfo {author} {\bibfnamefont {V.}~\bibnamefont
  {Saggio}}, \bibinfo {author} {\bibfnamefont {A.}~\bibnamefont {Dimi{\'c}}},
  \bibinfo {author} {\bibfnamefont {C.}~\bibnamefont {Greganti}}, \bibinfo
  {author} {\bibfnamefont {L.~A.}\ \bibnamefont {Rozema}}, \bibinfo {author}
  {\bibfnamefont {P.}~\bibnamefont {Walther}},\ and\ \bibinfo {author}
  {\bibfnamefont {B.}~\bibnamefont {Daki{\'c}}},\ }\bibfield  {title} {\bibinfo
  {title} {Experimental few-copy multipartite entanglement detection},\
  }\href@noop {} {\bibfield  {journal} {\bibinfo  {journal} {Nature physics}\
  }\textbf {\bibinfo {volume} {15}},\ \bibinfo {pages} {935} (\bibinfo {year}
  {2019})}\BibitemShut {NoStop}%
\bibitem [{\citenamefont {Paz}(1971)}]{PazBook}%
  \BibitemOpen
  \bibfield  {author} {\bibinfo {author} {\bibfnamefont {A.}~\bibnamefont
  {Paz}},\ }\href@noop {} {\emph {\bibinfo {title} {Introduction to
  probabilistic automata}}}\ (\bibinfo  {publisher} {Academic Press},\ \bibinfo
  {year} {1971})\BibitemShut {NoStop}%
\bibitem [{\citenamefont {Rabin}(1963)}]{Rabin1963}%
  \BibitemOpen
  \bibfield  {author} {\bibinfo {author} {\bibfnamefont {M.~O.}\ \bibnamefont
  {Rabin}},\ }\bibfield  {title} {\bibinfo {title} {Probabilistic automata},\
  }\href {https://doi.org/10.1016/S0019-9958(63)90290-0} {\bibfield  {journal}
  {\bibinfo  {journal} {{Information} and {Control}}\ }\textbf {\bibinfo
  {volume} {6}},\ \bibinfo {pages} {230} (\bibinfo {year} {1963})}\BibitemShut
  {NoStop}%
\bibitem [{\citenamefont {Kleinmann}\ \emph {et~al.}(2011)\citenamefont
  {Kleinmann}, \citenamefont {Gühne}, \citenamefont {Portillo}, \citenamefont
  {Åke Larsson},\ and\ \citenamefont {Cabello}}]{KleinmannNJP2011}%
  \BibitemOpen
  \bibfield  {author} {\bibinfo {author} {\bibfnamefont {M.}~\bibnamefont
  {Kleinmann}}, \bibinfo {author} {\bibfnamefont {O.}~\bibnamefont {Gühne}},
  \bibinfo {author} {\bibfnamefont {J.~R.}\ \bibnamefont {Portillo}}, \bibinfo
  {author} {\bibfnamefont {J.}~\bibnamefont {Åke Larsson}},\ and\ \bibinfo
  {author} {\bibfnamefont {A.}~\bibnamefont {Cabello}},\ }\bibfield  {title}
  {\bibinfo {title} {Memory cost of quantum contextuality},\ }\href
  {http://stacks.iop.org/1367-2630/13/i=11/a=113011} {\bibfield  {journal}
  {\bibinfo  {journal} {New J. Phys.}\ }\textbf {\bibinfo {volume} {13}},\
  \bibinfo {pages} {113011} (\bibinfo {year} {2011})}\BibitemShut {NoStop}%
\bibitem [{\citenamefont {Fagundes}\ and\ \citenamefont
  {Kleinmann}(2017)}]{Fagundes2017}%
  \BibitemOpen
  \bibfield  {author} {\bibinfo {author} {\bibfnamefont {G.}~\bibnamefont
  {Fagundes}}\ and\ \bibinfo {author} {\bibfnamefont {M.}~\bibnamefont
  {Kleinmann}},\ }\bibfield  {title} {\bibinfo {title} {Memory cost for
  simulating all quantum correlations from the peres–mermin scenario},\
  }\href {http://stacks.iop.org/1751-8121/50/i=32/a=325302} {\bibfield
  {journal} {\bibinfo  {journal} {J. Phys. A}\ }\textbf {\bibinfo {volume}
  {50}},\ \bibinfo {pages} {325302} (\bibinfo {year} {2017})}\BibitemShut
  {NoStop}%
\bibitem [{\citenamefont {Budroni}\ \emph {et~al.}(2019)\citenamefont
  {Budroni}, \citenamefont {Fagundes},\ and\ \citenamefont
  {Kleinmann}}]{BudroniNJP2019}%
  \BibitemOpen
  \bibfield  {author} {\bibinfo {author} {\bibfnamefont {C.}~\bibnamefont
  {Budroni}}, \bibinfo {author} {\bibfnamefont {G.}~\bibnamefont {Fagundes}},\
  and\ \bibinfo {author} {\bibfnamefont {M.}~\bibnamefont {Kleinmann}},\
  }\bibfield  {title} {\bibinfo {title} {Memory cost of temporal
  correlations},\ }\href {https://doi.org/10.1088/1367-2630/ab3cb4} {\bibfield
  {journal} {\bibinfo  {journal} {New J. Phys.}\ }\textbf {\bibinfo {volume}
  {21}},\ \bibinfo {pages} {093018} (\bibinfo {year} {2019})}\BibitemShut
  {NoStop}%
\bibitem [{\citenamefont {Budroni}\ \emph {et~al.}(2022)\citenamefont
  {Budroni}, \citenamefont {Cabello}, \citenamefont {G\"uhne}, \citenamefont
  {Kleinmann},\ and\ \citenamefont {Larsson}}]{Context_review}%
  \BibitemOpen
  \bibfield  {author} {\bibinfo {author} {\bibfnamefont {C.}~\bibnamefont
  {Budroni}}, \bibinfo {author} {\bibfnamefont {A.}~\bibnamefont {Cabello}},
  \bibinfo {author} {\bibfnamefont {O.}~\bibnamefont {G\"uhne}}, \bibinfo
  {author} {\bibfnamefont {M.}~\bibnamefont {Kleinmann}},\ and\ \bibinfo
  {author} {\bibfnamefont {J.-A.}\ \bibnamefont {Larsson}},\ }\bibfield
  {title} {\bibinfo {title} {Kochen-specker contextuality},\ }\href
  {https://doi.org/10.1103/RevModPhys.94.045007} {\bibfield  {journal}
  {\bibinfo  {journal} {Rev. Mod. Phys.}\ }\textbf {\bibinfo {volume} {94}},\
  \bibinfo {pages} {045007} (\bibinfo {year} {2022})}\BibitemShut {NoStop}%
\bibitem [{\citenamefont {Garner}\ \emph {et~al.}(2017)\citenamefont {Garner},
  \citenamefont {Liu}, \citenamefont {Thompson}, \citenamefont {Vedral},\ and\
  \citenamefont {Gu}}]{GarnerNJP2017}%
  \BibitemOpen
  \bibfield  {author} {\bibinfo {author} {\bibfnamefont {A.~J.~P.}\
  \bibnamefont {Garner}}, \bibinfo {author} {\bibfnamefont {Q.}~\bibnamefont
  {Liu}}, \bibinfo {author} {\bibfnamefont {J.}~\bibnamefont {Thompson}},
  \bibinfo {author} {\bibfnamefont {V.}~\bibnamefont {Vedral}},\ and\ \bibinfo
  {author} {\bibfnamefont {M.}~\bibnamefont {Gu}},\ }\bibfield  {title}
  {\bibinfo {title} {Provably unbounded memory advantage in stochastic
  simulation using quantum mechanics},\ }\href
  {https://doi.org/10.1088/1367-2630/aa82df} {\bibfield  {journal} {\bibinfo
  {journal} {New J. Phys.}\ }\textbf {\bibinfo {volume} {19}},\ \bibinfo
  {pages} {103009} (\bibinfo {year} {2017})}\BibitemShut {NoStop}%
\bibitem [{\citenamefont {Elliott}\ and\ \citenamefont
  {Gu}(2018)}]{Elliott2018}%
  \BibitemOpen
  \bibfield  {author} {\bibinfo {author} {\bibfnamefont {T.~J.}\ \bibnamefont
  {Elliott}}\ and\ \bibinfo {author} {\bibfnamefont {M.}~\bibnamefont {Gu}},\
  }\bibfield  {title} {\bibinfo {title} {Superior memory efficiency of quantum
  devices for the simulation of continuous-time stochastic processes},\ }\href
  {https://doi.org/10.1038/s41534-018-0064-4} {\bibfield  {journal} {\bibinfo
  {journal} {npj Quantum Information}\ }\textbf {\bibinfo {volume} {4}},\
  \bibinfo {pages} {1} (\bibinfo {year} {2018})}\BibitemShut {NoStop}%
\bibitem [{\citenamefont {Elliott}\ \emph {et~al.}(2020)\citenamefont
  {Elliott}, \citenamefont {Yang}, \citenamefont {Binder}, \citenamefont
  {Garner}, \citenamefont {Thompson},\ and\ \citenamefont {Gu}}]{Elliott2019}%
  \BibitemOpen
  \bibfield  {author} {\bibinfo {author} {\bibfnamefont {T.~J.}\ \bibnamefont
  {Elliott}}, \bibinfo {author} {\bibfnamefont {C.}~\bibnamefont {Yang}},
  \bibinfo {author} {\bibfnamefont {F.~C.}\ \bibnamefont {Binder}}, \bibinfo
  {author} {\bibfnamefont {A.~J.~P.}\ \bibnamefont {Garner}}, \bibinfo {author}
  {\bibfnamefont {J.}~\bibnamefont {Thompson}},\ and\ \bibinfo {author}
  {\bibfnamefont {M.}~\bibnamefont {Gu}},\ }\bibfield  {title} {\bibinfo
  {title} {Extreme dimensionality reduction with quantum modeling},\ }\href
  {https://doi.org/10.1103/PhysRevLett.125.260501} {\bibfield  {journal}
  {\bibinfo  {journal} {Phys. Rev. Lett.}\ }\textbf {\bibinfo {volume} {125}},\
  \bibinfo {pages} {260501} (\bibinfo {year} {2020})}\BibitemShut {NoStop}%
\bibitem [{\citenamefont {Spee}(2020)}]{SpeePRA2019}%
  \BibitemOpen
  \bibfield  {author} {\bibinfo {author} {\bibfnamefont {C.}~\bibnamefont
  {Spee}},\ }\bibfield  {title} {\bibinfo {title} {Certifying the purity of
  quantum states with temporal correlations},\ }\href
  {https://doi.org/10.1103/PhysRevA.102.012420} {\bibfield  {journal} {\bibinfo
   {journal} {Phys. Rev. A}\ }\textbf {\bibinfo {volume} {102}},\ \bibinfo
  {pages} {012420} (\bibinfo {year} {2020})}\BibitemShut {NoStop}%
\bibitem [{\citenamefont {Hoffmann}\ \emph {et~al.}(2018)\citenamefont
  {Hoffmann}, \citenamefont {Spee}, \citenamefont {Gühne},\ and\ \citenamefont
  {Budroni}}]{HoffmannNJP2018}%
  \BibitemOpen
  \bibfield  {author} {\bibinfo {author} {\bibfnamefont {J.}~\bibnamefont
  {Hoffmann}}, \bibinfo {author} {\bibfnamefont {C.}~\bibnamefont {Spee}},
  \bibinfo {author} {\bibfnamefont {O.}~\bibnamefont {Gühne}},\ and\ \bibinfo
  {author} {\bibfnamefont {C.}~\bibnamefont {Budroni}},\ }\bibfield  {title}
  {\bibinfo {title} {Structure of temporal correlations of a qubit},\ }\href
  {https://doi.org/10.1088/1367-2630/aae87f} {\bibfield  {journal} {\bibinfo
  {journal} {New Journal of Physics}\ }\textbf {\bibinfo {volume} {20}},\
  \bibinfo {pages} {102001} (\bibinfo {year} {2018})}\BibitemShut {NoStop}%
\bibitem [{\citenamefont {Spee}\ \emph
  {et~al.}(2020{\natexlab{a}})\citenamefont {Spee}, \citenamefont {Siebeneich},
  \citenamefont {Gloger}, \citenamefont {Kaufmann}, \citenamefont {Johanning},
  \citenamefont {Kleinmann}, \citenamefont {Wunderlich},\ and\ \citenamefont
  {Gühne}}]{SpeeNJP2020}%
  \BibitemOpen
  \bibfield  {author} {\bibinfo {author} {\bibfnamefont {C.}~\bibnamefont
  {Spee}}, \bibinfo {author} {\bibfnamefont {H.}~\bibnamefont {Siebeneich}},
  \bibinfo {author} {\bibfnamefont {T.~F.}\ \bibnamefont {Gloger}}, \bibinfo
  {author} {\bibfnamefont {P.}~\bibnamefont {Kaufmann}}, \bibinfo {author}
  {\bibfnamefont {M.}~\bibnamefont {Johanning}}, \bibinfo {author}
  {\bibfnamefont {M.}~\bibnamefont {Kleinmann}}, \bibinfo {author}
  {\bibfnamefont {C.}~\bibnamefont {Wunderlich}},\ and\ \bibinfo {author}
  {\bibfnamefont {O.}~\bibnamefont {Gühne}},\ }\bibfield  {title} {\bibinfo
  {title} {Genuine temporal correlations can certify the quantum dimension},\
  }\href {https://doi.org/10.1088/1367-2630/ab6d42} {\bibfield  {journal}
  {\bibinfo  {journal} {New Journal of Physics}\ }\textbf {\bibinfo {volume}
  {22}},\ \bibinfo {pages} {023028} (\bibinfo {year}
  {2020}{\natexlab{a}})}\BibitemShut {NoStop}%
\bibitem [{\citenamefont {Spee}\ \emph
  {et~al.}(2020{\natexlab{b}})\citenamefont {Spee}, \citenamefont {Budroni},\
  and\ \citenamefont {Gühne}}]{SpeeNJP2020_a}%
  \BibitemOpen
  \bibfield  {author} {\bibinfo {author} {\bibfnamefont {C.}~\bibnamefont
  {Spee}}, \bibinfo {author} {\bibfnamefont {C.}~\bibnamefont {Budroni}},\ and\
  \bibinfo {author} {\bibfnamefont {O.}~\bibnamefont {Gühne}},\ }\bibfield
  {title} {\bibinfo {title} {Simulating extremal temporal correlations},\
  }\href {https://doi.org/10.1088/1367-2630/abb899} {\bibfield  {journal}
  {\bibinfo  {journal} {New Journal of Physics}\ }\textbf {\bibinfo {volume}
  {22}},\ \bibinfo {pages} {103037} (\bibinfo {year}
  {2020}{\natexlab{b}})}\BibitemShut {NoStop}%
\bibitem [{\citenamefont {Erker}\ \emph {et~al.}(2017)\citenamefont {Erker},
  \citenamefont {Mitchison}, \citenamefont {Silva}, \citenamefont {Woods},
  \citenamefont {Brunner},\ and\ \citenamefont {Huber}}]{ErkerPRX2017}%
  \BibitemOpen
  \bibfield  {author} {\bibinfo {author} {\bibfnamefont {P.}~\bibnamefont
  {Erker}}, \bibinfo {author} {\bibfnamefont {M.~T.}\ \bibnamefont
  {Mitchison}}, \bibinfo {author} {\bibfnamefont {R.}~\bibnamefont {Silva}},
  \bibinfo {author} {\bibfnamefont {M.~P.}\ \bibnamefont {Woods}}, \bibinfo
  {author} {\bibfnamefont {N.}~\bibnamefont {Brunner}},\ and\ \bibinfo {author}
  {\bibfnamefont {M.}~\bibnamefont {Huber}},\ }\bibfield  {title} {\bibinfo
  {title} {Autonomous quantum clocks: Does thermodynamics limit our ability to
  measure time?},\ }\href {https://doi.org/10.1103/PhysRevX.7.031022}
  {\bibfield  {journal} {\bibinfo  {journal} {Phys. Rev. X}\ }\textbf {\bibinfo
  {volume} {7}},\ \bibinfo {pages} {031022} (\bibinfo {year}
  {2017})}\BibitemShut {NoStop}%
\bibitem [{\citenamefont {Woods}\ \emph {et~al.}(2022)\citenamefont {Woods},
  \citenamefont {Silva}, \citenamefont {P\"utz}, \citenamefont {Stupar},\ and\
  \citenamefont {Renner}}]{Woods2022}%
  \BibitemOpen
  \bibfield  {author} {\bibinfo {author} {\bibfnamefont {M.~P.}\ \bibnamefont
  {Woods}}, \bibinfo {author} {\bibfnamefont {R.}~\bibnamefont {Silva}},
  \bibinfo {author} {\bibfnamefont {G.}~\bibnamefont {P\"utz}}, \bibinfo
  {author} {\bibfnamefont {S.}~\bibnamefont {Stupar}},\ and\ \bibinfo {author}
  {\bibfnamefont {R.}~\bibnamefont {Renner}},\ }\bibfield  {title} {\bibinfo
  {title} {Quantum clocks are more accurate than classical ones},\ }\href
  {https://doi.org/10.1103/PRXQuantum.3.010319} {\bibfield  {journal} {\bibinfo
   {journal} {PRX Quantum}\ }\textbf {\bibinfo {volume} {3}},\ \bibinfo {pages}
  {010319} (\bibinfo {year} {2022})}\BibitemShut {NoStop}%
\bibitem [{\citenamefont {Budroni}\ \emph {et~al.}(2021)\citenamefont
  {Budroni}, \citenamefont {Vitagliano},\ and\ \citenamefont
  {Woods}}]{tick_sequence_paper}%
  \BibitemOpen
  \bibfield  {author} {\bibinfo {author} {\bibfnamefont {C.}~\bibnamefont
  {Budroni}}, \bibinfo {author} {\bibfnamefont {G.}~\bibnamefont
  {Vitagliano}},\ and\ \bibinfo {author} {\bibfnamefont {M.~P.}\ \bibnamefont
  {Woods}},\ }\bibfield  {title} {\bibinfo {title} {Ticking-clock performance
  enhanced by nonclassical temporal correlations},\ }\href
  {https://doi.org/10.1103/PhysRevResearch.3.033051} {\bibfield  {journal}
  {\bibinfo  {journal} {Phys. Rev. Research}\ }\textbf {\bibinfo {volume}
  {3}},\ \bibinfo {pages} {033051} (\bibinfo {year} {2021})}\BibitemShut
  {NoStop}%
\bibitem [{\citenamefont {Vieira}\ and\ \citenamefont
  {Budroni}(2022)}]{Vieira_Budroni}%
  \BibitemOpen
  \bibfield  {author} {\bibinfo {author} {\bibfnamefont {L.~B.}\ \bibnamefont
  {Vieira}}\ and\ \bibinfo {author} {\bibfnamefont {C.}~\bibnamefont
  {Budroni}},\ }\bibfield  {title} {\bibinfo {title} {Temporal correlations in
  the simplest measurement sequences},\ }\href
  {https://doi.org/10.22331/q-2022-01-18-623} {\bibfield  {journal} {\bibinfo
  {journal} {{Quantum}}\ }\textbf {\bibinfo {volume} {6}},\ \bibinfo {pages}
  {623} (\bibinfo {year} {2022})}\BibitemShut {NoStop}%
\bibitem [{\citenamefont {Zurel}\ \emph {et~al.}(2020)\citenamefont {Zurel},
  \citenamefont {Okay},\ and\ \citenamefont {Raussendorf}}]{Zurel:2020PRL}%
  \BibitemOpen
  \bibfield  {author} {\bibinfo {author} {\bibfnamefont {M.}~\bibnamefont
  {Zurel}}, \bibinfo {author} {\bibfnamefont {C.}~\bibnamefont {Okay}},\ and\
  \bibinfo {author} {\bibfnamefont {R.}~\bibnamefont {Raussendorf}},\
  }\bibfield  {title} {\bibinfo {title} {Hidden variable model for universal
  quantum computation with magic states on qubits},\ }\href
  {https://doi.org/10.1103/PhysRevLett.125.260404} {\bibfield  {journal}
  {\bibinfo  {journal} {Phys. Rev. Lett.}\ }\textbf {\bibinfo {volume} {125}},\
  \bibinfo {pages} {260404} (\bibinfo {year} {2020})}\BibitemShut {NoStop}%
\bibitem [{\citenamefont {Vitagliano}\ and\ \citenamefont
  {Budroni}(2023)}]{Vitagliano2022}%
  \BibitemOpen
  \bibfield  {author} {\bibinfo {author} {\bibfnamefont {G.}~\bibnamefont
  {Vitagliano}}\ and\ \bibinfo {author} {\bibfnamefont {C.}~\bibnamefont
  {Budroni}},\ }\bibfield  {title} {\bibinfo {title} {Leggett-{G}arg
  macrorealism and temporal correlations},\ }\href
  {https://doi.org/10.1103/PhysRevA.107.040101} {\bibfield  {journal} {\bibinfo
   {journal} {Phys. Rev. A}\ }\textbf {\bibinfo {volume} {107}},\ \bibinfo
  {pages} {040101} (\bibinfo {year} {2023})}\BibitemShut {NoStop}%
\bibitem [{\citenamefont {Brandao}\ \emph {et~al.}(2014)\citenamefont
  {Brandao}, \citenamefont {Harrow}, \citenamefont {Lee},\ and\ \citenamefont
  {Peres}}]{asymptotics1}%
  \BibitemOpen
  \bibfield  {author} {\bibinfo {author} {\bibfnamefont {F.~G.}\ \bibnamefont
  {Brandao}}, \bibinfo {author} {\bibfnamefont {A.~W.}\ \bibnamefont {Harrow}},
  \bibinfo {author} {\bibfnamefont {J.~R.}\ \bibnamefont {Lee}},\ and\ \bibinfo
  {author} {\bibfnamefont {Y.}~\bibnamefont {Peres}},\ }\bibfield  {title}
  {\bibinfo {title} {Adversarial hypothesis testing and a quantum stein's lemma
  for restricted measurements},\ }in\ \href@noop {} {\emph {\bibinfo
  {booktitle} {Proceedings of the 5th conference on Innovations in theoretical
  computer science}}}\ (\bibinfo {year} {2014})\ pp.\ \bibinfo {pages}
  {183--194}\BibitemShut {NoStop}%
\bibitem [{\citenamefont {Li}\ \emph {et~al.}(2022)\citenamefont {Li},
  \citenamefont {Tan},\ and\ \citenamefont {Tomamichel}}]{asymptotics2}%
  \BibitemOpen
  \bibfield  {author} {\bibinfo {author} {\bibfnamefont {Y.}~\bibnamefont
  {Li}}, \bibinfo {author} {\bibfnamefont {V.~Y.}\ \bibnamefont {Tan}},\ and\
  \bibinfo {author} {\bibfnamefont {M.}~\bibnamefont {Tomamichel}},\ }\bibfield
   {title} {\bibinfo {title} {Optimal adaptive strategies for sequential
  quantum hypothesis testing},\ }\href@noop {} {\bibfield  {journal} {\bibinfo
  {journal} {Communications in Mathematical Physics}\ }\textbf {\bibinfo
  {volume} {392}},\ \bibinfo {pages} {993} (\bibinfo {year}
  {2022})}\BibitemShut {NoStop}%
\bibitem [{\citenamefont {Gimbert}\ and\ \citenamefont
  {Oualhadj}(2010)}]{gimbert2010}%
  \BibitemOpen
  \bibfield  {author} {\bibinfo {author} {\bibfnamefont {H.}~\bibnamefont
  {Gimbert}}\ and\ \bibinfo {author} {\bibfnamefont {Y.}~\bibnamefont
  {Oualhadj}},\ }\bibfield  {title} {\bibinfo {title} {Probabilistic automata
  on finite words: Decidable and undecidable problems},\ }in\ \href@noop {}
  {\emph {\bibinfo {booktitle} {International Colloquium on Automata,
  Languages, and Programming}}}\ (\bibinfo {organization} {Springer},\ \bibinfo
  {year} {2010})\ pp.\ \bibinfo {pages} {527--538}\BibitemShut {NoStop}%
\bibitem [{\citenamefont {Elkouss}\ and\ \citenamefont
  {P{\'e}rez-Garc{\'\i}a}(2018)}]{undecidability}%
  \BibitemOpen
  \bibfield  {author} {\bibinfo {author} {\bibfnamefont {D.}~\bibnamefont
  {Elkouss}}\ and\ \bibinfo {author} {\bibfnamefont {D.}~\bibnamefont
  {P{\'e}rez-Garc{\'\i}a}},\ }\bibfield  {title} {\bibinfo {title} {Memory
  effects can make the transmission capability of a communication channel
  uncomputable},\ }\href@noop {} {\bibfield  {journal} {\bibinfo  {journal}
  {Nature communications}\ }\textbf {\bibinfo {volume} {9}},\ \bibinfo {pages}
  {1} (\bibinfo {year} {2018})}\BibitemShut {NoStop}%
\bibitem [{\citenamefont {Woods}\ \emph {et~al.}(2019)\citenamefont {Woods},
  \citenamefont {Silva},\ and\ \citenamefont {Oppenheim}}]{Woods2019}%
  \BibitemOpen
  \bibfield  {author} {\bibinfo {author} {\bibfnamefont {M.~P.}\ \bibnamefont
  {Woods}}, \bibinfo {author} {\bibfnamefont {R.}~\bibnamefont {Silva}},\ and\
  \bibinfo {author} {\bibfnamefont {J.}~\bibnamefont {Oppenheim}},\ }\bibfield
  {title} {\bibinfo {title} {Autonomous quantum machines and finite-sized
  clocks},\ }\href {https://doi.org/10.1007/s00023-018-0736-9} {\bibfield
  {journal} {\bibinfo  {journal} {Annales Henri Poincar{\'e}}\ }\textbf
  {\bibinfo {volume} {20}},\ \bibinfo {pages} {125} (\bibinfo {year}
  {2019})}\BibitemShut {NoStop}%
\bibitem [{\citenamefont {{Yang}}\ \emph {et~al.}(2019)\citenamefont {{Yang}},
  \citenamefont {{Baumg{\"a}rtner}}, \citenamefont {{Silva}},\ and\
  \citenamefont {{Renner}}}]{Yang2019}%
  \BibitemOpen
  \bibfield  {author} {\bibinfo {author} {\bibfnamefont {Y.}~\bibnamefont
  {{Yang}}}, \bibinfo {author} {\bibfnamefont {L.}~\bibnamefont
  {{Baumg{\"a}rtner}}}, \bibinfo {author} {\bibfnamefont {R.}~\bibnamefont
  {{Silva}}},\ and\ \bibinfo {author} {\bibfnamefont {R.}~\bibnamefont
  {{Renner}}},\ }\bibfield  {title} {\bibinfo {title} {Accuracy enhancing
  protocols for quantum clocks},\ }\href
  {https://doi.org/10.48550/arXiv.1905.09707} {\bibfield  {journal} {\bibinfo
  {journal} {arXiv}\ ,\ \bibinfo {eid} {arXiv:1905.09707}} (\bibinfo {year}
  {2019})},\ \Eprint {https://arxiv.org/abs/1905.09707} {arXiv:1905.09707}
  \BibitemShut {NoStop}%
\bibitem [{\citenamefont {{Yang}}\ and\ \citenamefont
  {{Renner}}(2020)}]{Yang2020}%
  \BibitemOpen
  \bibfield  {author} {\bibinfo {author} {\bibfnamefont {Y.}~\bibnamefont
  {{Yang}}}\ and\ \bibinfo {author} {\bibfnamefont {R.}~\bibnamefont
  {{Renner}}},\ }\bibfield  {title} {\bibinfo {title} {Ultimate limit on time
  signal generation},\ }\href {https://doi.org/10.48550/arXiv.2004.07857}
  {\bibfield  {journal} {\bibinfo  {journal} {arXiv}\ ,\ \bibinfo {eid}
  {arXiv:2004.07857}} (\bibinfo {year} {2020})},\ \Eprint
  {https://arxiv.org/abs/2004.07857} {arXiv:2004.07857} \BibitemShut {NoStop}%
\bibitem [{\citenamefont {Schwarzhans}\ \emph {et~al.}(2021)\citenamefont
  {Schwarzhans}, \citenamefont {Lock}, \citenamefont {Erker}, \citenamefont
  {Friis},\ and\ \citenamefont {Huber}}]{Schwarzhans2021}%
  \BibitemOpen
  \bibfield  {author} {\bibinfo {author} {\bibfnamefont {E.}~\bibnamefont
  {Schwarzhans}}, \bibinfo {author} {\bibfnamefont {M.~P.~E.}\ \bibnamefont
  {Lock}}, \bibinfo {author} {\bibfnamefont {P.}~\bibnamefont {Erker}},
  \bibinfo {author} {\bibfnamefont {N.}~\bibnamefont {Friis}},\ and\ \bibinfo
  {author} {\bibfnamefont {M.}~\bibnamefont {Huber}},\ }\bibfield  {title}
  {\bibinfo {title} {Autonomous temporal probability concentration: Clockworks
  and the second law of thermodynamics},\ }\href
  {https://doi.org/10.1103/PhysRevX.11.011046} {\bibfield  {journal} {\bibinfo
  {journal} {Phys. Rev. X}\ }\textbf {\bibinfo {volume} {11}},\ \bibinfo
  {pages} {011046} (\bibinfo {year} {2021})}\BibitemShut {NoStop}%
\bibitem [{\citenamefont {Meier}\ \emph {et~al.}(2023)\citenamefont {Meier},
  \citenamefont {Schwarzhans}, \citenamefont {Erker},\ and\ \citenamefont
  {Huber}}]{MeierPRL2023}%
  \BibitemOpen
  \bibfield  {author} {\bibinfo {author} {\bibfnamefont {F.}~\bibnamefont
  {Meier}}, \bibinfo {author} {\bibfnamefont {E.}~\bibnamefont {Schwarzhans}},
  \bibinfo {author} {\bibfnamefont {P.}~\bibnamefont {Erker}},\ and\ \bibinfo
  {author} {\bibfnamefont {M.}~\bibnamefont {Huber}},\ }\bibfield  {title}
  {\bibinfo {title} {Fundamental accuracy-resolution trade-off for timekeeping
  devices},\ }\href {https://doi.org/10.1103/PhysRevLett.131.220201} {\bibfield
   {journal} {\bibinfo  {journal} {Phys. Rev. Lett.}\ }\textbf {\bibinfo
  {volume} {131}},\ \bibinfo {pages} {220201} (\bibinfo {year}
  {2023})}\BibitemShut {NoStop}%
\bibitem [{\citenamefont {{Silva}}\ \emph {et~al.}(2023)\citenamefont
  {{Silva}}, \citenamefont {{Nurgalieva}},\ and\ \citenamefont
  {{Wilming}}}]{Silva2023}%
  \BibitemOpen
  \bibfield  {author} {\bibinfo {author} {\bibfnamefont {R.}~\bibnamefont
  {{Silva}}}, \bibinfo {author} {\bibfnamefont {N.}~\bibnamefont
  {{Nurgalieva}}},\ and\ \bibinfo {author} {\bibfnamefont {H.}~\bibnamefont
  {{Wilming}}},\ }\bibfield  {title} {\bibinfo {title} {Ticking clocks in
  quantum theory},\ }\href {https://doi.org/10.48550/arXiv.2306.01829}
  {\bibfield  {journal} {\bibinfo  {journal} {arXiv}\ ,\ \bibinfo {eid}
  {arXiv:2306.01829}} (\bibinfo {year} {2023})},\ \Eprint
  {https://arxiv.org/abs/2306.01829} {arXiv:2306.01829} \BibitemShut {NoStop}%
\bibitem [{\citenamefont {Pearson}\ \emph {et~al.}(2021)\citenamefont
  {Pearson}, \citenamefont {Guryanova}, \citenamefont {Erker}, \citenamefont
  {Laird}, \citenamefont {Briggs}, \citenamefont {Huber},\ and\ \citenamefont
  {Ares}}]{ClockExp2021}%
  \BibitemOpen
  \bibfield  {author} {\bibinfo {author} {\bibfnamefont {A.~N.}\ \bibnamefont
  {Pearson}}, \bibinfo {author} {\bibfnamefont {Y.}~\bibnamefont {Guryanova}},
  \bibinfo {author} {\bibfnamefont {P.}~\bibnamefont {Erker}}, \bibinfo
  {author} {\bibfnamefont {E.~A.}\ \bibnamefont {Laird}}, \bibinfo {author}
  {\bibfnamefont {G.~A.~D.}\ \bibnamefont {Briggs}}, \bibinfo {author}
  {\bibfnamefont {M.}~\bibnamefont {Huber}},\ and\ \bibinfo {author}
  {\bibfnamefont {N.}~\bibnamefont {Ares}},\ }\bibfield  {title} {\bibinfo
  {title} {Measuring the thermodynamic cost of timekeeping},\ }\href
  {https://doi.org/10.1103/PhysRevX.11.021029} {\bibfield  {journal} {\bibinfo
  {journal} {Phys. Rev. X}\ }\textbf {\bibinfo {volume} {11}},\ \bibinfo
  {pages} {021029} (\bibinfo {year} {2021})}\BibitemShut {NoStop}%
\bibitem [{Note1()}]{Note1}%
  \BibitemOpen
  \bibinfo {note} {It doesn't matter which of the initial states is chosen as
  these choices are equivalent up to relabelling.}\BibitemShut {Stop}%
\bibitem [{\citenamefont {Gharibian}(2010)}]{Gharibian}%
  \BibitemOpen
  \bibfield  {author} {\bibinfo {author} {\bibfnamefont {S.}~\bibnamefont
  {Gharibian}},\ }\bibfield  {title} {\bibinfo {title} {Strong np-hardness of
  the quantum separability problem},\ }\href@noop {} {\bibfield  {journal}
  {\bibinfo  {journal} {Quantum Information \& Computation}\ }\textbf {\bibinfo
  {volume} {10}},\ \bibinfo {pages} {343} (\bibinfo {year} {2010})}\BibitemShut
  {NoStop}%
\bibitem [{\citenamefont {McCulloch}(2008)}]{MPO}%
  \BibitemOpen
  \bibfield  {author} {\bibinfo {author} {\bibfnamefont {I.~P.}\ \bibnamefont
  {McCulloch}},\ }\href {https://doi.org/10.48550/ARXIV.0804.2509} {\bibinfo
  {title} {Infinite size density matrix renormalization group, revisited}}
  (\bibinfo {year} {2008})\BibitemShut {NoStop}%
\bibitem [{\citenamefont {Puterman}(2014)}]{markov_book}%
  \BibitemOpen
  \bibfield  {author} {\bibinfo {author} {\bibfnamefont {M.~L.}\ \bibnamefont
  {Puterman}},\ }\href@noop {} {\emph {\bibinfo {title} {Markov decision
  processes: discrete stochastic dynamic programming}}}\ (\bibinfo  {publisher}
  {John Wiley \& Sons},\ \bibinfo {year} {2014})\BibitemShut {NoStop}%
\bibitem [{\citenamefont {Lasserre}\ \emph {et~al.}(2008)\citenamefont
  {Lasserre}, \citenamefont {Henrion}, \citenamefont {Prieur},\ and\
  \citenamefont {Tr\'{e}lat}}]{LasserreOCP2008}%
  \BibitemOpen
  \bibfield  {author} {\bibinfo {author} {\bibfnamefont {J.~B.}\ \bibnamefont
  {Lasserre}}, \bibinfo {author} {\bibfnamefont {D.}~\bibnamefont {Henrion}},
  \bibinfo {author} {\bibfnamefont {C.}~\bibnamefont {Prieur}},\ and\ \bibinfo
  {author} {\bibfnamefont {E.}~\bibnamefont {Tr\'{e}lat}},\ }\bibfield  {title}
  {\bibinfo {title} {Nonlinear optimal control via occupation measures and
  {LMI}-relaxations},\ }\href {https://doi.org/10.1137/070685051} {\bibfield
  {journal} {\bibinfo  {journal} {SIAM Journal on Control and Optimization}\
  }\textbf {\bibinfo {volume} {47}},\ \bibinfo {pages} {1643} (\bibinfo {year}
  {2008})}\BibitemShut {NoStop}%
\bibitem [{\citenamefont {Krivine}(1964)}]{krivine}%
  \BibitemOpen
  \bibfield  {author} {\bibinfo {author} {\bibfnamefont {J.~L.}\ \bibnamefont
  {Krivine}},\ }\bibfield  {title} {\bibinfo {title} {Anneaux préordonnés},\
  }\href@noop {} {\bibfield  {journal} {\bibinfo  {journal} {Journal
  d’Analyse Mathématique}\ }\textbf {\bibinfo {volume} {12}},\ \bibinfo
  {pages} {307} (\bibinfo {year} {1964})}\BibitemShut {NoStop}%
\bibitem [{\citenamefont {Krivine}(1974)}]{stengle}%
  \BibitemOpen
  \bibfield  {author} {\bibinfo {author} {\bibfnamefont {J.~L.}\ \bibnamefont
  {Krivine}},\ }\bibfield  {title} {\bibinfo {title} {A nullstellensatz and a
  positivstellensatz in semialgebraic geometry},\ }\href@noop {} {\bibfield
  {journal} {\bibinfo  {journal} {Math. Ann.}\ }\textbf {\bibinfo {volume}
  {207}},\ \bibinfo {pages} {87} (\bibinfo {year} {1974})}\BibitemShut
  {NoStop}%
\bibitem [{\citenamefont {Schm{\"u}dgen}(1991)}]{schmuedgen}%
  \BibitemOpen
  \bibfield  {author} {\bibinfo {author} {\bibfnamefont {K.}~\bibnamefont
  {Schm{\"u}dgen}},\ }\bibfield  {title} {\bibinfo {title} {The k-moment
  problem for compact semi-algebraic sets},\ }\href@noop {} {\bibfield
  {journal} {\bibinfo  {journal} {Mathematische Annalen}\ }\textbf {\bibinfo
  {volume} {289}},\ \bibinfo {pages} {203} (\bibinfo {year}
  {1991})}\BibitemShut {NoStop}%
\bibitem [{\citenamefont {Putinar}(1993)}]{putinar}%
  \BibitemOpen
  \bibfield  {author} {\bibinfo {author} {\bibfnamefont {M.}~\bibnamefont
  {Putinar}},\ }\bibfield  {title} {\bibinfo {title} {Positive polynomials on
  compact semi-algebraic sets},\ }\href {http://www.jstor.org/stable/24897130}
  {\bibfield  {journal} {\bibinfo  {journal} {Indiana University Mathematics
  Journal}\ }\textbf {\bibinfo {volume} {42}},\ \bibinfo {pages} {969}
  (\bibinfo {year} {1993})}\BibitemShut {NoStop}%
\bibitem [{\citenamefont {Lasserre}(2001)}]{lasserre}%
  \BibitemOpen
  \bibfield  {author} {\bibinfo {author} {\bibfnamefont {J.~B.}\ \bibnamefont
  {Lasserre}},\ }\bibfield  {title} {\bibinfo {title} {Global optimization with
  polynomials and the problem of moments},\ }\href
  {https://doi.org/10.1137/S1052623400366802} {\bibfield  {journal} {\bibinfo
  {journal} {SIAM Journal on Optimization}\ }\textbf {\bibinfo {volume} {11}},\
  \bibinfo {pages} {796} (\bibinfo {year} {2001})}\BibitemShut {NoStop}%
\bibitem [{\citenamefont {Parrilo}(2003)}]{parrilo}%
  \BibitemOpen
  \bibfield  {author} {\bibinfo {author} {\bibfnamefont {P.}~\bibnamefont
  {Parrilo}},\ }\bibfield  {title} {\bibinfo {title} {Semidefinite programming
  relaxations for semialgebraic problems},\ }\href
  {https://doi.org/10.1007/s10107-003-0387-5} {\bibfield  {journal} {\bibinfo
  {journal} {Mathematical Programming, Series B}\ }\textbf {\bibinfo {volume}
  {96}} (\bibinfo {year} {2003})}\BibitemShut {NoStop}%
\bibitem [{\citenamefont {Murray}\ \emph {et~al.}(2021)\citenamefont {Murray},
  \citenamefont {Chandrasekaran},\ and\ \citenamefont {Wierman}}]{signom_pos1}%
  \BibitemOpen
  \bibfield  {author} {\bibinfo {author} {\bibfnamefont {R.}~\bibnamefont
  {Murray}}, \bibinfo {author} {\bibfnamefont {V.}~\bibnamefont
  {Chandrasekaran}},\ and\ \bibinfo {author} {\bibfnamefont {A.}~\bibnamefont
  {Wierman}},\ }\bibfield  {title} {\bibinfo {title} {Newton polytopes and
  relative entropy optimization},\ }\href
  {https://doi.org/10.1007/s10208-021-09497-w} {\bibfield  {journal} {\bibinfo
  {journal} {Foundations of Computational Mathematics}\ }\textbf {\bibinfo
  {volume} {21}},\ \bibinfo {pages} {1703} (\bibinfo {year}
  {2021})}\BibitemShut {NoStop}%
\bibitem [{\citenamefont {Dressler}\ and\ \citenamefont
  {Murray}(2021)}]{signom_pos2}%
  \BibitemOpen
  \bibfield  {author} {\bibinfo {author} {\bibfnamefont {M.}~\bibnamefont
  {Dressler}}\ and\ \bibinfo {author} {\bibfnamefont {R.}~\bibnamefont
  {Murray}},\ }\href {https://doi.org/10.48550/ARXIV.2107.00345} {\bibinfo
  {title} {Algebraic perspectives on signomial optimization}} (\bibinfo {year}
  {2021})\BibitemShut {NoStop}%
\bibitem [{\citenamefont {ApS}(2019)}]{mosek}%
  \BibitemOpen
  \bibfield  {author} {\bibinfo {author} {\bibfnamefont {M.}~\bibnamefont
  {ApS}},\ }\href {http://docs.mosek.com/9.0/toolbox/index.html} {\emph
  {\bibinfo {title} {The MOSEK optimization toolbox for MATLAB manual. Version
  9.0.}}} (\bibinfo {year} {2019})\BibitemShut {NoStop}%
\bibitem [{\citenamefont {Diamond}\ and\ \citenamefont {Boyd}(2016)}]{cvxpy}%
  \BibitemOpen
  \bibfield  {author} {\bibinfo {author} {\bibfnamefont {S.}~\bibnamefont
  {Diamond}}\ and\ \bibinfo {author} {\bibfnamefont {S.}~\bibnamefont {Boyd}},\
  }\bibfield  {title} {\bibinfo {title} {{CVXPY}: {A} {P}ython-embedded
  modeling language for convex optimization},\ }\href@noop {} {\bibfield
  {journal} {\bibinfo  {journal} {Journal of Machine Learning Research}\
  }\textbf {\bibinfo {volume} {17}},\ \bibinfo {pages} {1} (\bibinfo {year}
  {2016})}\BibitemShut {NoStop}%
\bibitem [{\citenamefont {O'Donoghue}\ \emph {et~al.}(2022)\citenamefont
  {O'Donoghue}, \citenamefont {Chu}, \citenamefont {Parikh},\ and\
  \citenamefont {Boyd}}]{scs}%
  \BibitemOpen
  \bibfield  {author} {\bibinfo {author} {\bibfnamefont {B.}~\bibnamefont
  {O'Donoghue}}, \bibinfo {author} {\bibfnamefont {E.}~\bibnamefont {Chu}},
  \bibinfo {author} {\bibfnamefont {N.}~\bibnamefont {Parikh}},\ and\ \bibinfo
  {author} {\bibfnamefont {S.}~\bibnamefont {Boyd}},\ }\href@noop {} {\bibinfo
  {title} {{SCS}: Splitting conic solver, version 3.2.2}},\ \bibinfo
  {howpublished} {\url{https://github.com/cvxgrp/scs}} (\bibinfo {year}
  {2022})\BibitemShut {NoStop}%
\bibitem [{Note2()}]{Note2}%
  \BibitemOpen
  \bibinfo {note} {Recall that the variables of the polynomial optimisation are
  the coefficients of the $V_k$ rather than the variables of the sequential
  model $s_k$.}\BibitemShut {Stop}%
\bibitem [{\citenamefont {L\"{o}fberg}(2004)}]{yalmip}%
  \BibitemOpen
  \bibfield  {author} {\bibinfo {author} {\bibfnamefont {J.}~\bibnamefont
  {L\"{o}fberg}},\ }\bibfield  {title} {\bibinfo {title} {Yalmip : A toolbox
  for modeling and optimization in matlab},\ }in\ \href@noop {} {\emph
  {\bibinfo {booktitle} {Proceedings of the CACSD Conference}}}\ (\bibinfo
  {address} {Taipei, Taiwan},\ \bibinfo {year} {2004})\BibitemShut {NoStop}%
\bibitem [{Note3()}]{Note3}%
  \BibitemOpen
  \bibinfo {note} {Note that solvemoment sometimes allows optimal solutions to
  the problem to be directly extracted, as is the case in the examples below.
  This is not surprising, as in cases where the Lasserre hierarchy has
  converged at a given finite level, the optimal values of the polynomial
  variables can be found in the moment matrix.}\BibitemShut {Stop}%
\bibitem [{Note4()}]{Note4}%
  \BibitemOpen
  \bibinfo {note} {Deriving this upper bound analytically, is straightforward
  if we require in addition that $P(b=1|0)=0$, which we can impose without
  changing our upper bounds. In this case the polynomials simplifies in a way
  that inspecting the gradient immediately leads to these bounds.}\BibitemShut
  {Stop}%
\bibitem [{Note5()}]{Note5}%
  \BibitemOpen
  \bibinfo {note} {Notice that we have changed the notation for value functions
  from $\{V_k\}_k$ to $\{W_k\}_k$ here, to make the distinction from the $L=3$
  case clear.}\BibitemShut {Stop}%
\bibitem [{Note6()}]{Note6}%
  \BibitemOpen
  \bibinfo {note} {That $W_2(s_2, \lambda )$ upper bounds the reward follows
  like in the $L=3$ case above.}\BibitemShut {Stop}%
\bibitem [{Note7()}]{Note7}%
  \BibitemOpen
  \bibinfo {note} {Analytically, checking this is again more straightforward
  when assuming additionally that $P(b=1|0)=0$.}\BibitemShut {Stop}%
\bibitem [{\citenamefont {Wei}\ \emph {et~al.}(2020)\citenamefont {Wei},
  \citenamefont {Lauer}, \citenamefont {Srinivasan}, \citenamefont
  {Sundaresan}, \citenamefont {McClure}, \citenamefont {Toyli}, \citenamefont
  {McKay}, \citenamefont {Gambetta},\ and\ \citenamefont {Sheldon}}]{GHZ18}%
  \BibitemOpen
  \bibfield  {author} {\bibinfo {author} {\bibfnamefont {K.~X.}\ \bibnamefont
  {Wei}}, \bibinfo {author} {\bibfnamefont {I.}~\bibnamefont {Lauer}}, \bibinfo
  {author} {\bibfnamefont {S.}~\bibnamefont {Srinivasan}}, \bibinfo {author}
  {\bibfnamefont {N.}~\bibnamefont {Sundaresan}}, \bibinfo {author}
  {\bibfnamefont {D.~T.}\ \bibnamefont {McClure}}, \bibinfo {author}
  {\bibfnamefont {D.}~\bibnamefont {Toyli}}, \bibinfo {author} {\bibfnamefont
  {D.~C.}\ \bibnamefont {McKay}}, \bibinfo {author} {\bibfnamefont {J.~M.}\
  \bibnamefont {Gambetta}},\ and\ \bibinfo {author} {\bibfnamefont
  {S.}~\bibnamefont {Sheldon}},\ }\bibfield  {title} {\bibinfo {title}
  {Verifying multipartite entangled greenberger-horne-zeilinger states via
  multiple quantum coherences},\ }\href
  {https://doi.org/10.1103/PhysRevA.101.032343} {\bibfield  {journal} {\bibinfo
   {journal} {Phys. Rev. A}\ }\textbf {\bibinfo {volume} {101}},\ \bibinfo
  {pages} {032343} (\bibinfo {year} {2020})}\BibitemShut {NoStop}%
\bibitem [{\citenamefont {Song}\ \emph {et~al.}(2019)\citenamefont {Song},
  \citenamefont {Xu}, \citenamefont {Li}, \citenamefont {Zhang}, \citenamefont
  {Zhang}, \citenamefont {Liu}, \citenamefont {Guo}, \citenamefont {Wang},
  \citenamefont {Ren}, \citenamefont {Hao}, \citenamefont {Feng}, \citenamefont
  {Fan}, \citenamefont {Zheng}, \citenamefont {Wang}, \citenamefont {Wang},\
  and\ \citenamefont {Zhu}}]{GHZ20cat}%
  \BibitemOpen
  \bibfield  {author} {\bibinfo {author} {\bibfnamefont {C.}~\bibnamefont
  {Song}}, \bibinfo {author} {\bibfnamefont {K.}~\bibnamefont {Xu}}, \bibinfo
  {author} {\bibfnamefont {H.}~\bibnamefont {Li}}, \bibinfo {author}
  {\bibfnamefont {Y.-R.}\ \bibnamefont {Zhang}}, \bibinfo {author}
  {\bibfnamefont {X.}~\bibnamefont {Zhang}}, \bibinfo {author} {\bibfnamefont
  {W.}~\bibnamefont {Liu}}, \bibinfo {author} {\bibfnamefont {Q.}~\bibnamefont
  {Guo}}, \bibinfo {author} {\bibfnamefont {Z.}~\bibnamefont {Wang}}, \bibinfo
  {author} {\bibfnamefont {W.}~\bibnamefont {Ren}}, \bibinfo {author}
  {\bibfnamefont {J.}~\bibnamefont {Hao}}, \bibinfo {author} {\bibfnamefont
  {H.}~\bibnamefont {Feng}}, \bibinfo {author} {\bibfnamefont {H.}~\bibnamefont
  {Fan}}, \bibinfo {author} {\bibfnamefont {D.}~\bibnamefont {Zheng}}, \bibinfo
  {author} {\bibfnamefont {D.-W.}\ \bibnamefont {Wang}}, \bibinfo {author}
  {\bibfnamefont {H.}~\bibnamefont {Wang}},\ and\ \bibinfo {author}
  {\bibfnamefont {S.-Y.}\ \bibnamefont {Zhu}},\ }\bibfield  {title} {\bibinfo
  {title} {Generation of multicomponent atomic schrödinger cat states of up to
  20 qubits},\ }\href {https://doi.org/10.1126/science.aay0600} {\bibfield
  {journal} {\bibinfo  {journal} {Science}\ }\textbf {\bibinfo {volume}
  {365}},\ \bibinfo {pages} {574} (\bibinfo {year} {2019})}\BibitemShut
  {NoStop}%
\bibitem [{\citenamefont {Omran}\ \emph {et~al.}(2019)\citenamefont {Omran},
  \citenamefont {Levine}, \citenamefont {Keesling}, \citenamefont {Semeghini},
  \citenamefont {Wang}, \citenamefont {Ebadi}, \citenamefont {Bernien},
  \citenamefont {Zibrov}, \citenamefont {Pichler}, \citenamefont {Choi},
  \citenamefont {Cui}, \citenamefont {Rossignolo}, \citenamefont {Rembold},
  \citenamefont {Montangero}, \citenamefont {Calarco}, \citenamefont {Endres},
  \citenamefont {Greiner}, \citenamefont {Vuletić},\ and\ \citenamefont
  {Lukin}}]{GHZ20catRydberg}%
  \BibitemOpen
  \bibfield  {author} {\bibinfo {author} {\bibfnamefont {A.}~\bibnamefont
  {Omran}}, \bibinfo {author} {\bibfnamefont {H.}~\bibnamefont {Levine}},
  \bibinfo {author} {\bibfnamefont {A.}~\bibnamefont {Keesling}}, \bibinfo
  {author} {\bibfnamefont {G.}~\bibnamefont {Semeghini}}, \bibinfo {author}
  {\bibfnamefont {T.~T.}\ \bibnamefont {Wang}}, \bibinfo {author}
  {\bibfnamefont {S.}~\bibnamefont {Ebadi}}, \bibinfo {author} {\bibfnamefont
  {H.}~\bibnamefont {Bernien}}, \bibinfo {author} {\bibfnamefont {A.~S.}\
  \bibnamefont {Zibrov}}, \bibinfo {author} {\bibfnamefont {H.}~\bibnamefont
  {Pichler}}, \bibinfo {author} {\bibfnamefont {S.}~\bibnamefont {Choi}},
  \bibinfo {author} {\bibfnamefont {J.}~\bibnamefont {Cui}}, \bibinfo {author}
  {\bibfnamefont {M.}~\bibnamefont {Rossignolo}}, \bibinfo {author}
  {\bibfnamefont {P.}~\bibnamefont {Rembold}}, \bibinfo {author} {\bibfnamefont
  {S.}~\bibnamefont {Montangero}}, \bibinfo {author} {\bibfnamefont
  {T.}~\bibnamefont {Calarco}}, \bibinfo {author} {\bibfnamefont
  {M.}~\bibnamefont {Endres}}, \bibinfo {author} {\bibfnamefont
  {M.}~\bibnamefont {Greiner}}, \bibinfo {author} {\bibfnamefont
  {V.}~\bibnamefont {Vuletić}},\ and\ \bibinfo {author} {\bibfnamefont
  {M.~D.}\ \bibnamefont {Lukin}},\ }\bibfield  {title} {\bibinfo {title}
  {Generation and manipulation of {S}chrödinger cat states in {R}ydberg atom
  arrays},\ }\href {https://doi.org/10.1126/science.aax9743} {\bibfield
  {journal} {\bibinfo  {journal} {Science}\ }\textbf {\bibinfo {volume}
  {365}},\ \bibinfo {pages} {570} (\bibinfo {year} {2019})}\BibitemShut
  {NoStop}%
\bibitem [{\citenamefont {Huang}\ \emph {et~al.}(2020)\citenamefont {Huang},
  \citenamefont {Kueng},\ and\ \citenamefont {Preskill}}]{kueng_shadow}%
  \BibitemOpen
  \bibfield  {author} {\bibinfo {author} {\bibfnamefont {H.-Y.}\ \bibnamefont
  {Huang}}, \bibinfo {author} {\bibfnamefont {R.}~\bibnamefont {Kueng}},\ and\
  \bibinfo {author} {\bibfnamefont {J.}~\bibnamefont {Preskill}},\ }\bibfield
  {title} {\bibinfo {title} {Predicting many properties of a quantum system
  from very few measurements},\ }\href@noop {} {\bibfield  {journal} {\bibinfo
  {journal} {Nature Physics}\ }\textbf {\bibinfo {volume} {16}},\ \bibinfo
  {pages} {1050} (\bibinfo {year} {2020})}\BibitemShut {NoStop}%
\bibitem [{\citenamefont {Gühne}\ and\ \citenamefont
  {Tóth}(2009)}]{review_detection}%
  \BibitemOpen
  \bibfield  {author} {\bibinfo {author} {\bibfnamefont {O.}~\bibnamefont
  {Gühne}}\ and\ \bibinfo {author} {\bibfnamefont {G.}~\bibnamefont {Tóth}},\
  }\bibfield  {title} {\bibinfo {title} {Entanglement detection},\ }\href
  {https://doi.org/https://doi.org/10.1016/j.physrep.2009.02.004} {\bibfield
  {journal} {\bibinfo  {journal} {Physics Reports}\ }\textbf {\bibinfo {volume}
  {474}},\ \bibinfo {pages} {1} (\bibinfo {year} {2009})}\BibitemShut {NoStop}%
\bibitem [{\citenamefont {Mermin}(1990)}]{Mermin1990}%
  \BibitemOpen
  \bibfield  {author} {\bibinfo {author} {\bibfnamefont {N.~D.}\ \bibnamefont
  {Mermin}},\ }\bibfield  {title} {\bibinfo {title} {Extreme quantum
  entanglement in a superposition of macroscopically distinct states},\ }\href
  {https://doi.org/10.1103/PhysRevLett.65.1838} {\bibfield  {journal} {\bibinfo
   {journal} {Phys. Rev. Lett.}\ }\textbf {\bibinfo {volume} {65}},\ \bibinfo
  {pages} {1838} (\bibinfo {year} {1990})}\BibitemShut {NoStop}%
\bibitem [{\citenamefont {Boyer}(2004)}]{Boyer2004}%
  \BibitemOpen
  \bibfield  {author} {\bibinfo {author} {\bibfnamefont {M.}~\bibnamefont
  {Boyer}},\ }\href@noop {} {\bibinfo {title} {Extended ghz n-player games with
  classical probability of winning tending to 0}} (\bibinfo {year} {2004}),\
  \Eprint {https://arxiv.org/abs/quant-ph/0408090} {arXiv:quant-ph/0408090
  [quant-ph]} \BibitemShut {NoStop}%
\bibitem [{Note8()}]{Note8}%
  \BibitemOpen
  \bibinfo {note} {Given the small numbers of experimental rounds, we used the
  mean instead of the `median of means' recommended in \cite {kueng_shadow} to
  estimate the GHZ fidelity in the shadow protocols.}\BibitemShut {Stop}%
\bibitem [{\citenamefont {Flammia}\ and\ \citenamefont
  {Liu}(2011)}]{FlammiaPRL2011}%
  \BibitemOpen
  \bibfield  {author} {\bibinfo {author} {\bibfnamefont {S.~T.}\ \bibnamefont
  {Flammia}}\ and\ \bibinfo {author} {\bibfnamefont {Y.-K.}\ \bibnamefont
  {Liu}},\ }\bibfield  {title} {\bibinfo {title} {Direct fidelity estimation
  from few pauli measurements},\ }\href
  {https://doi.org/10.1103/PhysRevLett.106.230501} {\bibfield  {journal}
  {\bibinfo  {journal} {Phys. Rev. Lett.}\ }\textbf {\bibinfo {volume} {106}},\
  \bibinfo {pages} {230501} (\bibinfo {year} {2011})}\BibitemShut {NoStop}%
\bibitem [{Note9()}]{Note9}%
  \BibitemOpen
  \bibinfo {note} {The statement in~\cite {undecidability} is more general in
  the sense that the lemma is proven when additionally requiring the automata
  to be resettable. This property is, however, not required for our application
  and we thus omit it from the statement below.}\BibitemShut {Stop}%
\bibitem [{Note10()}]{Note10}%
  \BibitemOpen
  \bibinfo {note} {Notice that this restriction is not necessary. In
  particular, analogous methods can be used to deal with scenarios where the
  $f_k, r_k$ are $k$-dependent, as long as we can formulate analogous
  constraints in terms of polynomials of the corresponding bounded variable
  $\alpha $ below.}\BibitemShut {Stop}%
\bibitem [{Note11()}]{Note11}%
  \BibitemOpen
  \bibinfo {note} {It is advisable, on the grounds of numerical stability, to
  also restrict the ranges of possible values for $A, B,C,...$.}\BibitemShut
  {Stop}%
\bibitem [{Note12()}]{Note12}%
  \BibitemOpen
  \bibinfo {note} {To further support the conjecture that the value reported by
  Mosek is close to the actual solution of the SDP, we considered the value
  functions $\{W_{-j}\}_{j=0}^3, W^1, W^2$ output by the solver and then used
  the Lasserre hierarchy to verify that relations (\ref
  {asymptotic_cons_cycle}) held. We only found tiny violations of positivity
  (of order $-10^{-5}$) for the last two recursion relations.}\BibitemShut
  {Stop}%
\bibitem [{\citenamefont {Kingma}\ and\ \citenamefont {Ba}(2017)}]{adam}%
  \BibitemOpen
  \bibfield  {author} {\bibinfo {author} {\bibfnamefont {D.~P.}\ \bibnamefont
  {Kingma}}\ and\ \bibinfo {author} {\bibfnamefont {J.}~\bibnamefont {Ba}},\
  }\href@noop {} {\bibinfo {title} {Adam: A method for stochastic
  optimization}} (\bibinfo {year} {2017}),\ \Eprint
  {https://arxiv.org/abs/1412.6980} {arXiv:1412.6980 [cs.LG]} \BibitemShut
  {NoStop}%
\bibitem [{\citenamefont {Bravyi}\ and\ \citenamefont {Kitaev}(2005)}]{magic1}%
  \BibitemOpen
  \bibfield  {author} {\bibinfo {author} {\bibfnamefont {S.}~\bibnamefont
  {Bravyi}}\ and\ \bibinfo {author} {\bibfnamefont {A.}~\bibnamefont
  {Kitaev}},\ }\bibfield  {title} {\bibinfo {title} {Universal quantum
  computation with ideal clifford gates and noisy ancillas},\ }\href
  {https://doi.org/10.1103/PhysRevA.71.022316} {\bibfield  {journal} {\bibinfo
  {journal} {Phys. Rev. A}\ }\textbf {\bibinfo {volume} {71}},\ \bibinfo
  {pages} {022316} (\bibinfo {year} {2005})}\BibitemShut {NoStop}%
\bibitem [{\citenamefont {Veitch}\ \emph {et~al.}(2014)\citenamefont {Veitch},
  \citenamefont {Mousavian}, \citenamefont {Gottesman},\ and\ \citenamefont
  {Emerson}}]{magic2}%
  \BibitemOpen
  \bibfield  {author} {\bibinfo {author} {\bibfnamefont {V.}~\bibnamefont
  {Veitch}}, \bibinfo {author} {\bibfnamefont {S.~A.~H.}\ \bibnamefont
  {Mousavian}}, \bibinfo {author} {\bibfnamefont {D.}~\bibnamefont
  {Gottesman}},\ and\ \bibinfo {author} {\bibfnamefont {J.}~\bibnamefont
  {Emerson}},\ }\bibfield  {title} {\bibinfo {title} {The resource theory of
  stabilizer quantum computation},\ }\href
  {https://doi.org/10.1088/1367-2630/16/1/013009} {\bibfield  {journal}
  {\bibinfo  {journal} {New Journal of Physics}\ }\textbf {\bibinfo {volume}
  {16}},\ \bibinfo {pages} {013009} (\bibinfo {year} {2014})}\BibitemShut
  {NoStop}%
\bibitem [{Note13()}]{Note13}%
  \BibitemOpen
  \bibinfo {note} {The dual problem of (\ref {dual}) does not seem to have a
  unique solution, because we find that the use of eq. (\ref
  {perturbed_single_q}) results in infeasible semidefinite
  programs.}\BibitemShut {Stop}%
\bibitem [{\citenamefont {Aaronson}(2018)}]{aaronson_shadow}%
  \BibitemOpen
  \bibfield  {author} {\bibinfo {author} {\bibfnamefont {S.}~\bibnamefont
  {Aaronson}},\ }\bibfield  {title} {\bibinfo {title} {Shadow tomography of
  quantum states},\ }in\ \href@noop {} {\emph {\bibinfo {booktitle}
  {Proceedings of the 50th annual ACM SIGACT symposium on theory of
  computing}}}\ (\bibinfo {year} {2018})\ pp.\ \bibinfo {pages}
  {325--338}\BibitemShut {NoStop}%
\bibitem [{Note14()}]{Note14}%
  \BibitemOpen
  \bibinfo {note} {More precisely, a sufficient condition for the existence of
  an SOS decomposition for a positive polynomial $p(z)$ is that, for some
  $K>0$, the polynomial $K-\DOTSB \sum@ \slimits@ _iz_i^2$ admits a SOS
  decomposition. The latter, in turn, implies that $Z$ is bounded.}\BibitemShut
  {Stop}%
\bibitem [{Note15()}]{Note15}%
  \BibitemOpen
  \bibinfo {note} {This adaptation of the more common Lasserre-Parrilo
  hierarchy relies on the Schmuedgen Positivstellensatz \cite {schmuedgen}
  instead and is employed here to obtain a better numerical performance in the
  problems considered later.}\BibitemShut {Stop}%
\bibitem [{Note16()}]{Note16}%
  \BibitemOpen
  \bibinfo {note} {For convenience, we choose the first and third entry here
  since we perform Pauli measurements and thus this way the equation of motion
  can be formulated more directly.}\BibitemShut {Stop}%
\end{thebibliography}%

\end{document}